\documentclass{lmcs}
\pdfoutput=1

\usepackage{lastpage}
\lmcsdoi{15}{3}{19}
\lmcsheading{}{\pageref{LastPage}}{}{}%
{Feb.~26,~2018}{Aug.~21,~2019}{}


\usepackage{booktabs}
\usepackage{amssymb}
\usepackage{microtype}
\usepackage{mathtools}
\usepackage{cite}
\usepackage{xspace}
\usepackage{latexsym}
\usepackage{tikz}
\usetikzlibrary{arrows}
\usetikzlibrary{decorations.markings}
\usetikzlibrary{positioning}
\tikzset{
  big arrow/.style={
    decoration={markings, mark=at position 1 with {\arrow[scale=1.5]{>}}},
    postaction={decorate},},
}
\usepackage{enumitem}
\usepackage{doi}
\usepackage{fontawesome}
\usepackage{hyperref}
\definecolor{linkblue}{RGB}{0,0,180}

\setlist[itemize]{leftmargin=*}
\setlist[enumerate]{leftmargin=*}

\newcommand\toolname[1]{\mbox{\tsfs{#1}}}
\newcommand\tool[1]{\toolname{#1}\xspace}
\newcommand\ceta{\tool{C\kern-0.2exe\kern-0.5exT\kern-0.5exA}}
\newcommand\kbcv{\tool{KBCV}}
\newcommand\mkbtt{\tool{mkb\kern-0.2exTT}}
\newcommand\cime{\tool{C\kern-0.1ex\textsl{i}ME}}
\newcommand\slothrop{\tool{Slothrop}}
\newcommand{\isafor}{\tsfs{Isa\kern-0.2exF\kern-0.2exo\kern-0.2exR}%
\xspace}

\newcommand{\theory}[2]{\nolinkurl{#2}}
\newcommand{\theorydir}[1]{\nolinkurl{#1}}

\newcommand{\m}[1]{\mathsf{#1}}
\newcommand{\tsfs}[1]{\textsf{\small #1}}
\newcommand{\mc}[1]{\mathcal{#1}}
\newcommand{\e}[1]{\m{#1}\mathstrut}

\renewcommand{\AA}{\mc{A}}
\newcommand{\BB}{\mc{B}}
\newcommand{\EE}{\mc{E}}
\newcommand{\FF}{\mc{F}}
\newcommand{\MM}{\mc{M}}

\newcommand{\PP}{\mc{P}}
\newcommand{\RR}{\mc{R}}
\renewcommand{\SS}{\mc{S}}
\newcommand{\TT}{\mc{T}}
\newcommand{\VV}{\mc{V}}
\newcommand{\KK}{\mc{K}}
\newcommand{\QQ}{\mc{Q}}
\newcommand{\NF}{\m{NF}}
\newcommand{\Pos}{\mc{P}\m{os}}
\newcommand{\PosF}{\Pos_\FF}
\newcommand{\PosV}{\Pos_\VV}
\newcommand{\Var}{\mc{V}\m{ar}}
\newcommand{\ARS}{\langle{} A, {\to} \rangle}
\newcommand{\ARSn}[3][\alpha]{\ensuremath{\langle #2, {\{ {\xr[#1]{}} \}}_{#1 \in #3} \rangle}}

\newcommand{\permute}{\cdot}
\newcommand{\from}{\mathrel{\leftarrow}}
\newcommand{\fromto}{\xlr{}}
\newcommand{\fromB}[1]{\mathrel{_{#1}{\leftarrow}}}
\renewcommand{\conv}{\fromto^*}
\newcommand{\join}{\mathrel{\downarrow}}
\newcommand{\FromB}[1]{\mathrel{{\vphantom{\to}_{#1}}{\from}}}
\newcommand{\FromT}[1]{\mathrel{{\vphantom{\to}^{#1}}{\from}}}
\newcommand{\xfrom}[2]{\mathrel{{\vphantom{\xleftarrow{#2}}_{#1}}%
{\xleftarrow{#2}}}}
\def\test#1#2#3{\setbox0=\hbox{$\vphantom{#1}^{#2}_{#3}$}%
                \dimen0=\wd0%
                \setbox1=\hbox{$\scriptstyle #2$}%
                \advance\dimen0-\wd1%
                \setbox1=\hbox{\hskip\dimen0\copy1}%
                \dimen0=\wd0%
                \setbox2=\hbox{$\scriptstyle #3$}%
                \advance\dimen0-\wd2%
                \setbox2=\hbox{\hskip\dimen0\copy2}%
                {\vphantom{#1}^{\box1}_{\box2}}{#1}
}
\newcommand{\From}[2]{\mathrel{\test{\from}{#2}{#1}}}
\newcommand{\FRom}[3]{\mathrel{\test{\xleftarrow{#3}}{#2}{#1}}}

\newcommand{\scup}{\hspace{.3mm}\cup\hspace{.3mm}}
\newcommand{\KBC}[1]{\tsfs{KB}$_\m{#1}$\xspace}
\newcommand{\KBf}{\KBC{f}}
\newcommand{\KBg}{\KBC{g}}
\newcommand{\KBi}{\KBC{i}}
\newcommand{\KBl}{\KBC{l}}
\newcommand{\KBo}{\KBC{o}}
\newcommand{\vd}[1]{\vdash_\m{#1}}
\newcommand{\vdf}{\vd{f}}
\newcommand{\vdg}{\vd{g}}
\newcommand{\vdi}{\vd{i}}
\newcommand{\vdl}{\vd{l}}
\newcommand{\vdo}{\vd{o}}
\newcommand{\vdfr}[1]{\vdf^\tsfs{\tiny #1}}
\newcommand{\vdir}[1]{\vdi^\tsfs{\tiny #1}}
\newcommand{\vdor}[1]{\vdo^\tsfs{\tiny #1}}
\newcommand{\vdlr}[1]{\vdl^\tsfs{\tiny #1}}

\newcommand{\kbo}{{\m{kbo}}}
\newcommand{\lex}{{\m{lex}}}
\newcommand{\lpo}{{\m{lpo}}}
\newcommand{\mul}{{\m{mul}}}

\newcommand{\KBO}[1][>]{#1_\kbo}
\newcommand{\LPO}[1][>]{#1_\lpo}
\newcommand{\MUL}[1][>]{#1_\mul}

\newcommand{\peak}[5]{#1 \xleftarrow{#2} #3 \xrightarrow{#4} #5}
\newcommand{\cpeak}[4]{#1 \xleftarrow{#2} #3 \xrightarrow{\epsilon} #4}
\newcommand{\tds}{\triangledown_{\hspace{-.5mm}s}}

\newcommand{\CP}{\m{CP}}
\newcommand{\PCP}{\m{PCP}}
\newcommand{\LCP}{\m{LCP}}
\newcommand{\EEn}{\EE_n}
\newcommand{\EEw}{\EE_\omega}
\newcommand{\EEi}{\EE_\infty}
\newcommand{\RRn}{{\RR_n}}
\newcommand{\RRw}{{\RR_\omega}}
\newcommand{\RRi}{{\RR_\infty}}
\newcommand{\SSi}{{\SS_\infty}}
\newcommand{\SSw}{{\SS_\omega}}
\newcommand{\ER}{\ensuremath{\EE \scup \RR}}

\newcommand{\REwgt}{\ensuremath{\EEw^> \scup \RRw}}
\newcommand{\ERw}{\ensuremath{\EEw \scup \RRw}}
\newcommand{\REiw}{\ensuremath{\EEw \scup \RRi}}
\newcommand{\ERi}{\ensuremath{\EEi \scup \RRi}}
\newcommand{\encompasses}{\mathrel{\makebox[0pt]{\makebox[9pt][r]%
{\raise 0.9pt \hbox{$\cdot$}}}{\unrhd}}} 
\newcommand{\rencompasses}{\mathrel{\reflectbox{$\encompasses$}}}
\newcommand{\smprencompasses}{\mathrel{\makebox[0pt]{\makebox[8pt][r]%
{\lower 0.7pt \hbox{$\cdot$}}}{\rhd}}} 
\newcommand{\prencompasses}{\mathrel{\ooalign{\hss\ensuremath%
{\rhd}\hss\cr\kern0.4ex\raise0.0ex\hbox{\scalebox{1.0}%
{$\cdot$}}}}}

\newcommand{\dDg}[1]{\smash{\raisebox{-.5mm}{\rotatebox{#1}{%
$\scriptstyle \eqslantless$}}}}
\newcommand{\ddg}{\dDg{90}}
\newlength{\dotheight}
\newcommand{\sdot}[1]{\settoheight{\dotheight}{$scriptstyle #1$}%
\vphantom{\rule{1pt}{\dotheight}}\smash{\dot{#1}}}

\newcommand{\pe}{\smprencompasses}

\newcommand{\seq}[2][n]{{#2_1},\dots,{#2_{#1}}}
\newcommand{\ito}{\xrightarrow{\smash{\m{i}}}}

\newcommand{\xlr}[2][]{\xleftrightarrow[#1]{#2}}
\newcommand{\xlrs}[2][]{\smash{\xlr[\smash{\raisebox{2pt}{$\scriptstyle%
#1$}}]{#2}}}
\newcommand{\xrs}[2][]{\mathrel{\smash{\xrightarrow[#1]%
{\smash{\raisebox{-1pt}{$\scriptstyle #2$}}}}}}
\newcommand{\CCR}[3]{\mathrel{\smash{%
\xrightarrow[#1]{#2}}^{#3}\vphantom{\xrightarrow[#1]{#2}}}}
\newcommand{\CCL}[3]{\mathrel{\smash{%
{}^{#3}{\xleftarrow[#1]{#2}}}\vphantom{\xleftarrow[#1]{#2}}}}
\newcommand{\CCC}[3]{\mathrel{\smash{%
\xleftrightarrow[#1]{#2}}^{#3}\vphantom{\xleftrightarrow[#1]{#2}}}}

\newcommand{\Xlr}[3]{\xleftrightarrow[\smash{\raisebox{#3mm}%
{\smash{$\scriptstyle #1$}}}]{#2}}

\makeatletter

\newcommand{\r@rrow}[3]{%
  \newcommand{#1}[2][]{%
    \def\next{#2\@ifempty{##1}{}{_{##1}}\@ifempty{##2}{}{^{##2}}}%
    \mathchoice{#3[##1]{##2}}{\next}{\next}{\next}%
  }%
}
\newcommand{\l@rrow}[3]{%
  \newcommand{#1}[2][]{%
    \def\next####1{%
      \setbox0=\hbox{$####1\vphantom{#2}\@ifempty{##1}{}{_{\vphantom{##1}}}%
      \@ifempty{##2}{}{^{##2}}$}%
      \setbox1=\hbox{$####1\vphantom{#2}\@ifempty{##1}{}{_{##1}}%
      \@ifempty{##2}{}{^{\vphantom{##2}}}$}%
      \setbox2=\vbox{\hbox to\wd0{}\hbox to\wd1{}}
      \mathrel{\hskip\wd2\hskip-\wd0\box0\hskip-\wd1\box1{#2}}%
    }%
    \mathchoice{#3[##1]{##2}}{\next\textstyle}%
    {\next\scriptstyle}{\next\scriptscriptstyle}%
  }%
}

\l@rrow{\xl}{\leftarrow}{\xleftarrow}
\r@rrow{\xr}{\rightarrow}{\xrightarrow}
\r@rrow{\sxlr}{\leftrightarrow}{\xleftrightarrow}
\l@rrow{\xphl}{\phleftarrow}{\xphleftarrow}
\r@rrow{\xphr}{\phrightarrow}{\xphrightarrow}

\makeatother

\newcommand\isaforversion{2.37}
\newcommand\isabelleversion{Isabelle2019}
\newcommand\afpbase{https://www.isa-afp.org/browser_info/\isabelleversion/AFP}
\newcommand\afplink[2][Relative_Rewriting]{{\normalfont%
  \urlex{\afpbase/Abstract-Rewriting/#1.html\##2}}}
\newcommand\isaforbase{http://cl-informatik.uibk.ac.at/isafor/v\isaforversion}
\newcommand\isaforlink[2]{{\normalfont%
  \urlex{\isaforbase/LMCS2019/#1.html\##2}}}
\newcommand\extlinkicon{\faicon{check-square-o}}
\newcommand\extlink[1]{%
  \href{#1}{\textcolor{linkblue}{\extlinkicon}}}
\newcommand\urlex[1]{\extlink{#1}}

\newcommand{\lemref}[2][]{Lemma~\ifthenelse{\equal{#1}{}}%
  {\ref{lem:#2}}{\ref{lem:#2}(\ref{lem:#2:#1})}} 
\newcommand{\defref}[1]{Definition~\ref{def:#1}}
\newcommand{\corref}[1]{Corollary~\ref{cor:#1}}
\newcommand{\secref}[1]{Section~\ref{sec:#1}}
\newcommand{\ssecref}[1]{Subsection~\ref{ssec:#1}}
\renewcommand{\ssecref}[1]{Section~\ref{ssec:#1}}
\newcommand{\thmref}[1]{Theorem~\ref{thm:#1}}
\newcommand{\exaref}[1]{Example~\ref{exa:#1}}

\newcommand\ie{that is}
\newcommand\etal{et al.}

\begin{document}

\title{Abstract Completion, Formalized}
\titlecomment{The research described in this paper is supported
by JSPS KAKENHI Grant Number 17K00011, JSPS Core to Core Program and FWF
(Austrian Science Fund) projects P27528, T789 and P27502.
A preliminary version of this paper appeared in
the proceedings of the 2nd International Conference on Formal Structures
for Computation and Deduction~\cite{HMSW17}
and we also incorporated parts of our preceding work~\cite{HMS14}.}

\author[N.~Hirokawa]{Nao Hirokawa\rsuper{a}}
\address{\lsuper{a}School of Information Science, JAIST, Nomi, Japan}
\email{hirokawa@jaist.ac.jp}
\author[A.~Middeldorp]{Aart Middeldorp\rsuper{b}}
\address{\lsuper{b}Department of Computer Science, University of Innsbruck, Innsbruck, Austria}
\email{\{aart.middeldorp,christian.sternagel,sarah.winkler\}@uibk.ac.at}
\author[C.~Sternagel]{Christian Sternagel\rsuper{b}}
\author[S.~Winkler]{Sarah Winkler\rsuper{b}}

\begin{abstract}
\noindent
Completion is one of the most studied techniques in term rewriting and
fundamental to automated reasoning with equalities. In this paper we
present new correctness proofs of abstract completion,
both for finite and infinite runs.
For the special case of ground completion we present a new proof
based on random descent. We moreover extend the results to ordered
completion, an important extension of completion that aims to produce
ground-complete presentations of the initial equations. We present new
proofs concerning the completeness of ordered completion
for two settings. Moreover, we revisit and extend results of
M\'etivier concerning canonicity of rewrite systems. All proofs presented
in the paper have been formalized in Isabelle/HOL\@.
\end{abstract}

\maketitle

\section{Introduction}

Reasoning with equalities is pervasive in computer science and
mathematics, and has consequently been one of the main research areas of
automated deduction. Indeed completion as introduced by Knuth and
Bendix~\cite{KB70} has evolved into a fundamental technique whose ideas
appear throughout automated reasoning whenever equalities are present.
Many variants of the original calculus have since been proposed.

Bachmair, Dershowitz, and Hsiang~\cite{BDH86} recast completion procedures
as inference systems. This style of presentation,
\emph{abstract completion}, has become the standard to describe completion
procedures and \emph{proof orders} the accompanying tool to establish
correctness~\cite{BDH86,BDP89,B91}, \ie, that under certain
conditions, exhaustive application of the inference rules results in a
terminating and confluent rewrite system whose equational theory is
equivalent to the initial set of equations.

In this paper we present new, modular correctness proofs,
not relying on proof orders, for five
abstract completion systems presented in the literature.
Here, we use \emph{modular} in the following sense:
Proof orders have to be powerful (and thus complex) enough to cover
all intermediate results (that is, proof orders are a \emph{global}
method), while for our new proofs, we \emph{locally} apply well-founded
induction with an order that is just strong enough for the current
intermediate result.
All proofs are fully formalized in Isabelle/HOL\@.
First, we consider finite (\KBf) and infinite (\KBi) runs of classical
Knuth-Bendix completion~\cite{KB70}.
These two settings demand different proofs since in the latter
case the inference system exhibits a stronger side condition. While
our correctness proof for \KBf relies on a new notion we dub
\emph{peak decreasingness}, for the case of \KBi we employ a simpler
version of this criterion called \emph{source decreasingness}. To enhance
applicability by covering efficient implementations, our proofs
support the critical pair criterion known as primality~\cite{KMN88}.

The relevance of infinite runs is illustrated by the following example.

\begin{exa}%
\label{example: braid monoid}
Consider the set of equations $\EE = \{ \m{aba} \approx \m{bab} \}$
of the three-strand positive braid monoid. Kapur and Narendran~\cite{KN85}
proved that $\EE$ admits no \emph{finite} complete presentation.
However, taking the Knuth-Bendix order~\cite{KB70} with $\m{a}$ and
$\m{b}$ of weight $1$ and $\m{a} > \m{b}$ in the precedence, completion
produces in the limit the following infinite complete
presentation of $\EE$
\begin{align*}
\{ \m{aba} \to \m{bab} \} \cup
\{ \m{ab}^n\m{ab} \to \m{babba}^{n-1} \mid n \geqslant 2 \}
\end{align*}
which can be used to decide the validity problem for $\EE$.%
\footnote{Burckel~\cite{B01} constructed a complete rewrite system
consisting of four rules with an additional symbol, which is no longer a
complete presentation of $\EE$ but can be also used to decide the
validity problem for $\EE$.}
\end{exa}

Completion procedures, when successful, produce a complete system.
Natural questions include whether such systems are unique and whether
all complete systems for a given set of equations can be obtained by
completion. For canonical systems, which are complete systems that
satisfy an additional normalization requirement,
M\'etivier~\cite{M83} obtained interesting results.
In this paper we revisit and extend his work.

A special case of \KBf that is known to be decidable is the
completion of ground systems~\cite{S93}. We present new
correctness and completeness proofs for the
corresponding inference system \KBg, based on the recent notion of
random descent~\cite{vOT16}.

On a given set of input equalities, Knuth-Bendix completion can behave
in three different ways: it may (1) succeed to compute a complete
system in finitely many steps, (2) fail due to unorientable equalities,
or (3) continuously compute approximations of a complete system without
ever terminating.
As a remedy to problem (2), ordered completion was developed by Bachmair,
Dershowitz, and Plaisted~\cite{BDP89}. Ordered completion never fails
and can produce a ground-complete system in the limit.
Although the price to be paid is that the resulting system is in general
only complete on ground terms, this is actually sufficient for many
applications in theorem proving.
Refutational theorem proving~\cite{BDP89} owes its semi-decidability
to the unfailing nature of ordered completion.
Again employing peak decreasingness, we obtain a new correctness
proof of ordered completion (\KBo). Next, we turn to completeness
results for ordered completion, \ie, to sufficient criteria for an
ordered completion procedure to produce a complete system. We first
reprove the case of a total reduction order, which assumes a
slightly stronger notion of simplifiedness than the original
result~\cite{BDP89} though. Then we
consider the completeness result for linear completion (\KBl)
due to Devie~\cite{D91}.

For easy reference, Table~\ref{roadmap} provides pointers to
the main definitions and results we present in this paper.
\begin{table}
\caption{Roadmap.}%
\label{roadmap}
\begin{center}
\renewcommand{\arraystretch}{1.25}
\begin{tabular}{@{}rllllc@{}}
\toprule
& \KBf & \KBg & \KBi & \KBo & \KBl \\
\midrule
inference system &
\ref{def:KBf} & 
\ref{def:KBg} & 
\ref{def:KBi} & 
\ref{def:KBo} & 
\ref{def:KBl} 
\\
fairness &
\ref{def:KBf fairness} & 
-- &
\ref{def:KBi fairness} & 
\ref{def:KBo fairness} & 
\ref{def:KBl fairness} 
\\
correctness &
\ref{thm:KBf correctness} & 
\ref{thm:KBg correctness} & 
\ref{thm:KBi correctness} & 
\ref{thm:complete presentation okb} & 
\ref{thm:KBl correctness} 
\\[-0.7ex]
& & & &
\ref{thm:KBo correctness} & 
\\
completeness &
-- &
\ref{thm:KBg completeness} & 
-- &
\ref{thm:KBo completeness} & 
\ref{thm:KBl completeness} 
\\
\bottomrule
\end{tabular}
\renewcommand{\arraystretch}{1}
\end{center}
\end{table}

The remainder of this paper is organized as follows.
We present required preliminaries in \secref{preliminaries},
followed by the abstract confluence criteria of peak and source
decreasingness, as well as a fairly detailed analysis of critical
pairs.
In \secref{finite runs} we recall the
inference rules for (abstract) Knuth-Bendix completion and present our
formalized correctness proof for finite runs. In
\secref{canonicity} we present our results on
canonical systems and normalization equivalence.
We discuss ground completion in \secref{ground completion}.
Infinite runs are the subject of \secref{infinite runs} and
in \secref{ordered completion} we extend our correctness results to
ordered completion. Completeness of ordered completion is the
topic of \secref{completeness}.
We conclude in \secref{conclusion} with a few suggestions for future
research.

Our formalizations are part of
the \emph{Isabelle Formalization of Rewriting}
\isafor~\cite{TS09}\footnote{\url{http://cl-informatik.uibk.ac.at/isafor}}
version~{\isaforversion}.
Below we list the relevant Isabelle theory files grouped by their
subdirectories inside \isafor:
\begin{center}
\def\ind{\makebox[5mm]{}}%
\begin{tabular}[t]{l}
\theorydir{thys/Abstract_Completion/} \\
\ind\theory{thys/Abstract_Completion}{Abstract_Completion.thy} \\
\ind\theory{thys/Abstract_Completion}{Completion_Fairness.thy} \\
\ind\theory{thys/Abstract_Completion}{CP.thy} \\
\ind\theory{thys/Abstract_Completion}{Ground_Completion.thy} \\
\ind\theory{thys/Abstract_Completion}{Peak_Decreasingness.thy} \\
\ind\theory{thys/Abstract_Completion}{Prime_Critical_Pairs.thy}
\end{tabular}
\hfil
\begin{tabular}[t]{l}
\theorydir{thys/Confluence_and_Completion/} \\
\ind\theory{thys/Confluence_and_Completion}{Ordered_Completion.thy} \\
\\
\theorydir{thys/Normalization_Equivalence/} \\
\ind\theory{thys/Normalization_Equivalence}{Encompassment.thy} \\
\ind\theory{thys/Normalization_Equivalence}{Normalization_Equivalence.thy}
\end{tabular}
\end{center}
In the remainder we provide hyperlinks (marked by \extlinkicon) to an HTML
rendering of our formalization.
Moreover, whenever we say that a proof is ``formalized,'' what we mean is
that it is``formalized in Isabelle/HOL.''
And when we ``present a formalized proof,'' we give a textual representation of 
a formalized proof.

This paper and the accompanying formalization are substantially extended
and revised versions of some of our previous work we published in the
ITP~\cite{HMS14} and FSCD~\cite{HMSW17} conferences.
The former presented a new correctness proof for finite runs of
Knuth-Bendix completion. Its modular design separates
concerns rather than relying on a single proof order, thus
rendering it more formalization friendly.
In revised form, these results are included in
\secref{finite runs}.
The FSCD contribution extended this novel proof approach to both
infinite runs and ordered completion (see Sections~\ref{sec:infinite runs} and~\ref{sec:ordered completion}).
It moreover incorporated canonicity results (\secref{canonicity}).
In addition to these results we present
new and
formalized proofs of correctness and completeness of
ground completion (\secref{ground completion}), as well as
completeness of ordered completion for two different cases
(\secref{completeness}).
At the end of each section, we remark on the novelty of the respective
results and their proofs.

\section{Preliminaries}%
\label{sec:preliminaries}

We assume familiarity with the basic notions of
abstract rewrite systems, term rewrite systems, and
completion~\cite{B91,BN98}, but nevertheless shortly recapitulate
terminology and notation that we use in the remainder.

\subsection{Rewrite Systems}

For an arbitrary binary relation $\xr[\alpha]{}$, we write
$\xl[\alpha]{}$, $\xlrs[\alpha]{}$, $\xr[\alpha]{=}$, $\xr[\alpha]{+}$,
and $\xr[\alpha]{*}$
to denote its \emph{inverse}, its \emph{symmetric closure}, its
\emph{reflexive closure}, its \emph{transitive closure}, and its
\emph{reflexive transitive closure}, respectively.
The reflexive, transitive, and symmetric closure $\xlrs[\alpha]{*}$
of $\xr[\alpha]{}$ is called \emph{conversion}, and a sequence
of the form
\(
c_0 \mathrel{\xlrs[\alpha]{}} c_1 \mathrel{\xlrs[\alpha]{}} \cdots
\mathrel{\xlrs[\alpha]{}} c_n
\)
is referred to as a conversion between $c_0$ and $c_n$ (of length $n$).
For a binary relation $R$ without arrow notation, we also write $R^{-1}$
for its \emph{inverse} and $R^\pm$ for its \emph{symmetric closure}
$R \cup R^{-1}$.
We further use $\downarrow_\alpha$ as abbreviation for the
\emph{joinability relation} $\xr[\alpha]{*} \cdot \xl[\alpha]{*}$, where
from here on $\cdot$ denotes relation composition.
If $a \xr[\alpha]{} b$ for no $b$
then we say that $a$ is a \emph{($\xr[\alpha]{}$-)normal form}. 
The set of all normal forms of a given relation $\xr[\alpha]{}$ is denoted
by $\NF({\xr[\alpha]{}})$. By
$a \xr[\alpha]{!} b$ we abbreviate ${a \xr[\alpha]{*} b} \land {b \in
\NF({\xr[\alpha]{}})}$ and we call $b$ a \emph{normal form of $a$}.
Given two binary relations $\xr[\alpha]{}$ and $\xr[\beta]{}$, we use
$\xr[\alpha]{} / \xr[\beta]{}$ as shorthand for
the \emph{relative rewrite relation}
$\xr[\beta]{*} \cdot \xr[\alpha]{} \cdot \xr[\beta]{*}$.
An \emph{abstract rewrite system} (ARS for short) $\AA$ is a set $A$,
the carrier, equipped with a binary relation $\to$. Sometimes
we partition the binary relation into parts according to a set $I$ of
indices (or labels). Then we write
$\AA = \ARSn{A}{I}$ where we denote the part of the relation with label
$\alpha$ by $\xr[\alpha]{}$, \ie,
${\to} = \bigcup~\{ {\xr[\alpha]{}} \mid \alpha \in I \}$.

We assume a given signature $\FF$ and a set of variables $\VV$. The set
of terms built up from $\FF$ and $\VV$ is denoted by
$\TT(\FF,\VV)$, while $\TT(\FF)$ denotes the set of ground terms.
\emph{Positions} are strings of positive integers which are
used to address subterms. The set of positions in a term $t$ is
denoted by $\Pos(t)$. The subset consisting of the positions addressing
function symbols in $t$ is denoted by $\PosF(t)$ whereas
$\PosV(t) = \Pos(t) - \PosF(t)$ is the set of variable positions in $t$.
We write
$p \leqslant q$ if $p$ is a prefix of $q$ and $p \parallel q$ if neither
$p \leqslant q$ nor $q \leqslant p$. If $p \leqslant q$ then the unique
position $r$ such that $pr = q$ is denoted by $q \backslash p$.
A \emph{substitution} is a mapping $\sigma$ from variables to terms such
that its domain $\{ x \in \VV \mid \sigma(x) \neq x \}$ is finite.
Applying a substitution $\sigma$ to a term $t$ is written $t\sigma$. A
variable substitution is a substitution from $\VV$ to $\VV$ and a
\emph{renaming} is a bijective variable substitution. A term $s$ is a
\emph{variant} of a term $t$ if $s = t\sigma$ for some renaming $\sigma$.
A pair of terms $(s,t)$ is sometimes considered an \emph{equation},
then we write $s \approx t$, and sometimes a \emph{(rewrite) rule}, then
we write $s \to t$. In the latter case we assume the
\emph{variable condition}, \ie, that the left-hand side $s$ is not a
variable and that variables of the right-hand side $t$ are all contained
in $t$.
A set $\EE$ of equations is called an \emph{equational system} (ES for
short) and a set $\RR$ of rules a \emph{term rewrite system} (TRS for
short).
Sets of pairs of terms $\EE$ induce a \emph{rewrite relation}
$\xr[\EE]{}$ by closing their components under contexts and substitutions.
A rewrite step $s \to_\EE t$ at a position $p \in \Pos(s)$ is called
\emph{innermost} and denoted by $s \ito_\EE t$ if no proper subterm of
$s|_p$ is reducible in $\EE$.
The \emph{equational theory} induced by $\EE$ consists of all
pairs of terms $s$ and $t$ such that $s \conv_\EE t$.
If $\ell \to r$ is a rewrite rule and $\sigma$ is a renaming then the
rewrite rule $\ell\sigma \to r\sigma$ is a \emph{variant} of $\ell \to r$.
A TRS is said to be \emph{variant-free} if it does not contain rewrite
rules that are variants of each other.

Two terms $s$ and $t$ are called \emph{literally similar}, written
$s \doteq t$, if $s\sigma = t$ and $s = t\tau$ for some substitutions
$\sigma$ and $\tau$.
Two TRSs $\RR_1$ and $\RR_2$
are called \emph{literally similar}, denoted by $\RR_1 \doteq \RR_2$, if
every rewrite rule in $\RR_1$ has a variant in $\RR_2$ and vice versa.
The following result is folklore; we formalized the non-trivial proof.

\begin{lem}%
\label{lem:variants-terms}
Two terms $s$ and $t$ are variants of each other if and only if
$s \doteq t$.
\hfill\isaforlink{Term_More}{lem:variants_imp_renaming}
\end{lem}

We say that $s$ \emph{encompasses} $t$, written $s \encompasses t$,
whenever $s = C[t\sigma]$ for some context $C$ and substitution $\sigma$.
\emph{Proper encompassment} is defined by
${\prencompasses} = {\encompasses \setminus \rencompasses}$
and known to be well-founded.
The identity ${\encompasses} = {\prencompasses} \cup {\doteq}$ is
well-known.
For a well-founded order $>$, we write $\MUL$ to denote its
\emph{multiset extension} and $>_\lex$ to denote its
\emph{lexicographic extension} as defined by Baader and
Nipkow~\cite{BN98}.

A TRS $\RR$ is \emph{terminating} if $\to_\RR$ is well-founded, and
\emph{weakly normalizing} if every term has a normal form.
It is
\emph{(ground-)confluent} if $s \From{\RR}{*} \cdot \to_\RR^* t$ 
implies $s \to_\RR^* \cdot \From{\RR}{*} t$ for all (ground) terms
$s$ and $t$. It is \emph{(ground-)complete} if it is terminating and 
(ground) confluent. We say that $\RR$ is a
\emph{complete presentation} of
an ES $\EE$ if $\RR$ is complete and ${\conv_\RR} = {\conv_\EE}$.
A TRS $\RR$ is \emph{left-reduced} if
$\ell \in \NF(\RR \setminus \{ \ell \to r \})$ for every rewrite rule
$\ell \to r$ in $\RR$, and
\emph{right-reduced} if
$r \in \NF(\RR)$ for every rewrite rule $\ell \to r$ in $\RR$. A
\emph{reduced} TRS is left- and right-reduced. A reduced complete TRS is
called \emph{canonical}.

We make use of the following result due to Bachmair and
Dershowitz~\cite{BD86}, where \emph{quasi-commutation} of $R$ over $S$
means that the inclusion $S \cdot R \subseteq R \cdot {(R \cup S)}^*$ holds.

\begin{lem}%
\label{lem:BD}
Let $R$ and $S$ be binary relations.
\begin{enumerate}
\item%
\label{lem:BD:1}
If $R$ quasi-commutes over $S$ then well-foundedness of $\mathrel{R} /
\mathrel{S}$ and $R$ coincide.
\hfill\afplink[Abstract_Rewriting]{qc_SN_relto_iff}
\item%
\label{lem:BD:2}
If $\mathrel{R} / \mathrel{S}$ and $S$ are well-founded then $R \cup S$ is
well-founded.
\hfill\afplink[Relative_Rewriting]{SN_relto_split}
\end{enumerate}
\end{lem}

\begin{lem}%
\label{lem:proper encompassment extension1}
If $R$ is a well-founded rewrite relation then
$\mathrel{(R \cup {\prencompasses})} / \mathrel{\encompasses}$
is well-founded.
\hfill
\isaforlink{Abstract_Completion}{lem:SN_encomp_Un_less_relto_encompeq}
\end{lem}

\begin{proof}
First we show the inclusion ${\encompasses \cdot \mathrel{R}}
\subseteq {\mathrel{R} \cdot \encompasses}$.
Suppose $s \encompasses t \mathrel{R} u$. So $s = C[t\sigma]$ for some
context $C$ and substitution $\sigma$. Because $R$ is closed under
contexts and substitutions, $s \mathrel{R} C[u\sigma]$. Moreover,
$C[u\sigma] \encompasses u$. This establishes the inclusion, and we
conclude that $R$ (quasi-)commutes over $\encompasses$. 
Because $R$ is well-founded, it follows from \lemref[1]{BD}
that the relation $\mathrel{R} / \mathrel{\encompasses}$ is well-founded
too. Then $\mathrel{R} / \mathrel{\prencompasses}$ is well-founded since
it is contained in $\mathrel{R} / \mathrel{\encompasses}$.
As $\prencompasses$ is well-founded, it follows from \lemref[2]{BD}
that $\mathrel{R} \cup \mathrel{\prencompasses}$ is well-founded.
We have ${\encompasses} \cdot {\prencompasses} \subseteq {\prencompasses}$
and thus $R \cup {\prencompasses}$ quasi-commutes over $\encompasses$.
Another application of \lemref[1]{BD} yields the well-foundedness of
$\mathrel{(R \cup {\prencompasses})} / \mathrel{\encompasses}$.
\end{proof}

\subsection{Abstract Confluence Criteria}%
\label{sec:decreasingness}

We use the following simple confluence criterion for ARSs to replace
Newman's Lemma in the correctness proof of abstract completion.
In the sequel, we will refer to a conversion of the form
$\FromB{\AA} \cdot \to_\AA$ as a \emph{peak}.

\begin{defi}[Peak Decreasingness
\isaforlink{Peak_Decreasingness}{asm:peak_decreasing}]%
\label{def:peak decreasing}
An ARS $\AA = \ARSn{A}{I}$ is \emph{peak decreasing} if there exists a
well-founded order $>$ on $I$ such that for all $\alpha, \beta \in I$ the
inclusion
\[
{\FromB{\alpha} \cdot \to_\beta}
~\subseteq~
{\xlr[\,\vee\alpha\beta~]{*}}
\]
holds. Here ${\vee\alpha\beta}$ denotes the set
$\{ \gamma \in I \mid \text{$\alpha > \gamma$ or $\beta > \gamma$} \}$ and
if $J \subseteq I$ then
$\smash{\Xlr{\,J~}{*}{-.1}}$ denotes a conversion consisting of
${\xr[J]{}} = \bigcup~\{ {\xr[\gamma]{}} \mid \gamma \in J \}$ steps.
\end{defi}

Peak decreasingness is a special case of decreasing diagrams~\cite{vO94},
which is known as a very powerful confluence criterion. For the sake
of completeness, we present an easy direct (and formalized) proof of the
sufficiency of peak decreasingness for confluence. We denote by
$\MM(J)$ the set of all multisets over a set $J$.

\begin{lem}%
\label{lem:pd => cr}
Every peak decreasing ARS is confluent.
\hfill\isaforlink{Peak_Decreasingness}{lem:CR}
\end{lem}
\begin{proof}
Let $>$ be a well-founded order on $I$ which shows that the ARS
$\AA = \ARSn{A}{I}$ is peak decreasing. With every conversion $C$ in $\AA$
we
associate the multiset $M_C$ consisting of the labels of its steps. These
multisets are compared by the multiset extension $\MUL$ of $>$, which is a
well-founded order on $\MM(I)$. We prove
${\conv} \subseteq {\join}$ by
well-founded induction on $\MUL$. Consider a conversion $C$ between $a$
and $b$. We either have $a \join b$ or
$a \conv \cdot \from \cdot \to \cdot \conv b$. In the former case we are
done. In the latter case there exist labels $\alpha, \beta \in I$ and
multisets $\Gamma_1, \Gamma_2 \in \MM(A)$ such that
$M_C = \Gamma_1 \uplus \{ \alpha, \beta \} \uplus \Gamma_2$. By the peak
decreasingness assumption there exists a conversion $C'$ between $a$ and
$b$ such that $M_{C'} = \Gamma_1 \uplus \Gamma \uplus \Gamma_2$ with
$\Gamma \in \MM(\vee\alpha\beta)$. We obviously have
$\{ \alpha, \beta \} \MUL \Gamma$ and hence $M_C \MUL M_{C'}$. Finally,
we obtain $a \join b$ from the induction hypothesis.
\end{proof}

A similar criterion to show the Church-Rosser modulo property will be used
in \secref{completeness}. Here an ARS $\AA$ is called
\emph{Church-Rosser modulo} an ARS $\BB$ if the inclusion
\[
{\xlr[\AA \scup \BB\,]{*}} \subseteq
{\xrightarrow[\AA]{*} \cdot \xlr[\BB]{*} \cdot \xleftarrow[\AA]{*}}
\]
holds.

\begin{defi}[Peak Decreasingness Modulo
\isaforlink{Peak_Decreasingness}{asm:peak_decreasing_mod}]%
\label{def:peak decreasing modulo}
Consider two ARSs $\AA = \ARSn{A}{I}$ and $\BB = \ARSn[\beta]{B}{J}$.
Then $\AA$ is \emph{peak decreasing modulo} $\BB$ if there exists a
well-founded order $>$ on $I \cup J$ such that for all $\alpha \in I$
and $\gamma \in I \cup J$ the inclusion
\[
{{} \FromB{\alpha} \cdot \to_{\gamma} {}}
\:\subseteq\:
{\xlr[\smash{\,\vee \alpha\gamma~}]{*}}
\]
holds. Here $\vee \alpha\gamma$ denotes the set $\{ \delta \in I \cup J
\mid \text{$\alpha > \delta$ or $\gamma > \delta$} \}$.
\end{defi}

\begin{lem}%
\label{lem:pdm => crm}
If $\AA$ is peak decreasing modulo $\BB$ then $\AA$ is Church-Rosser
modulo $\BB$.
\hfill\isaforlink{Peak_Decreasingness}{lem:CRm}
\end{lem}
\begin{proof}
Let $x_1 \fromto_{\alpha_1} \cdots \fromto_{\alpha_n} x_{n+1}$ and
$M = \{ \seq{\alpha} \}$.  We use induction on $M$ with respect to
$>_\mul$ to show
$x_1 \to_\AA^{*} \cdot \fromto_\BB^{*} \cdot \From{\AA}{*} x_{n+1}$.
If the given conversion is not of the desired shape, there is an
index $1 \leqslant i < n$ such that
$x_i \FromB{\alpha} x_{i+1} \to_\gamma x_{i+2}$ or
$x_i \FromB{\gamma} x_{i+1} \to_\alpha x_{i+2}$
for some $\alpha \in I$ and $\gamma \in I \cup J$.
As the reasoning is similar, we only consider the former case. By peak
decreasingness there are labels
$\seq[m]{\beta}$ with
$x_i \fromto_{\beta_1} \cdots \fromto_{\beta_m} x_{i+2}$ such that
$\beta_j \in \vee\alpha\gamma$ for all $1 \leqslant j \leqslant m$.
Writing $N$ for the multiset $\{ \seq[m]{\beta} \}$,
we obtain $M >_\mul (M - \{ \alpha, \gamma \}) \uplus N$
from $\alpha, \gamma \in M$ and
$\{ \alpha, \gamma \} >_\mul N$.
Therefore, the claim follows from the induction hypothesis.
\end{proof}

For the correctness proof in \secref{infinite runs}
we use a simpler notion than peak decreasingness.

\begin{defi}[Source Decreasingness
\isaforlink{Peak_Decreasingness}{asm:source_decreasing}]
Let $\AA = \ARS$ be an ARS equipped with a well-founded relation $>$ on
$A$, and we write $b \xrightarrow{a} c$ if $b \to c$ and
$a = b$. We say that $\AA$ is \emph{source decreasing} if the inclusion
\[
{{} \from a \to {}} \:\subseteq\: {\xlr{\smash{\,\vee a~}}^*}
\]
holds for all $a \in A$. Here
$\from a \to$ denotes the binary relation consisting of all
pairs $(b,c)$ such that $a \to b$ and $a \to c$. Moreover,
$\xlrs{\smash{\,\vee a~}}^*$ denotes the binary relation
consisting of all pairs of elements that are connected by a conversion
in which all steps are labeled with an element smaller than $a$.
\end{defi}

Source decreasingness is the specialization of peak decreasingness to
source labeling~\cite[Example~6]{vO08}.
It is closely related to the \emph{connectedness-below}
criterion of Winkler and Buchberger~\cite{WB86}. Unlike the latter,
source decreasingness does not entail termination.
For instance, for $\m{a} > \m{b}$ and $\m{a} > \m{c}$ the non-terminating
ARS
\begin{center}
\begin{tikzpicture}[on grid,node distance=12mm,baseline=(1).baseline]
\node (1)              {$\e{b}$};
\node (2) [right=of 1] {$\e{a}$};
\node (3) [right=of 2] {$\e{c}$};
\draw[->] (1) edge[bend right] (2);
\draw[->] (2) edge[bend right] (1);
\draw[->] (2) edge[bend left] (3);
\draw[->] (3) edge[bend left] (2);
\end{tikzpicture}
\end{center}
is source decreasing but the connectedness-below criterion does not apply.

\begin{lem}%
\label{lem:sd => pd}
Every source decreasing ARS is peak decreasing.
\hfill
\isaforlink{Peak_Decreasingness}{sub:ars_source_decreasing}
\end{lem}

Peak decreasingness as a special case of decreasing diagrams was first
considered in our ITP publication~\cite{HMS14}
(the modulo version in \defref{peak decreasing modulo} is new).
Source decreasingness originates from our later FSCD
contribution~\cite{HMSW17}.

\subsection{Critical Peaks}%
\label{sec:critical peaks}

Completion is based on critical pair analysis. In this subsection we
present a version of the critical pair lemma that incorporates
primality (cf.\ \defref{pcp} below).

\begin{defi}[Overlaps
\isaforlink{CP}{def:overlap}]
An \emph{overlap} of a TRS $\RR$ is a triple
$\langle \ell_1 \to r_1, p, \ell_2 \to r_2 \rangle$,
consisting of two rewrite rules and a position, satisfying the
following properties:
\begin{itemize}
\item
there are renamings $\pi_1$ and $\pi_2$ such that
$\pi_1 (\ell_1 \to r_1), \pi_2 (\ell_2 \to r_2) \in \RR$
(\ie, the rules are variants of rules in $\RR$),
\smallskip
\item
$\Var(\ell_1 \to r_1) \cap \Var(\ell_2 \to r_2) = \varnothing$
(\ie, the rules have no common variables),
\smallskip
\item
$p \in \Pos_\FF(\ell_2)$,
\smallskip
\item
$\ell_1$ and ${\ell_2}|_p$ are unifiable,
\smallskip
\item
if $p = \epsilon$ then $\ell_1 \to r_1$ and $\ell_2 \to r_2$ are not
variants of each other.
\end{itemize}
\end{defi}

\noindent
In general this definition may lead to an infinite set of overlaps, since
there are infinitely many possibilities of taking variable disjoint
variants of rules. Fortunately it can be shown
that overlaps that originate from the same two rules are variants of each
other. Overlaps give rise to critical peaks and pairs.

\begin{defi}[Critical Peaks
\isaforlink{CP}{def:cpeaks2}
and Pairs
\isaforlink{CP}{def:CP2}]
Suppose $\langle \ell_1 \to r_1, p, \ell_2 \to r_2 \rangle$ is an overlap
of a TRS $\RR$. Let $\sigma$ be a most general unifier of $\ell_1$ and
${\ell_2}|_p$. The term $\ell_2\sigma{[\ell_1\sigma]}_p = \ell_2\sigma$
can be reduced in two different ways:
\begin{center}
\begin{tikzpicture}[minimum height=6mm]
\node(s){$\ell_2\sigma[\ell_1\sigma]_p = \ell_2\sigma$};
\coordinate(swest) at (s.west);
\coordinate(seast) at (s.east);
\node(t)[below left of=swest, anchor=north east]%
  {$\ell_2\sigma[r_1\sigma]_p$};
\node(u)[below right of=seast, anchor=north west]%
  {$r_2\sigma$};
\draw[big arrow] (s.south west) -- node [left=-1mm,yshift=2mm]%
  {\scriptsize $\ell_1 \to r_1$} node [right=0mm,yshift=-1mm]%
  {\scriptsize $p$} (t.north east);
\draw[big arrow] (s.south east) -- node [right=-1mm,yshift=2mm]%
  {\scriptsize $\ell_2 \to r_2$} node [left=0mm,yshift=-1mm]%
  {\scriptsize $\epsilon$} (u.north west);
\end{tikzpicture}
\end{center}
We call the quadruple
$(\ell_2\sigma{[r_1\sigma]}_p, p, \ell_2\sigma, r_2\sigma)$ a
\emph{critical peak}
and the equation
$\ell_2\sigma{[r_1\sigma]}_p \approx r_2\sigma$ a \emph{critical pair} of
$\RR$, obtained from the overlap. The set of all critical pairs of $\RR$
is denoted by $\CP(\RR)$.
\end{defi}

In our formalization of the above definition, instead of an arbitrary
most general unifier, we use \emph{the} most general unifier
computed by the formalized unification algorithm that is part of
\isafor (thereby removing one degree of freedom and making it easier to
show that only finitely many critical pairs have to be considered for
finite TRSs).

A critical peak $(t, p, s, u)$ is usually denoted by $\cpeak{t}{p}{s}{u}$.
It can be shown
that different critical peaks
and pairs obtained from two variants of the same overlap are variants of
each other. Since rewriting is equivariant under permutations, it is
enough to consult finitely many critical pairs or peaks for finite TRSs
(one for each pair of rules and each appropriate position) in order to
conclude rewriting related properties (like joinability or fairness,
see below) for all of them.

We present a variation of the well-known critical pair lemma for critical
peaks and its formalized proof. The slightly cumbersome statement is
essential to avoid gaps in the proof of \lemref{pcp_root} below.

\begin{lem}%
\label{lem:cpeakL}
Let $\RR$ be a TRS\@. If $t \xfrom{\RR}{p_1} s \xrightarrow{p_2}_\RR u$ then
one of the following holds:
\hfill\isaforlink{CP}{lem:peak_imp_join_or_S3_cpeaks}
\begin{enumerate}
\item%
\label{lem:cpeakL:a}
$t \join_\RR u$,
\item%
\label{lem:cpeakL:b}
$p_2 \leqslant p_1$ and
${\cpeak{t|_{p_2}}{p_1 \backslash p_2}{s|_{p_2}}{u|_{p_2}}}$
is an instance of a critical peak, or
\item%
\label{lem:cpeakL:c}
$p_1 \leqslant p_2$ and
${\cpeak{u|_{p_1}}{p_2 \backslash p_1}{s|_{p_1}}{t|_{p_1}}}$
is an instance of a critical peak.
\end{enumerate}
\end{lem}
\begin{proof}
Consider an arbitrary peak $t \FromB{p_1, \ell_1' \to r_1', \sigma_1'} s
\to_{p_2, \ell_2 \to r_2, \sigma_2} u$. If $p_1 \parallel p_2$ then
\[
t \to_{p_2, \ell_2 \to r_2, \sigma_2} t{[r_2\sigma_2]}_{p_2} =
u{[r_1'\sigma_1']}_{p_1} \fromB{p_1, \ell_1' \to r_1', \sigma_1'} u
\]
If the positions of the contracted redexes are not parallel then one of
them is above the other.  Without loss of generality we assume that
$p_1 \geqslant p_2$. Let $p = p_1 \backslash p_2$. Moreover, let $\pi$ be
a permutation such that $\ell_1 \to r_1 = \pi (\ell_1' \to r_1')$
and $\ell_2 \to r_2$ have no variables in common. Such a permutation
exists since we only have to avoid the finitely many variables of
$\ell_2 \to r_2$ and assume an infinite set of variables. Furthermore, let
$\sigma_1 = \pi^{-1} \permute \sigma_1'$. We have
$t = s{[r_1\sigma_1]}_{p_1} = s{[\ell_2\sigma_2{[r_1\sigma_1]}_p]}_{p_2}$ and
$u = s{[r_2\sigma_2]}_{p_2}$. We consider two cases depending on whether
$p \in \Pos_\FF(\ell_2)$ in conjunction with the fact that whenever
$p = \epsilon$ then $\ell_1 \to r_1$ and $\ell_2 \to r_2$ are not
variants, is true or not.
\begin{itemize}
\item
Suppose $p \in \Pos_\FF(\ell_2)$ and $p = \epsilon$ implies that
$\ell_1 \to r_1$ and $\ell_2 \to r_2$ are not variants. Let
$\sigma'(x) = \sigma_1(x)$ for $x \in \Var(\ell_1 \to r_1)$ and
$\sigma'(x) = \sigma_2(x)$, otherwise. The substitution $\sigma'$ is a
unifier of ${\ell_2}|_p$ and $\ell_1$: $({\ell_2}|_p)\sigma' =
(\ell_2\sigma_2)|_p = \ell_1\sigma_1 = \ell_1\sigma'$. Then
$\langle \ell_1 \to r_1, p, \ell_2 \to r_2 \rangle$ is an overlap. Let
$\sigma$ be a most general unifier of ${\ell_2}|_p$ and $\ell_1$. Hence
$\cpeak{\ell_2\sigma{[r_1\sigma]}_p}{p}{\ell_2\sigma}{r_2\sigma}$ is a
critical peak and there exists a substitution $\tau$ such that
$\sigma' = \sigma\tau$. Therefore
\[
\cpeak%
  {\ell_2\sigma_2{[r_1\sigma_1]}_p = (\ell_2\sigma{[r_1\sigma]}_p)\tau}%
  {p}
  {(\ell_2\sigma)\tau}
  {(r_2\sigma)\tau = r_2\sigma_2}
\]
and thus~\eqref{lem:cpeakL:b} is obtained.
\item
Otherwise, either $p = \epsilon$ and $\ell_1 \to r_1$, $\ell_2 \to r_2$
are variants, or $p \notin \Pos_\FF(\ell_2)$. In the former case it is
easy to show that $r_1\sigma_1 = r_2\sigma_2$ and hence $t = u$. In the
latter case, there exist positions $q_1$, $q_2$ such that $p = q_1 q_2$
and $q_1 \in \PosV(\ell_2)$. Let ${\ell_2}|_{q_1}$ be the variable $x$.
We have $\sigma_2(x)|_{q_2} = \ell_1\sigma_1$. Define the substitution
$\sigma_2'$ as follows:
\[
\sigma_2'(y) = \begin{cases}
\sigma_2(y){[r_1\sigma_1]}_{q_2} & \text{if $y = x$} \\
\sigma_2(y) & \text{if $y \neq x$}
\end{cases}
\]
Clearly $\sigma_2(x) \to_\RR \sigma_2'(x)$, and thus
$r_2\sigma_2 \to^* r_2\sigma_2'$. We also have
\[
\ell_2\sigma_2{[r_1\sigma_1]}_p = \ell_2\sigma_2{[\sigma_2'(x)]}_{q_1}
\to^* \ell_2\sigma_2' \to r_2\sigma_2'
\]
Consequently, $t \to^* s{[r_2\sigma_2']}_{p_2} \FromT{*} u$. Hence,~\eqref{lem:cpeakL:a} is concluded.
\qedhere
\end{itemize}
\end{proof}

\noindent
An easy consequence of the above lemma is that for every peak
$t \FromB{\RR} s \to_\RR u$ we have $t \join_\RR u$ or
$t \fromto_{\CP(\RR)} u$. It might be interesting to note that in our
formalization of the above proof we do actually not need the fact that
left-hand sides of rules are not variables.

\begin{defi}[Prime Critical Peaks and Pairs
\isaforlink{Prime_Critical_Pairs}{def:PCP}]%
\label{def:pcp}
A critical peak $\cpeak{t}{p}{s}{u}$ is \emph{prime} if all proper
subterms of $s|_p$ are normal forms.
A critical pair is called prime if it is derived from a prime critical
peak. We write $\PCP(\RR)$ to denote the set of all prime critical pairs
of a TRS $\RR$.
\end{defi}

\begin{defi}
Given a TRS $\RR$ and terms $s$, $t$, and $u$, we write
$t \mathrel{\tds} u$ if $s \to_\RR^+ t$, $s \to_\RR^+ u$, and
$t \join_\RR u$ or $t \fromto_{\PCP(\RR)} u$.
\hfill\isaforlink{Prime_Critical_Pairs}{def:nabla}
\end{defi}

\begin{lem}%
\label{lem:pcp_root}
Let $\RR$ be a TRS\@. If $\cpeak{t}{p}{s}{u}$ is a critical peak then
$t \mathrel{\tds^2} u$.
\hfill
\isaforlink{Prime_Critical_Pairs}{lem:cpeaks_imp_nabla2}
\end{lem}
\begin{proof}
First suppose that all proper subterms of $s|_p$ are normal forms. Then
$t \approx u \in \PCP(\RR)$ and thus $t \mathrel{\tds} u$. Since also
$u \mathrel{\tds} u$, we obtain the desired $t \mathrel{\tds^2} u$. This
leaves us with the case that there is a proper subterm of $s|_p$ that is
not a normal form. By considering an innermost redex in $s|_p$ we obtain
a position $q > p$ and a term $v$ such that $\smash{s \xrightarrow{q} v}$
and all proper subterms of $s|_q$ are normal forms. Now, if
$\smash{\cpeak{v}{q}{s}{u}}$ is an instance of a critical peak then
$v \to_{\PCP(\RR)} u$. Otherwise, $v \join_\RR u$ by \lemref{cpeakL},
since $q \not\leqslant \epsilon$. In both cases we obtain
$v \mathrel{\tds} u$. Finally, we analyze the peak
$\smash{\peak{t}{p}{s}{q}{v}}$ by another application of \lemref{cpeakL}.
\begin{enumerate}
\item
If $t \join_\RR v$, we obtain $t \mathrel{\tds} v$ and thus
$t \mathrel{\tds^2} u$, since also $v \mathrel{\tds} u$.
\medskip
\item
Since $p < q$, only the case that
$\smash{\cpeak{v|_p}{q \backslash p}{s|_p}{t|_p}}$ is an instance of a
critical peak remains. Moreover, all proper subterms of $s|_q$ are
normal forms and thus we have an instance of a prime critical peak. Hence
$t \fromto_{\PCP(\RR)} v$ and together with $v \mathrel{\tds} u$ we
conclude $t \mathrel{\tds^2} u$.
\qedhere
\end{enumerate}
\end{proof}

\begin{lem}%
\label{lem:pcp}
Let $\RR$ be a TRS\@. If $t \FromB{\RR} s \to_\RR u$ then
$t \mathrel{\tds^2} u$.
\hfill
\isaforlink{Prime_Critical_Pairs}{lem:peak_imp_nabla2}
\end{lem}
\begin{proof}
From \lemref{cpeakL}, either $t \join_\RR u$ and we are done, or
$t \FromB{\RR} s \to_\RR u$ contains a (possibly reversed) instance of a
critical peak. By \lemref{pcp_root} we conclude the proof, since
rewriting is closed under substitutions and contexts.
\end{proof}

The following result is due to Kapur
\etal~\cite[Corollary~4]{KMN88}.

\begin{cor}%
\label{cor:pcp - confluence}
A terminating TRS is confluent if and only if all its prime critical pairs
are joinable.
\hfill
\isaforlink{Prime_Critical_Pairs}{lem:SN_imp_CR_iff_PCP_join}
\end{cor}
\begin{proof}
Let $\RR$ be a terminating TRS such that
$\PCP(\RR) \subseteq {\join_\RR}$. We claim that $\RR$ is source
decreasing. As well-founded order we take ${>} = {\to_\RR^+}$. Consider an
arbitrary peak \mbox{$t \FromB{\RR} s \to_\RR u$}. \lemref{pcp} yields a
term $v$
such that $t \mathrel{\tds} v \mathrel{\tds} u$. From the assumption
$\PCP(\RR) \subseteq {\join_\RR}$ we obtain $t \join_\RR v \join_\RR u$.
Since $s \to_\RR^+ v$, all steps in the conversion
$t \join_\RR v \join_\RR u$ are labeled with a term that is smaller than
$s$. Since the two steps in the peak receive the same label $s$, source
decreasingness is established and hence we obtain the confluence of $\RR$
from \lemref{pd => cr}. The reverse direction is trivial.
\end{proof}

Note that unlike for ordinary critical pairs, joinability of prime
critical pairs does not imply local confluence.

\begin{exa}
Consider the TRS $\RR$ given by the three rules:
\begin{xalignat*}{3}
\m{f}(\m{a}) & \to \m{b} &
\m{f}(\m{a}) & \to \m{c} &
\m{a} & \to \m{a}
\end{xalignat*}
The set $\PCP(\RR)$ consists of the two pairs $\m{f}(\m{a}) \approx \m{b}$
and $\m{f}(\m{a}) \approx \m{c}$, which are trivially joinable. But $\RR$
is not locally confluent because the peak
$\m{b} \fromB{\RR} \m{f}(\m{a}) \to_\RR \m{c}$ is not joinable.
\end{exa}

The critical pair lemma (\lemref{cpeakL}) in this section is
due to Knuth and Bendix~\cite{KB70} and Huet~\cite{H80}.
The primality critical pair criterion was first presented by
Kapur, Musser, and Narendran~\cite{KMN88}.
Our presentation is based on the simpler correctness arguments
from our earlier work~\cite{HMS14,HMSW17}.

\section{Correctness for Finite Runs}%
\label{sec:finite runs}

The original completion procedure by Knuth and Bendix~\cite{KB70} was
presented as a concrete algorithm. Later on, Bachmair, Dershowitz, and
Hsiang~\cite{BDH86} presented an inference system for completion and
showed that all \emph{fair} implementations thereof (in particular the
original procedure) are correct. Abstracting from a concrete strategy,
their approach thus has the advantage to cover a variety of
implementations. Below, we recall the inference system, which constitutes
the basis of the results presented in this section.

\begin{defi}[Knuth-Bendix Completion
\isaforlink{Abstract_Completion}{ind:KB}]%
\label{def:KBf}
The inference system \KBf of abstract (Knuth-Bendix) completion
operates on pairs $(\EE,\RR)$ of sets of equations $\EE$ and rules
$\RR$ over a common signature $\FF$.
It consists of the following inference rules,
where we write $\EE,\RR$ for a pair $(\EE,\RR)$ and $\uplus$ denotes
disjoint set union:
\begin{center}
\bigskip
\begin{tabular}{@{}lcl@{\qquad\qquad}lcl@{}}
\tsfs{deduce} &
$\displaystyle \frac
{\EE,\RR}
{\EE \cup \{ s \approx t \},\RR}$
& if $s \FromB{\RR} \cdot \to_\RR t$
&
\tsfs{compose} &
$\displaystyle \frac
{\EE,\RR \uplus \{ s \to t \}}
{\EE,\RR \cup \{ s \to u \}}$
& if $t \to_\RR u$
\\ & \\
&
$\displaystyle \frac
{\EE \uplus \{ s \approx t \},\RR}
{\EE,\RR \cup \{ s \to t \}}$
& if $s > t$
&
&
$\displaystyle \frac
{\EE \uplus \{ s \approx t \},\RR}
{\EE \cup \{ u \approx t \},\RR}$
& if $s \to_\RR u$
\\[-.5ex]
\tsfs{orient} & & & \tsfs{simplify}
\\[-.5ex]
&
$\displaystyle \frac
{\EE \uplus \{ s \approx t \},\RR}
{\EE,\RR \cup \{ t \to s \}}$
& if $t > s$
&
&
$\displaystyle \frac
{\EE \uplus \{ s \approx t \},\RR}
{\EE \cup \{ s \approx u \},\RR}$
& if $t \to_\RR u$
\\ & \\
\tsfs{delete} &
$\displaystyle \frac
{\EE \uplus \{ s \approx s \},\RR}
{\EE,\RR}$ &
&
\tsfs{collapse} &
$\displaystyle \frac
{\EE,\RR \uplus \{ t \to s \}}
{\EE \cup \{ u \approx s \},\RR}$
&
if $t \to_\RR u$
\end{tabular}
\bigskip
\end{center}
Here $>$ is a fixed reduction order on $\TT(\FF,\VV)$.
\end{defi}

\defref{KBf} differs from most of the
inference systems in the literature (like those devised by Bachmair and
Dershowitz~\cite{B91,BD94}) in that we do
not impose an encompassment condition on \tsfs{collapse}.
As long as we only consider \emph{finite} runs (see \defref{KBf fairness}
below)---like in Sections~\ref{sec:finite runs}
to~\ref{sec:ground completion}---this change is
valid (as shown by Sternagel and Thiemann~\cite{ST13}).

Concerning notation, we write $(\EE,\RR) \vdf (\EE',\RR')$ whenever we
can obtain $(\EE',\RR')$ from $(\EE,\RR)$ by applying one of the
inference rules of \defref{KBf}.
While it is well-known that applying the inference rules of \KBf does not
affect the equational theory induced by $\EE \cup \RR$,
our formulation is new and paves the way for a simple correctness proof.

\begin{lem}%
\label{lem:KB rewrite steps}
Suppose $(\EE,\RR) \vdf (\EE',\RR')$. Then, the following two inclusions
hold:
\begin{enumerate}
\item%
\label{lem:KB rewrite steps:a}
If $s \xrightarrow[\ER]{} t$ then
$s \xrightarrow[\RR']{=} \cdot \xrightarrow[\EE' \scup \RR']{=}
\cdot \xleftarrow[\RR']{=} t$.
\hfill\isaforlink{Abstract_Completion}{lem:KB_subset}
\item%
\label{lem:KB rewrite steps:b}
If $s \xrightarrow[\EE' \scup \RR']{} t$ then
$s \xleftrightarrow[\EE \scup \RR]{*} t$.
\hfill\isaforlink{Abstract_Completion}{lem:KB_subset'}
\end{enumerate}
\end{lem}
\begin{proof}
By inspecting the inference rules of \KBf we easily obtain the
following inclusions:
\begin{xalignat*}{2}
\intertext{\tsfs{deduce}} \\[-3ex]
\EE \cup \RR ~&\subseteq~
\EE' \cup \RR'
&
\EE' \cup \RR' ~&\subseteq~
\EE \cup \RR \cup {\xleftarrow[\RR]{} \cdot \xrightarrow[\RR]{}}
\intertext{\tsfs{orient}} \\[-3ex]
\EE \cup \RR ~&\subseteq~
\EE' \cup \RR' \cup {(\RR')}^{-1}
&
\EE' \cup \RR' ~&\subseteq~
\EE \cup \RR \cup \EE^{-1}
\intertext{\tsfs{delete}} \\[-3ex]
\EE \cup \RR ~&\subseteq~
\EE' \cup \RR' \cup {=}
&
\EE' \cup \RR' ~&\subseteq~
\EE \cup \RR
\intertext{\tsfs{compose}} \\[-3ex]
\EE \cup \RR ~&\subseteq~
\EE' \cup \RR' \cup {\xrightarrow[\RR']{} \cdot \xleftarrow[\RR']{}}
&
\EE' \cup \RR' ~&\subseteq~
\EE \cup \RR \cup {\xrightarrow[\RR]{} \cdot \xrightarrow[\RR]{}}
\intertext{\tsfs{simplify}} \\[-3ex]
\EE \cup \RR ~&\subseteq~
\EE' \cup \RR' \cup {\xrightarrow[\RR']{} \cdot \xrightarrow[\EE']{}}
\cup {\xrightarrow[\EE']{} \cdot \xleftarrow[\RR']{}}
&
\EE' \cup \RR' ~&\subseteq~
\EE \cup \RR \cup {\xleftarrow[\RR]{} \cdot \xrightarrow[\EE]{}}
\cup {\xrightarrow[\EE]{} \cdot \xrightarrow[\RR]{}}
\intertext{\tsfs{collapse}} \\[-3ex]
\EE \cup \RR ~&\subseteq~
\EE' \cup \RR' \cup {\xrightarrow[\RR']{} \cdot \xrightarrow[\EE']{}}
&
\EE' \cup \RR' ~&\subseteq~
\EE \cup \RR \cup {\xleftarrow[\RR]{} \cdot \xrightarrow[\RR]{}}
\end{xalignat*}
Consider for instance the \tsfs{collapse} rule and suppose that
$s \approx t \in \EE \cup \RR$. If $s \approx t \in \EE$ then
$s \approx t \in \EE'$ because $\EE \subseteq \EE'$. If
$s \approx t \in \RR$ then either $s \approx t \in \RR'$ or
$s \to_\RR u$ with $u \approx t \in \EE'$ and thus
$s \to_{\RR'} \cdot \to_{\EE'} t$. This proves the inclusion on the left.
For the inclusion on the right the reasoning is similar. Suppose that
$s \approx t \in \EE' \cup \RR'$. If $s \approx t \in \RR'$ then
$s \approx t \in \RR$ because $\RR' \subseteq \RR$. If
$s \approx t \in \EE'$ then either $s \approx t \in \EE$ or
there exists a rule $u \to t \in \RR$ with $u \to_\RR s$ and thus
$s \FromB{\RR} \cdot \to_\RR t$.

Since rewrite relations are closed under contexts and substitutions,
the inclusions in the right column prove statement~\eqref{lem:KB rewrite
steps:b}. Moreover note that each
inclusion in the left column is a special case of
\[
\EE \cup \RR ~\subseteq~
{\xrightarrow[\RR']{=} \cdot \xrightarrow[\EE' \scup \RR']{=}
\cdot \xleftarrow[\RR']{=}}
\]
and thus also statement~\eqref{lem:KB rewrite steps:a} follows from
closure under contexts and substitutions of rewrite relations.
\end{proof}

\begin{cor}%
\label{cor:KB equational theory}
If $(\EE,\RR) \vdf^* (\EE',\RR')$ then the relations
$\xlr[\ER]{*}$ and $\xlr[\EE' \scup \RR']{*}$ coincide.
\hfill\isaforlink{Abstract_Completion}{lem:KB_conversion}
\end{cor}

The next lemma states that termination of $\RR$ is preserved by
applications of the inference rules of \KBf. It is the final result
in this section whose proof refers to the inference rules.

\begin{lem}%
\label{lem:KB termination}
If $(\EE,\RR) \vdf^* (\EE',\RR')$ and $\RR \subseteq {>}$ then
$\RR' \subseteq {>}$.
\hfill\isaforlink{Abstract_Completion}{lem:KB_rtrancl_rules_subset_less}
\end{lem}
\begin{proof}
We consider a single step $(\EE,\RR) \vdf (\EE',\RR')$. The
statement of the lemma follows by a straightforward induction proof.
Observe that \tsfs{deduce}, \tsfs{delete}, and \tsfs{simplify} do not
change the set of rewrite rules and hence $\RR' = \RR \subseteq {>}$. For
\tsfs{collapse} we have $\RR' \subsetneq \RR \subseteq {>}$. In the case
of \tsfs{orient} we have $\RR' = \RR \cup \{ s \to t \}$ with $s > t$ and
hence $\RR' \subseteq {>}$ follows from the assumption
$\RR \subseteq {>}$. Finally, consider an application of \tsfs{compose}.
So $\RR = \RR'' \uplus \{ s \to t \}$ and
$\RR' = \RR'' \cup \{ s \to u \}$ with $t \to_\RR u$. We obtain $s > t$
from the assumption $\RR \subseteq {>}$. Since $>$ is a reduction order,
$t > u$ follows from $t \to_\RR u$. Transitivity of $>$ yields $s > u$ and
hence $\RR' \subseteq {>}$ as desired.
\end{proof}

To guarantee that the result of a finite \KBf derivation is a
complete TRS equivalent to the initial $\EE$, \KBf derivations must
satisfy the \emph{fairness condition} that prime
critical pairs of the final TRS $\RRn$ which were not considered during
the derivation are joinable in $\RRn$.

\begin{defi}[Finite Runs and Fairness]%
\label{def:KBf fairness}
A \emph{finite run} for a given ES $\EE$ is a finite sequence
\[
\EE_0,\RR_0
~\vdf~ \EE_1,\RR_1
~\vdf~ \cdots
~\vdf~ \EEn,\RRn
\]
such that $\EE_0 = \EE$ and $\RR_0 = \varnothing$.
The run is \emph{fair} if $\EEn = \varnothing$ and
\[
\PCP(\RRn) ~\subseteq~ {\join_{\RRn}} \cup
\bigcup_{i=0}^n {\xlr[\EE_i]{}}
\]
\end{defi}

The reason for writing $\fromto_{\EE_i}$ instead of $\EE_i$ in the
definition of fairness is that critical pairs are ordered, so in a fair
run a (prime) critical pair $s \approx t$ of $\RRn$ may be ignored by
\tsfs{deduce} if $t \approx s$ was generated, or more
generally, if $s \fromto_{\EE_i} t$ holds at some point in the run.
Non-prime critical pairs can always be ignored.
Note that our fairness condition differs from earlier notions by
permitting that (prime) critical pairs may be joinable in $\RRn$.
This was done to allow for more flexibility in
implementations. Our proofs smoothly extend to the relaxed
condition.

According to the main result of this section (\thmref{KBf correctness}),
a completion procedure that produces fair runs is correct. The challenge
is the confluence proof of $\RRn$. We show that $\RRn$ is peak decreasing
by labeling rewrite steps (not only in $\RR_n$) with multisets of terms.
As well-founded order on these multisets we take the multiset extension
$\MUL$ of the given reduction order $>$.

\begin{defi}[Labeled Rewriting
\isaforlink{Abstract_Completion}{def:mstep}]%
\label{def:mset labeled rewriting}
Let $\to$ be a rewrite relation and $M$ a finite multiset of terms. We
write \smash{$s \xrightarrow{M} t$} if $s \to t$ and there exist terms
$s', t' \in M$
such that $s' \geqslant s$ and $t' \geqslant t$. Here $\geqslant$ denotes
the reflexive closure of the given reduction order $>$.
\end{defi}

Since both $\to$ and $\geqslant$ are closed
under contexts and substitutions, we have
\text{$C[t\sigma] \xrs{M'} C[u\sigma]$} whenever $t \xrs{M} u$ and
$M' = \{ C[s\sigma] \mid s \in M \}$,
for all contexts $C$ and substitutions $\sigma$.

\begin{lem}%
\label{lem:KB key}
Let $(\EE,\RR) \vdf (\EE',\RR')$. If
$\smash{t \CCC{\ER}{M}{*} u}$ and $\RR' \subseteq {>}$ then
$\smash{t \CCC{\EE' \scup \RR'}{M}{*} u}$.
\hfill
\isaforlink{Abstract_Completion}{lem:msteps_subset}
\end{lem}
\begin{proof}
We consider a single $(\EE \cup \RR)$-step from $t$ to $u$.
The lemma follows then by induction on the length of the
conversion between $t$ and $u$. According to \lemref[a]{KB rewrite steps}
there exist terms $v$ and $w$ such that
\[
t \xrightarrow[\RR']{=}
v \xrightarrow[\EE' \scup \RR']{=}
w \xleftarrow[\RR']{=} u
\]
We claim that the (non-empty) steps can be labeled by $M$. There exist
terms $t', u' \in M$ with $t' \geqslant t$ and
$u' \geqslant u$. Since $\RR' \subseteq {>}$, we have
$t \geqslant v$ and $u \geqslant w$ and thus also
$t' \geqslant v$ and $u' \geqslant w$. Hence
\[
t \CCR{\RR'}{M}{=}
v \CCR{\EE' \scup \RR'}{M}{=}
w \CCL{\RR'}{M}{=} u
\]
and thus also $t \CCC{\EE' \scup \RR'}{M}{*} u$.
\end{proof}

\begin{thm}%
\label{thm:KBf correctness}
For every fair run $\Gamma$
\hfill
\isaforlink{Completion_Fairness}{lem:finite_fair_new_run}
\[
\EE_0,\RR_0
~\vdf~ \EE_1,\RR_1
~\vdf~ \cdots
~\vdf~ \EEn,\RRn
\]
the TRS $\RRn$ is a complete presentation of $\EE$.
\end{thm}
\begin{proof}
We have $\EEn = \varnothing$. From \corref{KB equational theory} we know
that ${\conv_\EE} = {\conv_\RRn}$. \lemref{KB termination} yields
$\RRn \subseteq {>}$ and hence $\RRn$ is terminating. It remains to
prove that $\RRn$ is confluent.
Let
\[
t \xleftarrow[\RR_n]{M_1} s \xrightarrow[\RR_n]{M_2} u
\]
be a labeled local peak in $\RR_n$.
From \lemref{pcp} we obtain $t \mathrel{\tds^2} u$.
Let $v \mathrel{\tds} w$ appear in this sequence
(so $t = v$ or $w = u$). We obtain
\[
(v,w) \in {\join_{\RR_n}} \cup \bigcup_{i=0}^n {\xlr[\EE_i]{}}
\]
from the definition of $\tds$ and fairness of $\Gamma$. We label all steps
between $v$ and $w$ with the multiset $\{ v, w \}$. Because $s > v$ and
$s > w$ we have $M_1 \MUL \{ v, w \}$ and $M_2 \MUL \{ v, w \}$. Hence by
repeated applications of \lemref{KB key} we obtain a conversion
in $\RR_n$ between $v$ and $w$ in which each step is labeled with a
multiset that is smaller than both $M_1$ and $M_2$. It follows that
$\RR_n$ is peak decreasing and thus confluent by \lemref{pd => cr}.
\end{proof}

A completion procedure is a program that generates \KBf runs. In
order to ensure that the final outcome $\RR_n$ is a complete presentation
of the initial ES, fair runs should be produced. Fairness
requires that prime critical pairs of $\RR_n$ are considered during the
run. Of course, $\RR_n$ is not known during the run, so to be on the safe
side, prime critical pairs of any $\RR$ that appears during the run should
be generated by \tsfs{deduce}.
In particular, there is no need to deduce equations that
are not prime critical pairs. So we may strengthen the condition
$s \FromB{\RR} \cdot \to_\RR t$ of \tsfs{deduce} to
$s \approx t \in \PCP(\RR)$ without affecting \thmref{KBf correctness}.

The following example shows that the success of a run may depend on
the order in which inference rules are applied~\cite{BDP89}.

\begin{exa}%
\label{exa:strategy}
Consider the ES $\EE$ consisting of the four equations
\begin{xalignat*}{4}
\m{a} &\approx \m{b} &
\m{a} &\approx \m{c} &
\m{f}(\m{b}) &\approx \m{b} &
\m{f}(\m{a}) &\approx \m{d}
\end{xalignat*}
and the reduction order $\LPO$ with the partial precedence
$\m{a} > \m{b} > \m{d}$ and $\m{a} > \m{c} > \m{d}$
but where $\m{b}$ and $\m{c}$ are incomparable.
One possible run is
\begin{alignat*}{2}
(\EE,\varnothing) ~
& \mathrel{{\vdfr{orient}}^+} &&
(\{ \m{a} \approx \m{c}, \m{f}(\m{a}) \approx \m{d} \},
 \{ \m{a} \to \m{b},\m{f}(\m{b}) \to \m{b} \}) \\
& \mathrel{{\vdfr{simplify}}^+} &&
(\{ \m{b} \approx \m{c}, \m{f}(\m{b}) \approx \m{d} \},
 \{ \m{a} \to \m{b}, \m{f}(\m{b}) \to \m{b} \}) \\
& \vdfr{simplify} &&
(\{ \m{b} \approx \m{c}, \m{b} \approx \m{d} \},
 \{ \m{a} \to \m{b}, \m{f}(\m{b}) \to \m{b} \}) \\
& \vdfr{orient} &&
(\{ \m{b} \approx \m{c} \},
 \{ \m{a} \to \m{b}, \m{f}(\m{b}) \to \m{b}, \m{b} \to \m{d} \}) \\
& \vdfr{collapse} &&
(\{ \m{b} \approx \m{c}, \m{f}(\m{d}) \approx \m{b} \},
 \{ \m{a} \to \m{b}, \m{b} \to \m{d} \}) \\
& \mathrel{{\vdfr{simplify}}^+} &~&
(\{ \m{d} \approx \m{c}, \m{f}(\m{d}) \approx \m{d} \},
 \{ \m{a} \to \m{b}, \m{b} \to \m{d} \}) \\
& \mathrel{{\vdfr{orient}}^+} &&
(\varnothing,
 \{ \m{a} \to \m{b}, \m{b} \to \m{d}, \m{c} \to \m{d},
    \m{f}(\m{d}) \to \m{d}
\})
\end{alignat*}
which derives a complete presentation of $\EE$. However, the run
\begin{alignat*}{2}
(\EE,\varnothing) ~
& \vdfr{orient} &&
(\{ \m{a} \approx \m{b}, \m{f}(\m{b}) \approx \m{b},
    \m{f}(\m{a}) \approx \m{d} \},
 \{ \m{a} \to \m{c} \}) \\
& \mathrel{{\vdfr{simplify}}^+} &~&
(\{ \m{c} \approx \m{b}, \m{f}(\m{b}) \approx \m{b},
    \m{f}(\m{c}) \approx \m{d} \},
 \{ \m{a} \to \m{c} \}) \\
& \mathrel{{\vdfr{orient}}^+} &&
(\{ \m{c} \approx \m{b} \},
 \{ \m{a} \to \m{c}, \m{f}(\m{b}) \to \m{b}, \m{f}(\m{c}) \to \m{d} \})
\end{alignat*}
cannot be extended to a successful one because the equation
$\m{c} \approx \m{b}$ cannot be oriented.
\end{exa}

The following example shows that even after a \KBf run derived a
complete system, exponentially many steps might be performed to obtain a
canonical TRS\@.

\begin{exa}
Consider the ESs
$\EE_n = \{ \m{f}(\m{g}^i(\m{c})) \approx \m{g}(\m{f}^i(\m{c})) \mid
0 \leqslant i \leqslant n \}$ for $n \geqslant 1$.
By taking the Knuth-Bendix order $\KBO$ with precedence $\m{f} > \m{g}$
and where $w(\m{f}) = w(\m{g})$, all equations can be oriented from left
to right. Since there are no critical pairs, the resulting TRSs
$\RR_n = \{ \m{f}(\m{g}^i(\m{c})) \to \m{g}(\m{f}^i(\m{c})) \mid
0 \leqslant i \leqslant n \}$
are complete by \thmref{KBf correctness}.
However, it is not canonical since right-hand sides are not
normal forms. When applying \tsfs{compose} steps in a naive way by
simplifying the rules in descending order, exponentially many steps are
required to obtain a canonical system~\cite{PSK96}. However, when
processing the rules in reverse order only a polynomial number of steps
is necessary.
\end{exa}

This section resumes our results on finite runs~\cite{HMS14}.
The presented correctness proof differs substantially from all earlier
proofs in that it does not rely on a proof order~\cite{BDH86} but is
instead based on peak decreasingness. It supports a relaxed side
condition of the \textsf{collapse} rule as first used in~\cite{ST13},
but in contrast to the latter demands only prime critical pairs to
be considered.

\section{Canonicity and Normalization Equivalence}%
\label{sec:canonicity}

A natural question arising in the context of completion concerns
uniqueness of resulting systems: Is there a single
complete presentation of a given equational theory conforming to a
certain reduction order? M\'etivier~\cite{M83} showed that for reduced
and hence canonical systems this is indeed the case, up to renaming
variables. In this section we revisit his work, aiming at
generalizing his uniqueness result for canonical TRSs and at establishing
a transformation to simplify ground-complete TRSs. A key notion
to that end is normalization equivalence.

\begin{defi}[Conversion/Normalization Equivalence]
Two ARSs $\AA$ and $\BB$ are said to be \emph{(conversion) equivalent} if
${\conv_\AA} = {\conv_\BB}$ and \emph{normalization equivalent} if
${\to_\AA^!} = {\to_\BB^!}$.
\end{defi}

The following example shows that these two equivalence notions do not
coincide.

\begin{exa}%
\label{equivalence example}
Consider the four ARSs:
\begin{xalignat*}{2}
&\begin{tikzpicture}[on grid,node distance=12mm,baseline=(1).baseline]
\node (0)              {$\AA_1\colon\mathstrut$};
\node (1) [right=of 0] {$\e{a}$};
\node (2) [right=of 1] {$\e{b}$};
\draw[->] (1) -- (2);
\end{tikzpicture}
&
&\begin{tikzpicture}[on grid,node distance=12mm,baseline=(1).baseline]
\node (0)              {$\BB_1\colon\mathstrut$};
\node (1) [right=of 0] {$\e{a}$};
\node (2) [right=of 1] {$\e{b}$};
\draw[->] (2) -- (1);
\end{tikzpicture}
\\
&\begin{tikzpicture}[on grid,node distance=12mm,baseline=(1).baseline]
\node (0)              {$\AA_2\colon\mathstrut$};
\node (1) [right=of 0] {$\e{a}$};
\node (2) [right=of 1] {$\e{b}$};
\draw[->] (1) -- (2);
\draw[->] (2) edge[out=295,in=245,loop] ();
\end{tikzpicture}
&
&\begin{tikzpicture}[on grid,node distance=12mm,baseline=(1).baseline]
\node (0)              {$\BB_2\colon\mathstrut$};
\node (1) [right=of 0] {$\e{a}$};
\node (2) [right=of 1] {$\e{b}$};
\draw[->] (1) edge[out=295,in=245,loop] ();
\draw[->] (2) edge[out=295,in=245,loop] ();
\end{tikzpicture}
\\[-3ex]
\end{xalignat*}
While $\AA_1$ and $\BB_1$ are conversion equivalent but not normalization
equivalent, the ARSs $\AA_2$ and $\BB_2$ are normalization equivalent but
not conversion equivalent.
\end{exa}

The easy proof (by induction on the length of conversions) of the
following result is omitted.

\begin{lem}%
\label{lem:(normalization) equivalence}
Normalization equivalent terminating ARSs are equivalent.
\hfill\isaforlink{Normalization_Equivalence}{lem:WN_NE_imp_conv_eq}
\end{lem}

Note that the termination assumption can be weakened to weak
normalization. However, the present version suffices to prove the
following lemma that we employ in our proof of M\'etivier's transformation
result~\cite{M83} (\thmref{Metivier normalization equivalent} below).

\begin{lem}%
\label{lem:reduced abstract}
Let $\AA$ and $\BB$ be ARSs such that $\NF(\BB) \subseteq \NF(\AA)$ and
either
\begin{enumerate}
\item%
\label{lem:reduced abstract:a}
${\to_\BB} \subseteq {\to_\AA^+}$ or
\hfill\isaforlink{Normalization_Equivalence}{lem:complete_NE_intro}

\smallskip
\item%
\label{lem:reduced abstract:b}
${\to_\BB} \subseteq {\conv_\AA}$ and $\BB$ is terminating.
\hfill\isaforlink{Normalization_Equivalence}{lem:complete_NE_intro1}
\end{enumerate}
\smallskip
If $\AA$ is complete then $\BB$ is complete and normalization equivalent
to $\AA$.
\end{lem}

\begin{proof}
We first show ${\to_\BB^!} \subseteq {\to_\AA^!}$. In
case~\eqref{lem:reduced
abstract:a}, from the
inclusion ${\to_\BB} \subseteq {\to_\AA^+}$ we infer that $\BB$ is
terminating. Moreover, ${\to_\BB^*} \subseteq {\to_\AA^*}$ and, since
$\NF(\BB) \subseteq \NF(\AA)$, also ${\to_\BB^!} \subseteq {\to_\AA^!}$.
For case~\eqref{lem:reduced abstract:b},
${\to_\BB^!} \subseteq {\to_\AA^!}$ holds because
${\to_\BB^!} \subseteq {\conv_\AA}$, so by confluence of $\AA$ and
$\NF(\BB) \subseteq \NF(\AA)$ we obtain
${\to_\BB^!} \subseteq {\to_\AA^!}$. Next we show that the reverse
inclusion ${\to_\AA^!} \subseteq {\to_\BB^!}$ holds in both cases. Let
$a \to_\AA^! b$. Because $\BB$ is terminating, $a \to_\BB^! c$ for some
$c \in \NF(\BB)$. So $a \to_\AA^! c$ and thus $b = c$ from the confluence
of $\AA$. It follows that $\AA$ and $\BB$ are normalization equivalent. It
remains to show that $\BB$ is locally confluent. This follows from the
sequence of inclusions
\[
{\FromB{\BB} \cdot \to_\BB}
~\subseteq~ {\conv_\AA}
~\subseteq~ {\to_\AA^! \cdot \From{\AA}{!}}
~\subseteq~ {\to_\BB^! \cdot \From{\BB}{!}}
\]
where we obtain the inclusions from ${\to_\BB} \subseteq {\conv_\AA}$,
confluence of $\AA$, termination of $\AA$, and normalization equivalence
of $\AA$ and $\BB$, respectively.
\end{proof}

In the above lemma, completeness can be weakened to semi-completeness
(\ie, the combination of confluence and weak normalization), which is not
true for \thmref{Metivier normalization equivalent} as shown by
Gramlich~\cite{G01}. Again, the present version suffices for our purposes.
Condition~\eqref{lem:reduced abstract:b} of the lemma
can be regarded as a specialization of an abstract result of
Toyama~\cite[Corollary~3.2]{T91} to complete systems
and will be used in \secref{ordered completion}.

\thmref{Metivier normalization equivalent} below shows that we can
always eliminate redundancy in a complete TRS\@. This is achieved by the
following two-stage transformation, where,
given a TRS $\RR$, we write $\RR_{\doteq}$ for a set of representatives of
the equivalence classes of rules in $\RR$ with respect to $\doteq$
(\ie, $\RR_{\doteq}$ is a variant-free version of $\RR$).

\begin{defi}%
\label{def:reduced transformation}
Given a terminating TRS $\RR$, the TRSs $\dot{\RR}$ and $\ddot{\RR}$ are
defined as follows:
\begin{align*}
\dot{\RR} &= {\{ \ell \to r{\downarrow_\RR} \mid
\ell \to r \in \RR \}}_{\doteq}
\tag*{\isaforlink{Normalization_Equivalence}{def:dot}} 
\\
\ddot{\RR} &= \{ \ell \to r \in \dot{\RR} \mid
\ell \in \NF(\dot{\RR} \setminus \{ \ell \to r \}) \}
\tag*{\isaforlink{Normalization_Equivalence}{def:ddot}} 
\end{align*}
Here $t{\downarrow_\RR}$ stands for an arbitrary but fixed normal form of
$t$.
\end{defi}

The TRS $\dot{\RR}$ is obtained from $\RR$ by first normalizing the
right-hand sides and then taking representatives of variants of the
resulting rules, thereby making sure that the result does not contain
several variants of the same rule. To obtain $\ddot{\RR}$ we remove the
rules of $\dot{\RR}$ whose left-hand sides are reducible with another rule
of $\dot{\RR}$. (This is the only place in the paper where
variant-freeness of TRSs is important.)

The following example shows why the result of $\dot{\RR}$ has to be
variant-free.

\begin{exa}%
\label{ex:2}
Consider the TRS $\RR$ consisting of the four rules
\begin{xalignat*}{4}
\m{f}(x) &\to \m{a} &
\m{f}(y) &\to \m{b} &
\m{a} &\to \m{c} &
\m{b} &\to \m{c}
\end{xalignat*}
Then the first transformation without taking representatives of rules
would yield $\dot{\RR}$
\begin{xalignat*}{4}
\m{f}(x) &\to \m{c} &
\m{f}(y) &\to \m{c} &
\m{a} &\to \m{c} &
\m{b} &\to \m{c}
\end{xalignat*}
and the second one $\ddot{\RR}$
\begin{xalignat*}{2}
\m{a} &\to \m{c} &
\m{b} &\to \m{c}
\end{xalignat*}
Note that $\ddot{\RR}$ is \emph{not} equivalent to $\RR$. This is
caused by the fact that the result of the first transformation was no
longer variant-free.
\end{exa}

The following result, due to M\'etivier~\cite[Theorem~7]{M83}, allows us
to obtain a canonical representation of any complete TRS.\footnote{%
We were not able to reconstruct enough detail for an Isabelle/HOL
formalization from its original proof. Another textbook proof~\cite[Exercise~7.4.7]{Terese} involves 13 steps with lots of redundancy.}
Our proof below proceeds by induction on the well-founded
encompassment order $\prencompasses$.

\begin{thm}%
\label{thm:Metivier normalization equivalent}
If $\RR$ is a complete TRS then $\ddot{\RR}$ is a normalization and
conversion equivalent canonical TRS\@.
\hfill\isaforlink{Normalization_Equivalence}{lem:canonical_NE_conv_eq}
\end{thm}

\begin{proof}
Let $\RR$ be a complete TRS\@. The inclusions
$\ddot{\RR} \subseteq \dot{\RR} \subseteq {\to_\RR^+}$ are obvious from
the definitions. Since $\RR$ and $\dot{\RR}$ have the same left-hand
sides, their normal forms coincide. We show that
$\NF(\ddot{\RR}) \subseteq \NF(\dot{\RR})$. To this end we show that
$\ell \notin \NF(\ddot{\RR})$ whenever $\ell \to r \in \dot{\RR}$ by
induction on $\ell$ with respect to the well-founded order
$\prencompasses$. If $\ell \to r \in \ddot{\RR}$ then
$\ell \notin \NF(\ddot{\RR})$ holds. So suppose
$\ell \to r \notin \ddot{\RR}$. By definition of $\ddot{\RR}$,
$\ell \notin \NF(\dot{\RR} \setminus \{ \ell \to r \})$. So there exists a
rewrite rule $\ell' \to r' \in \dot{\RR}$ different from $\ell \to r$ such
that $\ell \encompasses \ell'$. We distinguish two cases.
\begin{itemize}
\item
If $\ell \prencompasses \ell'$ then we obtain
$\ell' \notin \NF(\ddot{\RR})$ from the induction hypothesis and hence
$\ell \notin \NF(\ddot{\RR})$ as desired.
\item
If $\ell \doteq \ell'$ then by \lemref{variants-terms} there exists a
renaming $\sigma$ such that $\ell = \ell'\sigma$. Since $\dot{\RR}$ is
right-reduced by construction, $r$ and $r'$ are normal forms of
$\dot{\RR}$. The same holds for $r'\sigma$ because normal forms are closed
under renaming. We have
$r \FromB{\sdot{\RR}} \ell = \ell'\sigma \to_{\sdot{\RR}} r'\sigma$. Since
$\dot{\RR}$ is confluent as a consequence of \lemref[a]{reduced abstract},
$r = r'\sigma$. Hence $\ell' \to r'$ is a variant of $\ell \to r$,
contradicting the fact that $\dot{\RR}$ is variant-free (by construction).
\end{itemize}
From \lemref[a]{reduced abstract} we infer that the TRSs $\dot{\RR}$ and
$\ddot{\RR}$ are complete and normalization equivalent to $\RR$. The TRS
$\ddot{\RR}$ is right-reduced because $\ddot{\RR} \subseteq \dot{\RR}$ and
$\dot{\RR}$ is right-reduced. From $\NF(\ddot{\RR}) = \NF(\dot{\RR})$ we
easily infer that $\ddot{\RR}$ is left-reduced. It follows that
$\ddot{\RR}$ is canonical. It remains to show that $\ddot{\RR}$ is not
only normalization equivalent but also (conversion) equivalent to $\RR$.
This is an immediate consequence of \lemref{(normalization) equivalence}.
\end{proof}

Before we proceed to show uniqueness of normalization equivalent TRSs, we
need the following technical lemma.

\begin{lem}%
\label{lem:right-reduced minimal}
Let $\RR$ be a right-reduced TRS and let $s$ be a reducible term which is
minimal with respect to $\prencompasses$. If $s \to^+_\RR t$ then
$s \to t$ is a variant of a rule in $\RR$.
\hfill
\isaforlink{Normalization_Equivalence}{lem:right_reduced_min_step_rule}
\end{lem}

\begin{proof}
Let $\ell \to r$ be the rewrite rule that is used in the first step from
$s$ to $t$. So $s \encompasses \ell$. By assumption,
$s \prencompasses \ell$ does not hold and thus $s \doteq \ell$
because ${\encompasses} = {\prencompasses} \cup {\doteq}$.
According
to \lemref{variants-terms} there exists a renaming $\sigma$ such that
$s = \ell\sigma$. We have $s \to_\RR r\sigma \to_\RR^* t$. Because $\RR$ is
right-reduced, $r \in \NF(\RR)$. Since normal forms are closed under
renaming, also $r\sigma \in \NF(\RR)$ and thus $r\sigma = t$. It follows
that $s \to t$ is a variant of $\ell \to r$.
\end{proof}

In our formalization, the above proof is the first spot of this section
where we actually need that $\RR$ satisfies the variable condition (more
precisely, only the part of it that right-hand sides of rules do not
introduce fresh variables). We are now in a position to present
the main result of this section.

\begin{thm}%
\label{thm:Metivier 2a}
Normalization equivalent reduced TRSs are unique up to literal similarity.
\hfill
\isaforlink{Normalization_Equivalence}{lem:reduced_NE_imp_unique}
\end{thm}

\begin{proof}
Let $\RR$ and $\SS$ be normalization equivalent reduced TRSs. Suppose
$\ell \to r \in \RR$. Because $\RR$ is right-reduced, $r \in \NF(\RR)$ and
thus $\ell \neq r$. Hence $\ell \to_\SS^+ r$ by normalization equivalence.
Because $\RR$ is left-reduced, $\ell$ is a minimal (with respect to
$\prencompasses$) $\RR$-reducible term. Another application of
normalization equivalence yields that $\ell$ is minimal $\SS$-reducible.
Hence $\ell \to r$ is a variant of a rule in $\SS$ by
\lemref{right-reduced minimal}.
\end{proof}

\begin{exa}
Consider the rewrite system $\RR$ of combinatory logic with equality test,
studied by Klop~\cite{K80}:
\begin{xalignat*}{2}
@(@(@(\m{S},x),y), z) & \to @(x,@(z,@(y,z))) &
@(@(\m{K},x),y) & \to x \\
@(\m{I},x) & \to x &
@(@(\m{D},x),x) & \to \m{E}
\end{xalignat*}
The rewrite system $\RR$ is reduced, but neither terminating nor
confluent. One might ask whether there is another reduced rewrite system
that computes the same normal forms for every starting term.
\thmref{Metivier 2a} shows that $\RR$ is unique up to variable
renaming.
\end{exa}

As the final result of this section, we prove this
result of M\'etivier~\cite[Theorem~8]{M83} to be
an easy consequence of \thmref{Metivier 2a}. Here a TRS $\RR$ is
said to be compatible with a reduction order $>$ if $\ell > r$ for every
rewrite rule $\ell \to r$ of $\RR$.

\begin{thm}%
\label{thm:Metivier 2b}
Let $\RR$ and $\SS$ be equivalent canonical TRSs. If $\RR$ and $\SS$
are compatible with the same reduction order then $\RR \doteq \SS$.
\hfill
\isaforlink{Normalization_Equivalence}{lem:EQ_imp_litsim}
\end{thm}

\begin{proof}
Suppose $\RR$ and $\SS$ are compatible with the reduction order $>$.
We show that ${\to_\RR^!} \subseteq {\to_\SS^!}$.
Let $s \to_\RR^! t$. We show that $t \in \NF(\SS)$. Let $u$ be the unique
$\SS$-normal form of $t$. We have $t \to_\SS^! u$ and thus $t \conv_\RR u$
because $\RR$ and $\SS$ are equivalent. Since $t \in \NF(\RR)$, we have
$u \to_\RR^! t$. If $t \neq u$ then both $t > u$ (as $t \to_\SS^! u$) and
$u > t$ (as $u \to_\RR^! t$), which is impossible. Hence $t = u$ and thus
$t \in \NF(\SS)$. Together with $s \conv_\SS t$, which follows from the
equivalence of $\RR$ and $\SS$, we conclude that $s \to_\SS^! t$.
We obtain ${\to_\SS^!} \subseteq {\to_\RR^!}$ by symmetry. Hence
$\RR$ and $\SS$ are normalization equivalent and the result follows from
\thmref{Metivier 2a}.
\end{proof}

This section resumes our results on canonicity~\cite{HMSW17}.
While the results of \thmref{Metivier normalization equivalent} and
\thmref{Metivier 2b} are due to M{\'e}tivier\cite{M83}, we present novel
and simpler proofs based on the (new) auxiliary results
\lemref{reduced abstract} and \thmref{Metivier 2a}.

\section{Ground Completion}%
\label{sec:ground completion}

In this section we focus on the special case of ground equations, \ie,
equations where both sides are ground terms.

\begin{defi}[Ground Completion
\isaforlink{Ground_Completion}{ind:gkb}]%
\label{def:KBg}
The inference system \KBg consists of the inference
rules of \KBf except for \tsfs{deduce}.
\end{defi}

Snyder~\cite{S93} proved that sets of ground equations can always be
completed by \KBg, provided a \emph{ground-total} reduction order $>$
is used, \ie, for all ground terms $s, t \in \TT(\FF)$ either
$s > t$, $t > s$, or $s = t$. He further proved that every reduced ground
rewrite system is canonical and can be obtained by completion from any
equivalent set of ground equations. Below, we present the proofs
of these results that we formalized in Isabelle/HOL\@.

The following example illustrates the inference system \KBg
on a set of ground equations.

\begin{exa}
Consider the ES $\EE$ consisting of the ground equations
\begin{align*}
\m{f}(\m{f}(\m{f}(\m{a}))) &\approx \m{f}(\m{b}) &
\m{f}(\m{f}(\m{b})) &\approx \m{c} &
\m{f}(\m{c}) &\approx \m{a} &
\m{f}(\m{a}) &\approx \m{f}(\m{f}(\m{b}))
\intertext{As reduction order we take LPO induced by the total precedence
$\m{a} > \m{b} > \m{c} > \m{f}$. We start by applying \tsfs{orient}
to the last two equations:}
\m{f}(\m{f}(\m{f}(\m{a}))) &\approx \m{f}(\m{b}) &
\m{f}(\m{f}(\m{b})) &\approx \m{c} &
\m{f}(\m{c}) &\from \m{a} &
\m{f}(\m{a}) &\to \m{f}(\m{f}(\m{b}))
\intertext{An application of \tsfs{collapse} produces}
\m{f}(\m{f}(\m{f}(\m{a}))) &\approx \m{f}(\m{b}) &
\m{f}(\m{f}(\m{b})) &\approx \m{c} &
\m{f}(\m{c}) &\from \m{a} &
\m{f}(\m{f}(\m{c})) &\approx \m{f}(\m{f}(\m{b}))
\intertext{Next we orient the second equation:}
\m{f}(\m{f}(\m{f}(\m{a}))) &\approx \m{f}(\m{b}) &
\m{f}(\m{f}(\m{b})) &\to \m{c} &
\m{f}(\m{c}) &\from \m{a} &
\m{f}(\m{f}(\m{c})) &\approx \m{f}(\m{f}(\m{b}))
\intertext{Two applications of \tsfs{simplify} produce}
\m{f}(\m{f}(\m{f}(\m{f}(\m{c})))) &\approx \m{f}(\m{b}) &
\m{f}(\m{f}(\m{b})) &\to \m{c} &
\m{f}(\m{c}) &\from \m{a} &
\m{f}(\m{f}(\m{c})) &\approx \m{c}
\intertext{We continue by orienting the last equation:}
\m{f}(\m{f}(\m{f}(\m{f}(\m{c})))) &\approx \m{f}(\m{b}) &
\m{f}(\m{f}(\m{b})) &\to \m{c} &
\m{f}(\m{c}) &\from \m{a} &
\m{f}(\m{f}(\m{c})) &\to \m{c}
\intertext{Two applications of \tsfs{simplify} produce}
\m{c} &\approx \m{f}(\m{b}) &
\m{f}(\m{f}(\m{b})) &\to \m{c} &
\m{f}(\m{c}) &\from \m{a} &
\m{f}(\m{f}(\m{c})) &\to \m{c}
\intertext{Orienting the remaining equation followed by a collapse step
produces}
\m{c} &\from \m{f}(\m{b}) &
\m{f}(\m{c}) &\approx \m{c} &
\m{f}(\m{c}) &\from \m{a} &
\m{f}(\m{f}(\m{c})) &\to \m{c}
\intertext{Finally, we orient the only remaining equation and
\tsfs{collapse}, \tsfs{compose}, \tsfs{simplify}, and \tsfs{delete}
exhaustively, thereby obtaining the TRS $\RR$}
\m{c} &\from \m{f}(\m{b}) &
\m{f}(\m{c}) &\to \m{c} &
\m{c} &\from \m{a} &
&
\end{align*}
which constitutes a canonical presentation of $\EE$.
\end{exa}

The absence of \tsfs{deduce} from \KBg does not hurt for ground
systems. If $s \from \cdot \to t$ and the two contracted redexes are at
parallel positions then trivially $s \to \cdot \from t$. If the steps are
identical then $s = t$. In the remaining case one of the contracted
redexes is a subterm of the other contracted redex, and the effect of
\tsfs{deduce} is achieved by the \tsfs{collapse} inference rule.
On the contrary, the absence of \tsfs{deduce} is crucial to conclude
that \KBg derivations are always finite, as illustrated by the following
simple example.

\begin{exa}%
\label{exa:KBg deduce}
Consider the ground ES $\EE$ consisting of the single equation
$\m{a} \approx \m{b}$ and LPO induced by the precedence
$\m{a} > \m{b}$. Using \KBf (i.e., \KBg with \tsfs{deduce})
the following infinite run is possible:
\begin{alignat*}{2}
(\EE,\varnothing) ~
& \vdfr{orient} &&
(\varnothing, \{ \m{a} \to \m{b} \}) \\
& \vdfr{deduce} &~&
(\{ \m{b} \approx \m{b} \}, \{ \m{a} \to \m{b} \}) \\
& \vdfr{delete} &&
(\varnothing, \{ \m{a} \to \m{b} \}) \\
& \vdfr{deduce} && \dots
\end{alignat*}
\end{exa}

\begin{lem}%
\label{KB- termination}
There are no infinite runs
\(
\EE_0,\varnothing \,\vdg\, \EE_1,\RR_1 \,\vdg\, \cdots
\)
for finite ground ES $\EE_0$.
\hfill
\isaforlink{Ground_Completion}{lem:SN_on_GKB}
\end{lem}
\begin{proof}
Let $\succ$ denote the lexicographic combination of the multiset extension
$\MUL$ of the reduction order $>$ with the standard order on natural
numbers $>_\mathbb{N}$. Furthermore let $M(\EE,\RR)$ denote the (finite)
multiset of left-hand sides and right-hand sides occurring in $\EE$ and
$\RR$
\[
M(\EE,\RR) =
\bigcup\,\{ \{ s, t \} \mid (s,t) \in \EE \}
~\cup~
\bigcup\,\{ \{ s, t \} \mid (s,t) \in \RR \}
\]
and consider the function $P$ that maps the pair $(\EE,\RR)$ to
$(M(\EE,\RR), |\EE|)$. Now it is straightforward to verify that any
infinite $\vdg$-sequence would give rise to an infinite
sequence $P(\EE_0,\varnothing) \succ P(\EE_1,\RR_1) \succ \cdots$,
contradicting the well-foundedness of $\succ$.
\end{proof}

\begin{thm}%
\label{thm:KBg correctness}
If $>$ is total on $\EE$-equivalent ground terms then every maximal
\textup{\KBg} run produces an equivalent canonical presentation for
every ground ES $\EE$.
\hfill
\isaforlink{Ground_Completion}{lem:ground_max_run_canonical}
\end{thm}
\begin{proof}
Consider a maximal \KBg run
\(
\EE_0,\varnothing
\,\vdg\, \EE_1,\RR_1
\,\vdg\, \cdots
\,\vdg\, \EE_n,\RR_n
\)
where $\EE_0 = \EE$ is a ground ES\@. Because the run is
maximal, no inference rule of \KBg is applicable to the final pair
$(\EE_n,\RR_n)$. In particular, \tsfs{compose} and \tsfs{collapse} are not
applicable and hence the final TRS $\RR_n$ is reduced.
Since $\RR_n$ is ground, this means in particular that there are no
critical pairs. Moreover, termination of $\RR_n$ follows from
\lemref{KB termination} (since any \KBg run is also a \KBf run), so
$\RR_n$ is canonical.
From \corref{KB equational theory} and the
inclusion $\text{\KBg} \subseteq \text{\KBf}$ we infer that $\EE$ and
$\EE_n \cup \RR_n$ are equivalent. It follows that $>$ is total on
$\EE_n$-equivalent ground terms and thus $\EE_n = \varnothing$, for
otherwise the run could be extended with an application of \tsfs{delete}
or \tsfs{orient}. Hence $\RR_n$ and $\EE$ are equivalent.
\end{proof}

The restriction on the reduction order $>$ in the above correctness
theorem is easy to satisfy. In particular, it holds for any LPO or KBO
based on a total precedence.

Next we consider completeness of ground completion. Our proof makes
use of the following concept.

\begin{defi}[Random Descent
\isaforlink{Ground_Completion}{def:RD}]
An ARS $\AA$ has \emph{random descent} if for every conversion $a \conv b$
with normal form $b$ we have $a \to^n b$ with $n + l = r$. Here $l$ ($r$)
denotes the number of $\from$ ($\to$) steps in the conversion $a \conv b$.
\end{defi}

Random descent is useful in the analysis of rewrite
strategies~\cite{vOT16}. It generalizes a number of earlier concepts,
including the property
${\from \cdot \to} \subseteq ({\to \cdot \from}) \cup {=}$
which is known as \tsfs{WCR1} and holds for left-reduced ground TRSs.
We formalized a new, short and direct proof of the following result due
to van Oostrom~\cite{vO07}.
Here an element $a$ is said to be complete if it is both terminating
(there are no infinite rewrite sequences starting at $a$) and
confluent (if $b \FromT{*} a \to^* c$ then $b \join c$).

\begin{thm}%
\label{thm:random descent}
Let $\AA$ be an ARS with random descent. If $a \conv b$ with normal form
$b$ then $a$ is complete and all rewrite sequences from $a$ to $b$ have
the same length.
\hfill
\isaforlink{Ground_Completion}{lem:RD_NF}
\end{thm}
\begin{proof}
Let $l$ ($r$) be the number of $\from$ ($\to$) steps in the conversion
from $a$ to $b$. We have $l \leqslant r$ since $n + l = r$ for some $n$ by
random descent. First we prove termination of $a$. For a proof by
contradiction, suppose the existence of an infinite rewrite sequence
\[
a = a_0 \to a_1 \to a_2 \to \cdots
\]
Clearly, $a \to^{r-l} a_{r-l}$ and thus there exists a conversion
$a_{r-l} \FromT{*} a \conv b$ with $r$ backwards and $r$ forwards steps.
Hence $a_{r-l} = b$ by another application of random descent and
therefore $b \to a_{r-l+1}$, contradicting the fact that $b$ is a normal
form. Next we prove confluence of $a$. Suppose $c \FromT{*} a \to^* d$. We
obtain the two conversions $c \conv b$ and $d \conv b$, which are
transformed into $c \join d$ by two applications of random descent.
Finally, assume there are two rewrite sequences $a \to^m b$ and
$a \to^n b$ from $a$ to $b$ of length $m$ and $n$. Reversing the first
sequence and appending the second one yields a conversion $b \conv b$ with
$m$ backwards and $n$ forwards steps. A final application of random
descent yields $b \to^k b$ for some $k$ with $k + m = n$. Since $b$ is a
normal form, $k = 0$ and thus $m = n$ as desired.
\end{proof}

In the series of lemmas below, we establish that reduced ground TRSs
are canonical and have random descent.

\begin{lem}%
\label{lem:left-reduced WCR1}
Left-reduced TRSs enjoy the \tsfs{WCR1} property.
\hfill
\isaforlink{Ground_Completion}{lem:left_reduced_WCR1}
\end{lem}
\begin{proof}
This follows from a straightforward case analysis on the relative positions
of the two redexes that are part of a peak together with the fact that for
left-reduced TRSs the left-hand side alone uniquely determines the employed
rewrite rule.
\end{proof}

\begin{lem}%
\label{lem:left-reduced ground}
Left-reduced ground TRSs have random descent.
\hfill
\isaforlink{Ground_Completion}{lem:left_reduced_ground_RD}
\end{lem}
\begin{proof}
Let $\RR$ be a left-reduced TRS and $s \conv t$ a conversion between two
arbitrary but fixed terms $s$ and $t$ such that $t$ is a normal form.
We proceed by induction on the length of this conversion. If
it is empty or the first step is to the right, we are done. Otherwise, we
have $s \from u \conv t$ where the conversion between $u$ and $t$
has $l$ ($r$) $\from$ ($\to$) steps
and obtain $u \to^k t$ with $k + l = r$ by the induction
hypothesis. The remainder of the proof proceeds by induction on $k$
together with \lemref{left-reduced WCR1}.
\end{proof}

\begin{lem}%
\label{lem:right-reduced ground SN}
Right-reduced ground TRSs are terminating.
\hfill
\isaforlink{Ground_Completion}{lem:right_reduced_ground_SN}
\end{lem}
\begin{proof}
Let $\RR$ be a right-reduced ground TRS\@. For the sake of a contradiction,
assume that $\RR$ is non-terminating. Then there is a minimal
non-terminating term $t$ (that is, all its proper subterms are
terminating). This means that after a finite number of non-root steps
$t \to^* u$ there will be a root step $u \to v$ such that $v$ is
non-terminating. But since $\RR$ is right-reduced and
ground, $v$ is a ground normal form, deriving the desired
contradiction.
\end{proof}

\begin{cor}%
\label{cor:reduced ground}
Reduced ground TRSs are canonical and have random descent.
\hfill
\isaforlink{Ground_Completion}{lem:reduced_ground_RD_and_canonical}
\end{cor}
\begin{proof}
Let $\RR$ be a reduced ground TRS\@. Then, by \lemref{left-reduced ground},
$\RR$ has random descent. Moreover, by \lemref{right-reduced ground SN},
$\RR$ is terminating. Finally, since all terms are $\RR$-terminating,
confluence of $\RR$ is an immediate consequence of the definition of
random descent.
\end{proof}

\begin{thm}%
\label{thm:KBg completeness}
For every ground ES $\EE$ and every equivalent reduced
ground TRS $\RR$ there exist a reduction order $>$ and a derivation
$\EE,\varnothing \,\vdg\, \cdots \,\vdg\, \varnothing,\RR$.
\hfill
\isaforlink{Ground_Completion}{lem:gkb_complete}
\end{thm}
\begin{proof}
Let $>$ be a reduction order that contains $\RR$ and is
total on $\EE$-equivalent ground terms. Consider a maximal
\KBg run starting from $\EE$ and using $>$. According to
\thmref{KBg correctness}, the run produces an equivalent reduced
TRS $\RR'$. Since $\RR \subseteq {>}$ and $\RR' \subseteq {>}$, we obtain
$\RR = \RR'$ from \thmref{Metivier 2b}.
It remains to show that $>$ exists. Let $\sqsupset$ be a total precedence
and define $s > t$ if and only if $s \conv_\EE t$ and either
$d_\RR(s) > d_\RR(t)$ or both $d_\RR(s) = d_\RR(t)$ and
$s \LPO[\sqsupset] t$.\footnote{In the formalization we actually use
$\KBO[\sqsupset]$ with all weights set to $1$, since in contrast to LPO,
for KBO ground-totality for total precedences has already been formalized
before.}
Here $d_\RR(u)$ is the number of rewrite steps in
$\RR$ to normalize the term $u$, which is well-defined since all
normalizing sequences in a reduced ground TRS have the same length
as a consequence of \corref{reduced ground} and
\thmref{random descent}.
It is easy to show that $>$ has the required properties. The only
interesting cases are closure under contexts and substitutions. Both are
basically handled by the following observation:
$d_\RR(C[t\sigma]) = d_\RR(C[t{\downarrow}\sigma]) + d_\RR(t)$ for any
term $t$ (which holds due to random descent together with
termination). This allows us to lift $d_\RR(s) = d_\RR(t)$ and
$d_\RR(s) > d_\RR(t)$ into arbitrary contexts and substitutions.
\end{proof}

The above result cannot be generalized to left-linear right-ground
systems, as shown in the following example due to Dominik Klein
(personal communication).

\begin{exa}
Consider the ES $\EE$ consisting of the two equations
$\m{f}(x) \approx \m{f}(\m{a})$ and $\m{f}(\m{b}) \approx \m{b}$.
Let $>$ be a reduction order. If $\m{f}(\m{b}) > \m{b}$ does not hold, no
inference rule of \KBg\ is applicable to $(\EE,\varnothing)$. If
$\m{f}(\m{b}) > \m{b}$ then the second equation can be oriented
\begin{gather*}
(\EE,\varnothing) ~\vdg~
(\{ \m{f}(x) \approx \m{f}(\m{a}) \},\{ \m{f}(\m{b}) \to \m{b} \})
\end{gather*}
but no further inference steps of \KBg are possible.
Hence completion will fail on $\EE$. Nevertheless, the TRS $\RR$
consisting of the rewrite rule $\m{f}(x) \to \m{b}$ constitutes a
canonical presentation of $\EE$.
\end{exa}

The correctness result of ground completion (\thmref{KBg correctness}) is
due to Snyder~\cite{S93}, and our formalized proof basically follows
his approach. In addition, we present a new completeness proof based on
random descent (\thmref{KBg completeness}).

\section{Correctness for Infinite Runs}%
\label{sec:infinite runs}

Completion as presented in the preceding sections does not always
succeed in producing a finite complete presentation. It may fail because
an unorientable equation is encountered or it may run forever.
In the latter case it is possible that in the limit a possibly
infinite complete presentation is obtained. In this case, completion
can serve as a semi-decision procedure for the validity problem of
the initial equations~\cite{H81}. In this section we give a new proof that
fair infinite runs produce complete presentations of the initial equations,
provided the \tsfs{collapse} rule is restored to its original formulation
(cf.\ \defref{KBi} below).

The reason why this restriction is necessary is provided by
the following example (due to Baader and Nipkow~\cite{BN98}),
which shows that the correctness result (\thmref{KBf correctness}) of
\secref{finite runs} does not extend to infinite runs without further ado.

\begin{exa}%
\label{exa:collapse encompassment}
Consider the ES $\EE$ consisting of the equations
\begin{align*}
\m{a}\m{b}\m{a} &\approx \m{a}\m{b} &
\m{b}\m{b} &\approx \m{b}
\end{align*}
and LPO with precedence $\m{a} > \m{b}$ as reduction order. After two
\tsfs{orient} steps, we apply \tsfs{deduce} to generate the two critical
pairs:
\begin{align*}
\m{a}\m{b}\m{a} &\to \m{a}\m{b} &
\m{b}\m{b} &\to \m{b} &
\m{a}\m{b}\m{a}\m{b} &\approx \m{a}\m{b}\m{b}\m{a} &
\m{b}\m{b} &\approx \m{b}\m{b}
\intertext{The second one is immediately deleted and the first one is
simplified:}
\m{a}\m{b}\m{a} &\to \m{a}\m{b} &
\m{b}\m{b} &\to \m{b} &
\m{a}\m{b}\m{b} &\approx \m{a}\m{b}\m{a}
\intertext{and subsequently oriented:}
\m{a}\m{b}\m{a} &\to \m{a}\m{b} &
\m{b}\m{b} &\to \m{b} &
\m{a}\m{b}\m{a} &\to \m{a}\m{b}\m{b}
\intertext{At this point we use the third rule to \tsfs{collapse} the
first rule:}
\m{a}\m{b}\m{b} &\approx \m{a}\m{b} &
\m{b}\m{b} &\to \m{b} &
\m{a}\m{b}\m{a} &\to \m{a}\m{b}\m{b}
\intertext{An application of \tsfs{simplify} followed by \tsfs{delete}
results in:}
&&
\m{b}\m{b} &\to \m{b} &
\m{a}\m{b}\m{a} &\to \m{a}\m{b}\m{b}
\intertext{Repeating the above process produces}
&&
\m{b}\m{b} &\to \m{b} &
\m{a}\m{b}\m{a} &\to
\makebox[0mm][l]{$\m{a}\m{b}\m{b}\m{b}$}
\intertext{and then}
&&
\m{b}\m{b} &\to \m{b} &
\m{a}\m{b}\m{a} &\to
\makebox[0mm][l]{$\m{a}\m{b}\m{b}\m{b}\m{b}$}
\end{align*}
ad infinitum. Since none of the rules
$\m{a}\m{b}\m{a} \to \m{a}\m{b}^n$ survives, in the
limit we obtain the TRS consisting of the single rule
$\m{b}\m{b} \to \m{b}$. This TRS is complete but
not equivalent to $\EE$ as witnessed by non-joinability of
$\m{a}\m{b}\m{a}$ and $\m{a}\m{b}$.
\end{exa}

\begin{defi}[Knuth-Bendix Completion]%
\label{def:KBi}
The inference system \KBi consists of the inference
rules \tsfs{deduce}, \tsfs{orient}, \tsfs{delete}, \tsfs{compose}, and
\tsfs{simplify} of \KBf together with the following modified collapse
rule:
\begin{center}
\bigskip
\begin{tabular}{@{}l@{\qquad}c@{\qquad}l@{}}
\tsfs{collapse}$_{\smprencompasses}$ &
$\displaystyle \frac
{\EE,\RR \uplus \{ t \to s \}}
{\EE \cup \{ u \approx s \},\RR}$
& if $t \xrightarrow{\smprencompasses}_\RR u$
\end{tabular}
\bigskip
\end{center}
Here the condition \smash{$t \xrightarrow{\smprencompasses}_\RR u$} is
defined as \smash{$t \xrightarrow{} u$} using some rule
$\ell \to r \in \RR$ such that $t \prencompasses \ell$.
\end{defi}

Note that the collapse step in \exaref{collapse encompassment} does
not satisfy the encompassment condition from the previous
definition.

\smallskip

We write $(\EE,\RR) \vdi (\EE',\RR')$ if $(\EE',\RR')$ can be reached
from $(\EE,\RR)$ by employing one of the inference rules of \defref{KBi}.

\begin{defi}%
\label{def:KBi fairness}
An \emph{infinite run} is a maximal sequence of the form
\[
\Gamma\colon
(\EE_0,\RR_0) ~\vdi~ (\EE_1,\RR_1) ~\vdi~ (\EE_2,\RR_2) ~\vdi~ \cdots
\]
with $\RR_0 = \varnothing$. We define
\begin{align*}
\EEi & = \bigcup_{i \geqslant 0} \EE_i &
\RRi & = \bigcup_{i \geqslant 0} \RR_i &
\EEw & = \bigcup_{i \geqslant 0} \bigcap_{j \geqslant i} \EE_j &
\RRw & = \bigcup_{i \geqslant 0} \bigcap_{j \geqslant i} \RR_j
\end{align*}
Equations in $\EEw$ and rules in $\RRw$ are called \emph{persistent}.
The run $\Gamma$ is called \emph{fair} if $\EEw = \varnothing$ and the
inclusion
\(
\PCP(\RRw) \subseteq {\join_{\RRw}} \cup {\fromto_{\EEi}}
\)
holds.
\end{defi}

Bachmair \etal~\cite{BDH86} proved that for every
fair run satisfying $\EEw = \varnothing$
the TRS $\RRw$ constitutes a complete presentation of $\EE_0$.
The remainder of this section is dedicated to establish the same result,
but on a different route without encountering proof orders.

Compared to our proofs for finite runs from \secref{finite runs}, in the
following we will disentangle our reasoning about rules from our reasoning
about equations and furthermore replace peak decreasingness by the
slightly simpler concept of source decreasingness.
So why not use this more modular and simpler approach also in our earlier
proofs for finite runs? The main difference between the two situations is
the encompassment condition of \tsfs{deduce}. Unfortunately, without the
encompassment condition the equivalent of \lemref{Ri to Rw} below for
finite runs breaks down and we are forced to reason about rules and
equations simultaneously (\lemref{KB key}).
Nevertheless, it seems useful to also have a correctness proof for
\KBf (lacking the encompassment condition),
since out of the four completion tools we are
aware of (\cime3~\cite{CCFPU11},
\kbcv~\cite{kbcv}, \mkbtt~\cite{mkbtt},
\slothrop~\cite{slothrop}), only
\cime3 actually implements the encompassment condition.

\begin{lem}%
\label{lem:inclusions infinite}
If $(\EE,\RR) \vdi (\EE',\RR')$ then the following inclusions
hold:
\begin{enumerate}
\item
\(
\EE' \cup \RR' \:\subseteq\: {\xlr[\ER]{*}}
\)
\hfill\isaforlink{Abstract_Completion}{lem:KB_subset'}

\smallskip
\item%
\label{lem:inclusions infinite:E}
\(
\EE \setminus \EE' \:\subseteq\:
(\xrightarrow[\RR']{} \cdot \mathrel{\EE'}) \:\cup\:
(\mathrel{\EE'} \cdot \xleftarrow[\RR']{}) \:\cup\:
\RR' \:\cup\: \RR'^{-1} \:\cup\: {=}
\)
\hfill\isaforlink{Abstract_Completion}{lem:KB_E_subset}

\smallskip
\item%
\label{lem:inclusions infinite:R}
\(
\RR \setminus \RR' \:\subseteq\:
(\xrightarrow[\RR']{\smprencompasses} \cdot \mathrel{\EE'}) \:\cup\:
(\mathrel{\RR'} \cdot \xleftarrow[\RR']{})
\)
\hfill\isaforlink{Abstract_Completion}{lem:KB_R_subset}
\end{enumerate}
\end{lem}

\noindent
Together these properties reveal that inference steps do not change
the conversion relation.

\begin{cor}%
\label{cor:1}
If $(\EE,\RR) \vdi^{*} (\EE',\RR')$ then the relations
$\xlr[\ER]{*}$ and $\xlr[\EE' \scup \RR']{*}$ coincide.
\hfill
\isaforlink{Abstract_Completion}{lem:KB_conversion}
\end{cor}

Below, we consider an infinite run
$\Gamma\colon (\EE_0,\RR_0) \,\vdi\, (\EE_1,\RR_1) \,\vdi\,
(\EE_2,\RR_2) \,\vdi\, \cdots$ such that $\EEw = \varnothing$.
First we show that all rewrite rules are compatible with the reduction
order $>$.

\begin{lem}%
\label{lem:Rw terminating}
The inclusions $\RRw \subseteq \RRi$
$\subseteq {>}$ hold.
\hfill
\isaforlink{Abstract_Completion}{lem:R_per_subset_R_inf}
\isaforlink{Abstract_Completion}{lem:run_R_less}
\end{lem}

Next, we verify that every equality in $\EE_i$ can be turned into a
valley in $\RRi$. Note that in contrast to the proof order
approach~\cite{BDH86} and to the correctness proof for finite runs given
in \secref{finite runs} we reason separately about equations and rules.
This more local rationale simplifies the analysis as we can use different
well-founded induction arguments for the two cases, rather than
synthesizing an order that covers both.

\begin{lem}%
\label{lem:E to R}
The inclusion
$\EEi \subseteq {\join_\RRi}$ holds.
\hfill\isaforlink{Abstract_Completion}{lem:E_i_subset_join_R_inf}
\end{lem}
\begin{proof}
Let $s \approx t \in \EE_i$ for some $i \geqslant 0$. By induction
on $\{ s, t \}$ with respect to $>_\mul$ we show $s \join_\RRi t$. Because
$\EEw = \varnothing$, $s \approx t \in \EE_{j-1} \setminus \EE_j$
for some $j > i$. Following \lemref[E]{inclusions infinite},
we distinguish three cases.
\begin{itemize}
\item
If $s \approx t \in \RR_j \cup \RR_j^{-1} \cup {=}$ then the
claim trivially holds.
\smallskip
\item
If $s \to_{\RR_j} u$ and $u \approx t \in \EE_j$ for some term $u$ then
$\{ s, t \} >_\mul \{ u, t \}$ and thus
$u \join_\RRi t$ by the induction hypothesis. Hence also
$s \join_\RRi t$.
\smallskip
\item
Similarly, if $s \approx u \in \EE_j$ and $u \FromB{\RR_j} t$ for some
term $u$ then $\{ s, t \} >_\mul \{ s, u \}$ and we obtain
$s \join_\RRi t$ as in the preceding case.
\qedhere
\end{itemize}
\end{proof}

\begin{cor}%
\label{cor:2}
The inclusion ${\xrightarrow[\EEi]{}} \subseteq {\xlr[\RRi]{*}}$ holds.
\hfill\isaforlink{Abstract_Completion}{lem:rstep_E_i_subset}
\end{cor}

In order to show confluence of $\RRw$ we use
source decreasingness as defined in \secref{preliminaries}, employing
the following extension of the reduction order $>$.

\begin{defi}%
\label{def:succ}
We define ${\succ} =
{(\mathrel{({>} \cup {\prencompasses})} / \mathrel{\encompasses})}^+$.
\end{defi}

According to \lemref{proper encompassment extension1}, $\succ$ is a
well-founded order.
The next lemma allows us to transform every non-persistent rule
$\ell \to r$ into an $\RRw$-conversion below $\ell$.

\begin{lem}%
\label{lem:Ri to Rw}
The inclusion
${\xrightarrow[\RRi]{s}} \subseteq {\xlr[\RRw]{\smash{\ddg\,s}}^*}$
holds for all terms $s$.
\hfill\isaforlink{Abstract_Completion}{lem:slab_R_inf_subset}
\end{lem}
\begin{proof}
Let $s \xrightarrow{s}_\RRi t$ by employing
the rewrite rule $\ell \to r$. We prove
$s \xlr{\smash{\ddg\,s}}_\RRw^* t$ by
induction on $(\ell,r)$ with respect to $\succ_\lex$.
If $\ell \to r \in \RRw$ then the claim trivially holds.
Otherwise, $\ell \to r \in \RR_{i-1} \setminus \RR_i$
for some $i > 0$. Using
\lemref[R]{inclusions infinite},
we distinguish two cases.
\begin{itemize}
\item
Suppose $\ell \xrightarrow{\smprencompasses}_{\ell' \to r'} u$ and
$u \approx r \in \EE_i$ for some term $u$ and rule
$\ell' \to r' \in \RR_i$.
We obtain
$\ell \xrightarrow{\smash{\smprencompasses}}_{\ell' \to r'}^{} u
\join_\RRi r$ from \lemref{E to R}.
We have $\ell \prencompasses \ell'$ and both $\ell > u$ and
$\ell > r$. It follows that all rewrite rules $\ell'' \to r''$ employed in
$\ell \xrightarrow{\smash{\smprencompasses}} u \join_\RRi r$
satisfy $(\ell,r) \succ_\lex (\ell'',r'')$. Moreover, all steps in
$\ell \join_\RRi r$ are labeled with a term $\leqslant \ell$.
Hence we obtain $\ell \xlr{\smash{\ddg\,\ell}}_\RRw^* r$
from the induction hypothesis.
\smallskip
\item
Suppose $\ell \to u \in \RR_i$ and $u \from_{\ell' \to r'} r$
for some term $u$ and rewrite rule $\ell' \to r' \in \RR_i$.
We have $(\ell,r) \succ_\lex (\ell,u)$ and
$(\ell,r) \succ_\lex (\ell',r')$
because $r > u$ and $\ell > r \encompasses \ell' > r'$.
Moreover, both steps are labeled with a term $\leqslant \ell$ and thus
we obtain $\ell \xlr{\smash{\ddg\,\ell}}_\RRw^* r$
from the induction hypothesis.
\end{itemize}
So in both cases we have $\ell \xlr{\smash{\ddg\,\ell}}_\RRw^* r$
and thus also $s \xlr{\smash{\ddg\,s}}_\RRw^* t$.
\end{proof}

\begin{cor}%
\label{cor:3}
The relations $\xlr[\RRi]{*}$ and $\xlr[\RRw]{*}$ coincide.
\hfill\isaforlink{Abstract_Completion}{lem:rstep_R_inf_conv_iff}
\end{cor}

We arrive at the main theorem of this section.
Note that Bachmair's correctness proof~\cite{B91}
uses induction with respect to a well-founded order on conversions to
directly show that any conversion of $\EEi \cup \RRi$ can be transformed
into a joining sequence of $\RRw$.
In contrast, we prove confluence via source decreasingness.
This allows us to concentrate on \emph{local} peaks.

\begin{thm}%
\label{thm:KBi correctness}
If $\Gamma$ is fair then $\RRw$ is a complete presentation of $\EE_0$.
\hfill\isaforlink{Completion_Fairness}{lem:infinite_fair_run}
\end{thm}

\begin{proof}
We have $\EEw = \varnothing$ because $\Gamma$ is non-failing.
The TRS $\RRw$ is terminating by \lemref{Rw terminating}.
We show source decreasingness of labeled $\RRw$ reduction with
respect to the reduction order $>$. So let
$t \FRom{\RRw}{}{\smash{s}} s \xrightarrow{\smash{s}}_\RRw u$.
From \lemref{pcp} we obtain $t \mathrel{\tds^2} u$.
Let $v \mathrel{\tds} w$ appear in this sequence
(so $t = v$ or $w = u$). We have $s > v$, $s > w$, and
\(
(v,w) \in {\join_\RRw} \cup \fromto_{\EEi}
\)
by the definition of $\tds$ and fairness of $\Gamma$.
\smallskip
\begin{itemize}
\item
If $v \join_\RRw w$ then
$v \xrightarrow{\smash{\ddg\,v}}_\RRw^* \cdot
\FRom{\RRw}{*}{\smash{\ddg\,w}} w$
and thus $v \xlr{\smash{\vee s}}_\RRw^* w$.
\smallskip
\item
If $v \fromto_{\EEi} w$ then
$v \fromto_{\EE_i} w$ for some $i \geqslant 0$ then $v \join_\RRi w$
by \lemref{E to R}. We obtain $v \xlr{\smash{\vee s}}_\RRi^* w$ as
in the previous case and thus $v \xlr{\smash{\vee s}}_\RRw^* w$ by
\lemref{Ri to Rw}.
\end{itemize}
\smallskip
Hence $t \xlr{\smash{\vee s}}_\RRw^* u$. Confluence of $\RRw$ now follows
from Lemmata~\ref{lem:sd => pd} and~\ref{lem:pd => cr}.
It remains to show
${\fromto_{\EE_0}^*} = {\fromto_\RRw^*}$. Using \corref{1} we obtain
${\to_{\EE_i \scup \RR_i}} \subseteq {\fromto_{\EE_0}^*}$
by a straightforward induction on $i$. This in turn yields
${\fromto_{\EE_0}^*} = {\fromto_{\ERi}^*}$.
From \corref{2} we infer
${\fromto_{\ERi}^*} = {\fromto_\RRi^*}$ and we conclude by an
appeal to \corref{3}.
\end{proof}

\begin{exa}
Consider the ES $\EE$ and the KBO $>$ from
Example~\ref{example: braid monoid}.
Let $\PP_n$ for $n \geqslant 1 $ denote the TRS
$\{ \m{ab}^{i+1}\m{ab} \to \m{babba}^{i} \mid
1 \leqslant i \leqslant n \}$.
One possible infinite completion run is the following:
\begin{alignat*}{4}
(\EE,\varnothing) ~ & \vdir{orient} &~&
(\varnothing,\{ \m{aba} \to \m{bab} \})
&~& \vdir{deduce} &~&
(\{ \m{abbab} \approx \m{babba} \},\{ \m{aba} \to \m{bab} \}) \\
& \vdir{orient} &&
(\varnothing,\{ \m{aba} \to \m{bab} \} \cup \PP_1 )
&& \vdir{deduce} &&
(\{ \m{abbbab} \approx \m{babbaa} \},\{ \m{aba} \to \m{bab} \}
\cup \PP_1) \\
& \vdir{orient} &&
(\varnothing,\{ \m{aba} \to \m{bab} \} \cup \PP_2)
&& \vdir{\phantom{deduce}} && \cdots
\end{alignat*}
If this run is continued in a fair way we subsequently construct the TRSs
$\PP_n$ and can in the limit obtain the result
$\RRw = \{ \m{aba} \to \m{bab} \} \cup
\{ \m{ab}^{i+1}\m{ab} \to \m{babba}^i \mid i \geqslant 1 \}$, which is
complete according to \thmref{KBi correctness}.
\end{exa}

This section recapitulates our results on infinite runs~\cite{HMSW17}.
Our correctness proof (\thmref{KBi correctness}) differs substantially
from earlier proofs in the literature. Due to a less monolithic structure
we consider this proof to be more formalization friendly:
Instead of lexicographically combining several orders into a
single proof reduction relation, we use source decreasingness
together with different orders as necessary to prove auxiliary results.
In particular, our approach naturally supports prime critical pairs.

\section{Ordered Completion}%
\label{sec:ordered completion}

Completion may fail to construct a complete system if unorientable
equations are encountered.
For example, the ES $\EE$ consisting of the two equations
$\m{0} + x \approx x$ and $x + y \approx y + x$ admits no complete
presentation. (We will prove it in \secref{completeness}.)
This can happen even if a finite complete
system exists, as illustrated by the following example.

\begin{exa}%
\label{exa:okb1}
Consider the ES $\EE$~\cite{BD94} consisting of the
three equations
\begin{xalignat*}{3}
\m{1} \cdot (-x + x) &\approx \m{0} &
\m{1} \cdot (x + -x) &\approx x + -x &
-x + x &\approx y + - y
\end{xalignat*}
Any run of standard Knuth-Bendix completion will fail on this input
system; the first two equations may be oriented from left to right if
a suitable order is employed but no further steps are possible.
However,
the TRS $\RR$ consisting of the rules
\begin{xalignat*}{3}
\m{1} \cdot \m{0} &\to \m{0} &
x + -x &\to \m{0} &
-x + x &\to \m{0}
\end{xalignat*}
constitutes a canonical presentation of $\EE$.
\end{exa}

Ordered completion was developed to remedy this shortcoming.
In contrast to completion as presented in the preceding section it
never fails, though the resulting system is in general only ground
complete.

For an ES $\EE$, an ordered rewrite step is a rewrite
step using a rule from $\EE^>$, which is the infinite set of rewrite rules
$\ell\sigma \to r\sigma$ such that $\ell \approx r \in \EE^\pm$ and
$\ell\sigma > r\sigma$ for some substitution $\sigma$.

The following inference rules for ordered completion are due to Bachmair,
Dershowitz, and Plaisted~\cite{BDP89}.
In order to simplify the notation, we abbreviate $\REwgt$ to
$\SS$, and use the following shorthands.
We write $t \xrightarrow{\pe}_{\EE^>} u$ if there exist
an equation $\ell \approx r \in \EE^\pm$,
a context $C$,
and a substitution $\sigma$
such that
$t = C[\ell\sigma]$, $u = C[r\sigma]$,
$\ell\sigma > r\sigma$, and $t \prencompasses \ell$.
The union of $\xrightarrow{}_\RR$ and
$\smash{\xrightarrow{\pe}_{\EE^>}}$ is denoted by
$\smash{\xrightarrow{\pe_1}_\SS}$ and we write
$\smash{\xrightarrow{\pe_2}_\SS}$ for the union of
$\smash{\xrightarrow{\pe}_\RR}$ and
$\smash{\xrightarrow{\pe}_{\EE^>}}$.

\begin{defi}[Ordered Completion
\isaforlink{Ordered_Completion}{ind:oKBi}]%
\label{def:KBo}
The inference system \KBo of ordered completion
operates on pairs $(\EE,\RR)$ of equations $\EE$ and rules $\RR$
over a common signature $\FF$. It consists of the
following inference rules:
\begin{center}
\begin{tabular}{@{}lcl@{\qquad}lcl@{}}
\tsfs{deduce} &
$\displaystyle \frac
{\EE,\RR}
{\EE \cup \{ s \approx t \},\RR}$
& if
$s \xleftarrow[\RR \scup \EE^{\pm}]{} \cdot
\xrightarrow[\RR \scup \EE^{\pm}]{} t$
&
\tsfs{compose} &
$\displaystyle \frac
{\EE,\RR \uplus \{ s \to t \}}
{\EE,\RR \cup \{ s \to u \}}$
& if $t \xrightarrow{}_{\SS} u$
\\ & \\
&
$\displaystyle \frac
{\EE \uplus \{ s \approx t \},\RR}
{\EE,\RR \cup \{ s \to t \}}$
& if $s > t$
&
&
$\displaystyle \frac
{\EE \uplus \{ s \approx t \},\RR}
{\EE \cup \{ u \approx t \},\RR}$
&
if $s \xrightarrow{\pe_1}_{\SS} u$
\\[-.5ex]
\tsfs{orient} & & &
\tsfs{simplify}
\\[-.5ex]
&
$\displaystyle \frac
{\EE \uplus \{ s \approx t \},\RR}
{\EE,\RR \cup \{ t \to s \}}$
& if $t > s$
&
&
$\displaystyle \frac
{\EE \uplus \{ s \approx t \},\RR}
{\EE \cup \{ s \approx u \},\RR}$
&
if $t \xrightarrow{\pe_1}_{\SS} u$
\\ & \\
\tsfs{delete} &
$\displaystyle \frac
{\EE \uplus \{ s \approx s \},\RR}
{\EE,\RR}$ &
&
\tsfs{collapse} &
$\displaystyle \frac
{\EE,\RR \uplus \{ t \to s \}}
{\EE \cup \{ u \approx s \},\RR}$
&
if $t \xrightarrow{\pe_2}_{\SS} u$
\end{tabular}
\end{center}
\medskip
\end{defi}

The \tsfs{deduce} rule may be applied to any peak, though in practice it
is typically limited to the addition of extended critical pairs
(which are defined in \defref{extended overlap} below).
We write $(\EE,\RR) \vdo (\EE',\RR')$ if $(\EE',\RR')$ can be reached from
$(\EE,\RR)$ by employing one of the inference rules of
\defref{KBo}. We start by stating the equivalents of
\lemref{inclusions infinite} and \corref{1} for ordered completion.

\begin{lem}%
\label{lem:okb inclusions infinite}
If $(\EE,\RR) \vdo (\EE',\RR')$ then the following inclusions
hold:
\begin{enumerate}
\item
$
\EE' \cup \RR'
\:\subseteq\:
{\xlr[\ER]{*}}
$
\hfill\isaforlink{Ordered_Completion}{lem:oKBi_subset}

\smallskip
\item%
\label{lem:okb inclusions infinite:E}
$
\EE \setminus \EE' \:\subseteq\:
{(\xrightarrow[\SS']{\pe_1} \cdot \mathrel{\EE'}^\pm)}^\pm_{\vphantom{i}}
\:\cup\: \RR'^\pm_{\vphantom{i}}
\:\cup\: {=}
$
\hfill\isaforlink{Ordered_Completion}{lem:oKBi_E_supset}

\smallskip
\item%
\label{lem:okb inclusions infinite:R}
$
\RR \setminus \RR' \:\subseteq\:
(\xrightarrow[\SS']{\pe_2} \cdot \mathrel{\EE'}) \:\cup\:
(\mathrel{\RR'} \cdot \xleftarrow[\SS']{})
$
\hfill\isaforlink{Ordered_Completion}{lem:oKBi_R_supset}
\end{enumerate}
\end{lem}

\begin{cor}%
\label{cor:oKB conv}
If $(\EE,\RR) \vdo (\EE',\RR')$ then the relations
$\xlr[\EE\cup\RR]{*}$ and $\xlr[\EE'\cup\RR']{*}$ coincide.
\hfill
\isaforlink{Ordered_Completion}{lem:oKBi_conversion}
\end{cor}

We illustrate \KBo by means of an example.

\begin{exa}%
\label{exa:okb2}
Consider the ES $\EE$ consisting of the following three equations:
\begin{xalignat*}{3}
\m{f}(x) &\approx \m{f}(\m{a}) &
\m{f}(\m{b}) &\approx \m{b} &
\m{g}(\m{f}(\m{b}),x) &\approx \m{g}(x,\m{b})
\end{xalignat*}
By taking the Knuth-Bendix order $\KBO$ with precedence $\m{f} > \m{b}$ and
where all function symbols are assigned weight $1$,
the following \KBo inference sequence can be obtained:
\begin{alignat*}{2}
(\EE,\varnothing) ~ & \mathrel{{\vdor{orient}}^{+}} &~&
(\{ \m{f}(x) \approx \m{f}(\m{a}) \},
 \{ \m{f}(\m{b}) \to \m{b}, \m{g}(\m{f}(\m{b}),x) \to \m{g}(x,\m{b}) \}) \\
 & \vdor{deduce} &&
(\{ \m{f}(x) \approx \m{f}(\m{a}), \m{f}(\m{b}) \approx \m{f}(\m{a}) \},
 \{ \m{f}(\m{b}) \to \m{b}, \m{g}(\m{f}(\m{b}),x) \to \m{g}(x,\m{b}) \}) \\
 & \vdor{simplify} &&
(\{ \m{f}(x) \approx \m{f}(\m{a}), \m{b} \approx \m{f}(\m{a}) \},
 \{ \m{f}(\m{b}) \to \m{b}, \m{g}(\m{f}(\m{b}),x) \to \m{g}(x,\m{b}) \}) \\
 & \vdor{orient} &&
(\{ \m{f}(x) \approx \m{f}(\m{a}) \},
 \{ \m{f}(\m{b}) \to \m{b}, \m{g}(\m{f}(\m{b}),x) \to \m{g}(x,\m{b}),
    \m{f}(\m{a}) \to \m{b} \}) \\
 & \vdor{simplify} &&
(\{ \m{f}(x) \approx \m{b} \},
 \{ \m{f}(\m{b}) \to \m{b}, \m{g}(\m{f}(\m{b}),x) \to \m{g}(x,\m{b}),
    \m{f}(\m{a}) \to \m{b} \}) \\
 & \mathrel{{\vdor{collapse}}^+} &&
(\{ \m{f}(x) \approx \m{b}, \m{b} \approx \m{b},
    \m{g}(\m{b},x) \approx \m{g}(x,\m{b}) \}, \varnothing) \\
 & \vdor{orient} &&
(\{ \m{b} \approx \m{b}, \m{g}(\m{b},x) \approx \m{g}(x,\m{b}) \},
 \{ \m{f}(x) \to \m{b} \}) \\
 & \vdor{delete} &&
(\{ \m{g}(\m{b},x) \approx \m{g}(x,\m{b}) \}, \{ \m{f}(x) \to \m{b} \}) \\
 & \vdor{deduce} &&
(\{ \m{g}(\m{b},x) \approx \m{g}(x,\m{b}), \m{b} \approx \m{b} \},
 \{ \m{f}(x) \to \m{b} \})
 \end{alignat*}
This sequence can be extended to an (infinite) run by repeating the last
two steps. Then we have $\RRw = \{ \m{f}(x) \to \m{b} \}$ and
$\EEw = \{ \m{g}(\m{b},x) \approx \m{g}(x,\m{b}) \}$.
\end{exa}

Below, we consider an arbitrary run $\Gamma\colon
(\EE_0,\RR_0) \,\vdo\, (\EE_1,\RR_1) \,\vdo\, (\EE_2,\RR_2) \,\vdo\,
\cdots$\,.
In general $\EEi \subseteq {\join_\RRi}$ does not hold, as
\exaref{okb2} illustrates. So unlike in the preceding section we now
omit the condition $\EEw = \varnothing$. However, this comes at the price
of weaker properties of the resulting system, as the remainder of this
section shows.

\begin{lem}%
\label{lem:oKB Rw terminating}
The inclusions $\RRw \subseteq \RRi \subseteq {>}$
and $\EEw \subseteq \EEi$ hold.
\hfill
\isaforlink{Ordered_Completion}{lem:Rw_subset_Rinf}
\isaforlink{Ordered_Completion}{lem:Rinf_less}
\isaforlink{Ordered_Completion}{lem:Ew_subset_Einf}
\end{lem}

We use the relation $\xrs{M}$ from \defref{mset labeled rewriting}
to show that any equation step below a term set $M$ eventually turns into
a conversion over $\REiw$ that is still below $M$.
Note that just like in \secref{infinite runs} we
avoid the use of a synthesized termination argument by
handling equations and rules separately.

\begin{lem}%
\label{lem:oKB Ei to EwRi}
The inclusion
${\xrightarrow[\EEi]{S}} \:\subseteq\: {\xlr[\REiw]{\smash{S}}^*}$
holds for all sets $S$ of terms.
\hfill
\isaforlink{Ordered_Completion}{lem:Ei_subset_EwRi}
\end{lem}
\begin{proof}
Let $t \approx u \in \EEi$. We prove
\[
{\xrightarrow[t\,\approx\,u]{M}} \:\subseteq\:
{\xlr[\REiw]{\smash{M}}^*}
\]
by induction on $\{ t, u \}$ with respect to the well-founded
order $\succ_\mul$. If $t \approx u \in \EEw^\pm$ then the claim follows
trivially. Otherwise,
$t \approx u \in {(\EE_{i-1} \setminus \EE_i)}^\pm$ for some $i > 0$.
Using \lemref[E]{okb inclusions infinite}, we distinguish two
subcases.
\smallskip
\begin{itemize}
\item
Suppose $t \approx u \in {(\xrightarrow{\pe_1}_{\SS_i} \cdot
\mathrel{\EE_i^\pm})}^\pm_{\vphantom{i}}$.
There exist a term $t'$ and an equation
$v' \approx u' \in \EE_i^\pm$ such
that $\{ t, u \} = \{ t', u' \}$ and $t' \xrightarrow{\pe_1}_{\SS_i} v'$.
It is sufficient to show
\[
t' \xlr[\REiw]{\smash{\{ t',\,u' \}}}^* v' \quad\text{and}\quad
v' \xlr[\REiw]{\smash{\{ t',\,u' \}}}^* u'
\]
The second conversion follows from $t' > v'$ and
the induction hypothesis for
$v' \approx u' \in \EE_i^\pm$, which is applicable as
$\{ t, u \} = \{ t', u' \} \succ_\mul \{ v', u' \}$.
The first conversion is obtained as follows.
Because of $\smash{t' \xrightarrow{\pe_1}_{\SS_i} v'}$,
we have $t' \to_{\RR_i} v'$ or
$\smash{t' \xrightarrow{\pe}_{\EE_i^>} v'}$.
If $t' \to_{\RR_i} v'$ then this step can be labeled with
$\{ t', u' \}$ as $t' > v'$.
Otherwise, there exist an equation $\ell \approx r \in \EE_i^\pm$, a
context $C$, and a substitution $\sigma$ such that
$t' = C[\ell\sigma]$, $v' = C[r\sigma]$, $\ell\sigma > r\sigma$, and
$t' \prencompasses \ell$. We have $t' \succ \ell$ and $t' \succ r$ as
$t' \encompasses \ell\sigma > r\sigma \encompasses r$.
Therefore $\{ t', u' \} \succ_\mul \{ \ell, r \}$ holds, so
\[
\ell \xlr[\REiw]{\smash{\{ \ell,\,r \}}}^* r
\]
follows from the induction hypothesis. Closure under contexts and
substitutions now yields
$t \xlr{\smash{\{ t,\,u \}}}_{\REiw}^* u$.
\smallskip
\item
If $t \approx u \in \RR_i^\pm \cup {=}$ then
$t \xlr[\RRi]{\smash{\{ t,\,u \}}}^= u$.
\end{itemize}
\smallskip
In both cases $t \xlr{\smash{\{ t,\,u \}}}_{\REiw}^* u$
holds. Since $M$ contains upper bounds of $t$ and $u$ with respect to
$\geqslant$, the desired inclusion follows from the closure
under contexts and substitutions of
$\to_{\REiw}$ and $\geqslant$.
\end{proof}

Next, we show that a rewrite step that uses a rule in $\RRi$ and is
below a multiset of terms $M$ eventually turns into a
conversion over persistent rules and equations that is still below $M$.
To this end we write $\curlyvee t$ for the set
$\{ u \in \TT(\FF,\VV) \mid t \succ u \}$.

\begin{lem}%
\label{lem:oKB Ri to EwRw}
The inclusion
${\xrightarrow[\RRi]{M}} \:\subseteq\: {\xlr[\ERw]{\smash{M}}^*}$
holds for all multisets $M$ of terms.
\hfill
\isaforlink{Ordered_Completion}{lem:Ri_subset_ERw}
\end{lem}
\begin{proof}
Let $\ell \approx r \in \RRi$. We prove
\[
{\xrightarrow[\ell \to r]{\smash{M}}} \:\subseteq\:
{\xlr[\ERw]{\smash{M}}^*}
\]
by induction on $(\ell,r)$ with respect to the well-founded order
$\succ_\lex$. If $\ell \to r \in \RRw$ then the claim trivially holds.
Otherwise, there is some $i > 0$ such that
$\ell \to r \in \RR_{i-1} \setminus \RR_i$.
From \lemref{oKB Ei to EwRi} and the induction hypothesis the
inclusions
\begin{equation}
\label{eq:RiEi}
{\xrightarrow[\ERi]{N}}
\:\subseteq\: {\xlr[\REiw]{\smash{N}}^*}
\:\subseteq\: {\xlr[\ERw]{\smash{N}}^*}
\end{equation}
are obtained for every set $N \subseteq \curlyvee \ell$.
Using \lemref{okb inclusions infinite}, we
distinguish two cases.
\smallskip
\begin{itemize}
\item
Suppose $\ell \xrightarrow{\pe_2}_{\SS_i} u$ and $u \approx r \in \EE_i$
for some term $u$. There exist an equation
$\ell' \approx r' \in \EEi^\pm \cup \RRi$, a context $C$ and a
substitution $\sigma$ such that $\ell = C[\ell'\sigma]$, $u = C[r'\sigma]$,
$\ell\sigma > r\sigma$, and $\ell \prencompasses \ell'$.
We have $\ell \succ \ell', r'$ as $\ell \prencompasses \ell'$ and
$\ell \encompasses \ell'\sigma > r'\sigma \encompasses r'$
and thus
\[
\ell' \xlr[\ERi]{\smash{\{ \ell' \}}} r'
\]
Since
$\{ \ell', r' \} \subseteq {\curlyvee \ell}$ we obtain
$\ell' \xlr{\smash{\{ \ell',\,r' \}}}_{\ERw}^* r'$
from~\eqref{eq:RiEi}. Therefore,
$\ell \xlr{\smash{\{ \ell \}}}_{\ERw}^{*} u$ follows
from closure under contexts and substitutions and $\ell > u$.
Again from $\ell > u, r$ we obtain
$u \xlr{\smash{\curlyvee \ell}}_{\ERi} r$ and thus
$u \xlr{\smash{\curlyvee \ell}}_{\ERw} r$
follows from~\eqref{eq:RiEi}.
\smallskip
\item
Suppose $\ell \to u \in \RR_i$ and $u \FromB{\SS_i} r$ for some term $u$.
We have $r > u$ and thus $(\ell,r) \succ_\lex (\ell,u)$. Hence we can
apply the induction hypothesis to
$\ell \xrightarrow{\smash{\{ \ell \}}}_{\ell \to u} u$,
yielding $\ell \xlr{\smash{\{ \ell \}}}_{\ERw}^* u$.
From $\ell > r > u$ we obtain
$u \xlr{\smash{\curlyvee \ell}}_{\ERi} r$ and thus
$u \xlr{\smash{\curlyvee\ell}}_{\ERw}^* r$ follows by~\eqref{eq:RiEi}.
\end{itemize}
\smallskip
In both cases $\ell \xlr{\smash{\{ \ell \}}}_{\ERw}^* r$ holds.
Since $\to_{\ERw}$ and $\geqslant$ are closed under contexts and
substitutions, the desired inclusion on steps using $\ell \to r$ follows.
\end{proof}

We can combine the preceding lemmata to obtain an inclusion in
conversions over persistent equations and rules.

\begin{cor}%
\label{cor:REi subset REw}
The inclusion ${\xrightarrow[\ERi]{\smash{M}}} \:\subseteq\:
{\xlr[\ERw]{\smash{M}}^*}$ holds for all multisets of terms $M$.
\hfill
\isaforlink{Ordered_Completion}{lem:ERi_subset_ERw}
\end{cor}

For instance, in \exaref{okb2} we have
$\smash{\m{f}(x) \xlr{\,M\,}_{\EEi} \m{f}(\m{a})}$
for $M = \{ \m{f}(x), \m{f}(\m{a}) \}$ and the conversion
$\m{f}(x) \fromto \m{b} \fromto \m{f}(\m{a})$ in $\ERw$
clearly satisfies
$\smash{\m{f}(x) \xlr{\smash{\,M\,}}_{\ERw}^{*} \m{f}(\m{a})}$.

The results obtained so far are sufficient to show that ordered completion
can produce a complete system.

\begin{thm}%
\label{thm:complete presentation okb1}
If $\Gamma$ satisfies $\PCP(\RRw) \subseteq
{\join_\RRw \cup \fromto_{\EEi}}$
and $\EEw = \varnothing$ then $\RRw$ is a
complete presentation of $\EE_0$.
\hfill
\isaforlink{Ordered_Completion}{lem:Ew_empty_implies_CR_Rw}
\end{thm}
\begin{proof}
We prove that $\RRw$ is confluent by showing that labeled $\RRw$ reduction
on arbitrary terms is source decreasing. Consider
$t \FRom{\RRw}{}{\smash{s}} s \xrightarrow{\smash{s}}_{\RRw} u$.
From \lemref{pcp} we obtain $t \mathrel{\tds^2} u$
(where $\RRw$ takes the place of $\RR$ in the definition of $\tds$).
Let $v \mathrel{\tds} w$ appear in this sequence
(so $t = v$ or $w = u$). We have $s > v$, $s > w$, and
$v \join_{\RRw} w$ or $v \fromto_{\EEi} w$
by the definition of $\tds$
and the assumption $\PCP(\RRw) \subseteq {\fromto_{\EEi}}$.
\begin{itemize}
\item
If $v \join_{\RRw} w$ then
$v \xrightarrow{\smash{\ddg\,v}}_{\RRw}^* \cdot
\FRom{\RRw}{*}{\smash{\ddg\,w}} w$
and thus $v \xlr{\smash{\vee s}}_{\RRw}^* w$.
\medskip
\item
If $v \fromto_{\EEi} w$ then
$v \xlr{\smash{\vee s}}_{\RRw}^* w$ by \corref{REi subset REw}.
\end{itemize}
\smallskip
Hence $t \xlr{\smash{\vee s}}_{\RRw}^* u$. Confluence of $\RRw$
follows from Lemmata~\ref{lem:sd => pd} and~\ref{lem:pd => cr}.
Termination of $\RRw$ holds by
\lemref{oKB Rw terminating}. We have ${\conv_{\EE_0}} = {\conv_{\RRw}}$
by an easy induction argument using \corref{oKB conv}, so
$\RRw$ is a complete presentation of $\EE_0$.
\end{proof}

From now on we specialize our results to ground terms.
In the remainder of this section we therefore assume that $>$ is a
ground-total reduction order.
Before continuing with results on ordered completion, we define
extended critical pairs.

\begin{defi}[Extended Overlaps
\isaforlink{Ordered_Rewriting}{def:ooverlap}]%
\label{def:extended overlap}
An \emph{extended overlap} of a given ES $\EE$ is a triple
$\langle \ell_1 \approx r_1, p, \ell_2 \approx r_2 \rangle$
satisfying the following properties:
\begin{itemize}
\item
there are renamings $\pi_1$ and $\pi_2$ such that
$\pi_1 (\ell_1 \approx r_1), \pi_2 (\ell_2 \approx r_2)
\in \EE^\pm$ (\ie, the equations are variants of equations in $\EE^\pm$),
\smallskip
\item
$\Var(\ell_1 \approx r_1) \cap \Var(\ell_2 \approx r_2) = \varnothing$,
\smallskip
\item
$p \in \Pos_\FF(\ell_2)$,
\smallskip
\item
$\ell_1$ and ${\ell_2}|_p$ are unifiable with some mgu $\mu$, and
\smallskip
\item
$r_1\mu \not> \ell_1\mu$ and $r_2\mu \not> \ell_2\mu$.
\end{itemize}
An extended overlap gives rise to the \emph{extended critical pair}
$\ell_2{[r_1]}_p\mu \approx r_2\mu$.
An extended critical pair is called \emph{prime} if all proper subterms of
$\ell_1\mu$ are $\EE^>$-normal forms. The set of extended prime critical
pairs among equations in $\EE$ is denoted by $\PCP_>(\EE)$.
\end{defi}

For example, the equations
$\m{1} \cdot (x + -x) \approx x + -x$ and
$y + - y \approx -z + z $ are variable-disjoint variants of equations in
\exaref{okb1}. Neither of them can be oriented from right to left
(independent of the choice of $>$). Because of the peak
$\m{1} \cdot (-z + z) \fromto \m{1} \cdot (x + -x) \fromto x + -x$
they admit the extended overlap $\langle y + - y \approx -z + z, 1,
\m{1} \cdot (x + -x) \approx x + -x\rangle$
which gives rise to the extended critical pair
$\m{1} \cdot (-z + z) \approx x + -x$.
Note that since the second equation is unorientable, a run of a standard
completion procedure will not encounter this critical pair.

Extended critical pairs are important due to the Extended Critical
Pair Lemma~\cite{BDP89}, according to which these are the only peaks
relevant for ground confluence. In our formalization we use the following
variant. The proof employs a similar peak analysis as in
\lemref{cpeakL}.

\begin{lem}%
\label{lem:xcp}
Let $\EE$ be an ES and consider a peak
\hfill
\isaforlink{Ordered_Rewriting}{lem:non_ooverlap_GROUND_joinable}
\[
t \xleftarrow[r_1\,\approx\,\ell_1]{pq,\,\sigma_1} s
 \xrightarrow[\ell_2\,\approx\,r_2]{p,\,\sigma_2} u
\]
involving ground terms $s$, $t$, and $u$
such that $\ell_1\sigma_1 > r_1\sigma_1$ and
$\ell_2\sigma_2 > r_2\sigma_2$. If
$\ell_1 \approx r_1, \ell_2 \approx r_2 \in \EE$ do not form an
extended overlap at position $q$ then
$t \downarrow_{\EE^>} u$.
\end{lem}

In the sequel, we write $\SSw$ for the TRS $\REwgt$.

\begin{cor}%
\label{cor:Ei to Sw}
If $s \xlr[\EEi]{} t$ for ground terms $s$ and $t$ then
$s \xlr[\SSw]{\smash{\{ s, t \}}}^* t$.
\hfill
\isaforlink{Ordered_Completion}{lem:ground_Einf_Sw}
\end{cor}
\begin{proof}
We obtain $s \xlr{\smash{\{ s, t \}}}_{\ERw}^* t$ from
\corref{REi subset REw}. Since $>$ is ground-total, all
$\EEw$ steps in this conversion are ${(\EEw^>)}^\pm$ steps or
trivial steps between identical terms.
Hence
\smash{$s \xlr[\SSw]{\smash{\{ s, t \}}}^* t$} as desired.
\end{proof}

\begin{defi}%
\label{def:KBo fairness}
A run $(\EE_0,\RR_0) \,\vdo\, (\EE_1,\RR_1) \,\vdo\, (\EE_2,\RR_2)
\,\vdo\, \cdots$ is called \emph{fair} if the inclusion
$\PCP_>(\ERw) \subseteq {\join_{\SSw} \cup \fromto_{\EEi}}$ holds.
\end{defi}

The following lemma links extended prime
critical pairs to standard critical pairs and hence allows us
to reuse results from \secref{critical peaks}
for our main correctness result (\thmref{KBo correctness} below).

\begin{lem}%
\label{lem:xpcp}
For a TRS $\RR$ and an ES $\EE$, the inclusion
\smash{${\xlr[\PCP(\SS)]{}} \subseteq
{{\xlr[\PCP_>(\ER)]{}} \cup {\downarrow_{\SS}}}$}
holds on ground terms.
\hfill
\isaforlink{Ordered_Completion}{lem:PCP_xPCP}
\end{lem}
\begin{proof}
Suppose $s \fromto_e t$ for ground terms $s$ and $t$ and a prime
critical pair
$e\colon \ell_2\sigma{[r_1\sigma]}_p \approx r_2\sigma$ generated from
the overlap $\langle \ell_1 \to r_1, p, \ell_2 \to r_2 \rangle$ in $\SS$.
Let $u_i \approx v_i$ be the equation $\ell_i \approx r_i$ if
$\ell_i \to r_i \in \RR$ and the equation in $\EE^\pm$ such that
$\ell_i = u_i\tau_i$ and $r_i = v_i\tau_i$ for some substitution
$\tau_i$ if $\ell_i \to r_i \in \EE^>$.
In the former case we let $\tau_i$ be the empty substitution.
Since the equations $u_1 \approx v_1$ and $u_2 \approx v_2$ are
assumed to be variable-disjoint, the substitution
$\tau = \tau_1 \cup \tau_2$ is well-defined. We distinguish two cases.
\smallskip
\begin{itemize}
\item
If $p \notin \PosF(u_2)$ then
$\langle u_1 \approx v_1, p, u_2 \approx v_2\rangle$ is not an overlap
and hence $s \downarrow_{\SS} t$ by \lemref{xcp}.
\smallskip
\item
Suppose $p \in \PosF(u_2)$. Since
$u_2|_p\tau\sigma = \ell_2|_p\sigma = \ell_1\sigma = u_1\tau\sigma$ there
exist an mgu $\mu$ of $u_2|_p$ and $u_1$, and a substitution $\rho$ such
that $\mu\rho = \tau\sigma$.
Because $u_i\mu\rho = \ell_i\sigma > r_i\sigma = v_i\mu\rho$,
$v_i\mu > u_i\mu$ is impossible.
Hence $e'\colon u_2\mu{[v_1\mu]}_p \approx v_2\mu \in \CP_>(\ER)$
and
\[
\ell_2\sigma{[r_1\sigma]}_p
= u_2\mu\rho{[v_1\mu\rho]}_p
= u_2\mu{[v_1\mu]}_p\rho \xlr[~e'~]{} v_2\mu\rho
= r_2\sigma
\]
Since $e$ is prime, proper subterms of $\ell_2\sigma|_p = u_2\mu\rho|_p$
are irreducible with respect to $\SS$, and hence the same holds for
proper subterms of $u_2\mu$. It follows that
$e' \in \PCP_>(\ER)$ and thus $\ell_2\sigma{[r_1\sigma]}_p
\xlr[\PCP_>(\ER)]{} r_2\sigma$. Hence also
$s \xlr[\PCP_>(\ER)]{} t$.
\qedhere
\end{itemize}
\end{proof}

\noindent
This relationship between extended critical pairs among $\ER$ and
critical pairs among $\SS$ is the final ingredient for the main result of
this section. As in the preceding section, we establish correctness
of ordered completion via source decreasingness.

\begin{thm}%
\label{thm:KBo correctness}
If $\Gamma$ is fair then $\SSw$ is ground-complete and
$\smash{\xlr[~\EE_{0}~]{*}}$ and
$\smash{\xlr[\ERw]{*}}$ coincide.
\hfill
\isaforlink{Ordered_Completion}{lem:correctness_okb}
\end{thm}
\begin{proof}
Termination of $\SSw$ is a consequence of
\lemref{oKB Rw terminating} and the definition of $\EEw^>$.
Next we show that $\SSw$ is ground-confluent. To this end, we show that
labeled $\SSw$ reduction is source decreasing on ground terms.
So let $s$, $t$, and $u$ be ground terms such that
\[
t \xl[\SSw]{s} s \xr[\SSw]{s} u
\]
\lemref{pcp} yields $t \mathrel{\tds^2} u$
(where $\SSw$ takes the place of $\RR$ in the definition of $\tds$).
Let $v \mathrel{\tds} w$ appear in this sequence
(so $t = v$ or $w = u$ and both terms are ground). We have $s > v$,
$s > w$, and
\(
(v,w) \in {\join_{\SSw}} \cup {\fromto_{\EEi}}
\)
by the definition of $\tds$, \lemref{xpcp}, and fairness of
$\Gamma$.
\begin{itemize}
\smallskip
\item
If $v \join_{\SSw} w$ then
$v \xrightarrow{\smash{\ddg\,v}}_{\SSw}^* \cdot
\FRom{\SSw}{*}{\smash{\ddg\,w}} w$
and thus $v \xlr{\smash{\vee s}}_{\SSw}^* w$.
\medskip
\item
If $v \fromto_{\EEi} w$ then $v \fromto_{\EE_i} w$ for some
$i \geqslant 0$ and thus $v \xlr{\smash{\vee s}}_{\SSw}^* w$ by
\corref{Ei to Sw}.
\end{itemize}
\smallskip
Hence $t \xlr{\smash{\vee s}}_{\SSw}^* u$. Confluence of the ARS
that is obtained by restricting $\SSw$ to ground terms now
follows from Lemmata~\ref{lem:sd => pd} and~\ref{lem:pd => cr}.
It remains to show ${\fromto_{\EE_0}^*} = {\fromto_{\ERw}^*}$.
Using \corref{oKB conv} we obtain
${\to_{\EE_i \scup \RR_i}} \subseteq {\fromto_{\EE_0}^*}$ for all
$i \geqslant 0$ by a straightforward induction argument.
This in turn yields
${\fromto_{\ERi}^*} \subseteq {\fromto_{\EE_0}^*}$
and in particular
${\fromto_{\ERw}^*} \subseteq {\fromto_{\EE_0}^*}$.
The reverse inclusion follows from
\corref{REi subset REw} and the inclusion
${\fromto_{\EE_0}^*} \subseteq {\fromto_{\ERi}^*}$.
\end{proof}

If $\EEw$ is empty, the TRS $\RRw$ is not only ground-confluent
but actually confluent on all terms.
Even though this result is not surprising, we did not find it
explicitly stated in the literature.

\begin{thm}%
\label{thm:complete presentation okb}
If $\Gamma$ is fair and $\EEw = \varnothing$ then $\RRw$ is a
complete presentation of $\EE_0$.
\hfill\isaforlink{Ordered_Completion}{lem:Ew_empty_CR_Rw_gtotal}
\end{thm}
\begin{proof}
We have $\PCP(\RRw) \subseteq \PCP_>(\ERw)$ since
$\RRw \subseteq {>}$. Hence the result follows from fairness and
\thmref{complete presentation okb1}.
\end{proof}

\begin{exa}
Consider the ES $\EE$ from \exaref{okb1} and $\LPO$ with
precedence $\m{+} > \m{0}$.
After two \tsfs{orient} steps, we apply \tsfs{deduce}:
\begin{xalignat*}{4}
\m{1} \cdot (-x + x) &\to \m{0} &
\m{1} \cdot (x + -x) &\to x + -x &
-x + x &\approx y + -y &
\m{1} \cdot (-z + z) &\approx x + -x
\intertext{The newly added equation is simplified and then oriented:}
\m{1} \cdot (-x + x) &\to \m{0} &
\m{1} \cdot (x + -x) &\to x + -x &
-x + x &\approx y + -y &
x + -x &\to \m{0}
\intertext{Using the new rewrite rule, the remaining equation is
simplified, the second rule is subjected to \tsfs{compose} and
subsequently to \tsfs{collapse}:}
\m{1} \cdot (-x + x) &\to \m{0} &
\m{1} \cdot \m{0} &\approx \m{0} &
-x + x &\approx \m{0} &
x + -x &\to \m{0}
\intertext{Orienting both equations results in:}
\m{1} \cdot (-x + x) &\to \m{0} &
\m{1} \cdot \m{0} &\to \m{0} &
-x + x &\to \m{0} &
x + -x &\to \m{0}
\intertext{At this point the first rule is collapsed using the third rule,
and subsequently oriented (into an existing rule):}
&&
\m{1} \cdot \m{0} &\to \m{0} &
-x + x &\to \m{0} &
x + -x &\to \m{0}
\end{xalignat*}
This sequence can be extended to an infinite run by repeatedly
adding (using \tsfs{deduce}) and deleting the trivial equation
$\m{0} \approx \m{0}$.
Then the set of persistent rules $\RRw$ coincides with the TRS
$\RR$ from \exaref{okb1}, and $\EEw = \varnothing$.
\end{exa}

The final result in this section is in the spirit of
\thmref{Metivier normalization equivalent} but for ordered completion,
showing that a ground-complete system can be interreduced to some extent.

\begin{defi}%
\label{def:okb reduced transformation}
Given a ground-complete system $\SS = \EE^> \cup \RR$, we define
\begin{align*}
\RR' &= \{ \ell \to r \mid \text{$\ell \to r \in \dot{\QQ}$
and
$\ell \in \NF(\xrightarrow{\pe}_\SS)$} \}
\\
\EE' &= \{ s{\downarrow_{\RR'}} \approx t{\downarrow_{\RR'}} \mid
s \approx t \in \EE \} \setminus {=}
\end{align*}
where $\QQ = \RR \cup (\EE^\pm \cap {>})$
and $\dot{\QQ}$ is defined in \defref{reduced transformation}.
\end{defi}

Here we write $t \xrightarrow{\pe}_{\SS} u$ if there are
a rule $\ell \to r \in \SS$, a context $C$,
and a substitution $\sigma$ such that
$t = C[\ell\sigma]$, $u = C[r\sigma]$, and $t \prencompasses \ell$.
For example, if $\EE$ is empty and $\RR$ consists
of the single rule $\m{f}(x,y) \to \m{g}(x)$ we have
$\m{f}(y,z) \in \NF(\xrightarrow{\pe}_\SS)$, but
$\m{f}(\m{g}(x),y) \notin \NF(\xrightarrow{\pe}_\SS)$ and
$\m{f}(x,x) \notin \NF(\xrightarrow{\pe}_\SS)$.

\begin{thm}%
\label{thm:reduced okb}
If $\SS = \EE^> \cup \RR$ is ground-complete then
$\SS' = \EE'^> \cup \RR'$ is ground-complete and normalization and
conversion equivalent on ground terms.
\hfill
\isaforlink{Ordered_Completion}{lem:reduced_ground_complete}
\end{thm}
\begin{proof}
We first show $\NF(\SS')\subseteq \NF(\SS)$.
For a rule $\ell \to r \in \SS$, let $b_{\ell \to r}$ be
$\bot$ if $\ell \to r \in \QQ$ and $\top$ otherwise.
We prove $\ell \notin \NF(\SS')$
for every rule $\ell \to r \in \SS$, by induction on
$(\ell,b_{\ell \to r})$ with respect to the lexicographic
combination of $\prencompasses$ and the order where $\top > \bot$.
\begin{itemize}
\item
If $\ell \to r \in \QQ$ two cases can be distinguished.
If $\ell \notin \NF(\xrightarrow{\pe}_\SS)$ then
$\ell \prencompasses \ell'$ for some rule $\ell' \to r' \in \SS$
and thus $\ell' \notin \NF(\SS')$ by the induction hypothesis.
Hence also $\ell \notin \NF(\SS')$.
If $\ell \in \NF(\xrightarrow{\pe}_\SS)$ then,
by construction of $\RR'$, there is
some rule $\ell \to r' \in \RR'$ (modulo renaming), so
$\ell \notin \NF(\SS')$.
\smallskip
\item
If $\ell \to r \notin \QQ$ then $\ell = u\sigma$ and $r = v\sigma$
for some equation $u \approx v \in \EE^\pm$ and substitution $\sigma$
such that $\ell > r$. We distinguish two cases. First, if
$u \in \NF(\RR')$ then $u = u{\downarrow_{\RR'}}$. We have
$\ell > r \geqslant v{\downarrow_{\RR'}}\sigma$ because
$\RR' \subseteq {>}$ and hence
$u \neq v{\downarrow_{\RR'}}$. It follows that
$u \approx v{\downarrow_{\RR'}} \in \EE'^\pm$ and thus
$\ell \to v{\downarrow_{\RR'}}\sigma \in \EE'^\succ$.
Hence $\ell \notin \NF(\SS')$.
Second, if $u \notin \NF(\RR')$ then
$u \notin \NF(\dot{\QQ})$ since $\RR' \subseteq \dot{\QQ}$.
So there exists a rule $\ell' \to r' \in \QQ$ such that
$u \encompasses \ell'$. Clearly $\ell \encompasses \ell'$.
Since $\ell \to r \notin \QQ$, the induction hypothesis yields
$\ell' \notin \NF(\SS')$. Hence also $\ell \notin \NF(\SS')$.
\end{itemize}
We next establish the inclusion ${\to_{\SS'}} \subseteq {\conv_{\SS}}$ on
ground terms. We have $\EE' \cup \RR' \subseteq {\conv_{\ER}}$
by construction.
For ground terms $s$ and $t$, a step $s \to_{\SS'} t$ implies
$s \fromto_{\EE' \cup \RR' } t$ and hence existence
of a conversion $s \conv_{\ER} t$. We can also obtain such a
conversion where all intermediate terms are ground by replacing every
variable with some ground term. Since the reduction order $>$ is
ground-total, ${\to_{\ER}} \subseteq {\fromto_{\SS}^=}$ holds on
ground terms. Hence there is a conversion $s \conv_{\SS} t$.

Moreover, the system $\SS'$ is clearly terminating as it is included
in $>$. Thus the result follows from \lemref[b]{reduced abstract}, viewing
$\SS$ and $\SS'$ as ARSs on ground terms.
\end{proof}

We illustrate the transformation of
\defref{okb reduced transformation} on a concrete example.

\begin{exa}
Consider the following system with $\RR$ consisting of one rule and
$\EE$ consisting of three equations:
\begin{xalignat*}{3}
\m{s}(\m{s}(x)) + \m{s}(x) &\to \m{s}(x) + \m{s}(\m{s}(x)) &
x + \m{s}(y) &\approx \m{s}(x + y) &
x + y &\approx y + x \\
&&
\m{s}(x) + y &\approx \m{s}(x + y)
\end{xalignat*}
It is ground-complete for the lexicographic path order~\cite{KL80} with
${+} > \m{s}$ as precedence.
We have
$\QQ = \RR \cup \{
x + \m{s}(y) \to \m{s}(x + y),~
\m{s}(x) + y \to \m{s}(x+y)\}$.
Since the term $\m{s}(\m{s}(x)) + \m{s}(x)$ is reducible by the rule
$\m{s}(x) + x \to x + \m{s}(x) \in \SS$ and
$\m{s}(\m{s}(x)) + \m{s}(x) \pe \m{s}(x) + x$,
the rule of $\RR$ does not remain in $\RR'$. Hence,
$\RR' = \{ x + \m{s}(y) \to \m{s}(x + y),~ \m{s}(x) + y \to \m{s}(x+y) \}$
and $\EE' = \{ x + y \approx y + x \}$.
\end{exa}

One may wonder whether $\RR'$ can simply be defined as
$\ddot{\QQ}$ instead of imposing a strict encompassment condition.
The following example shows that this destroys reducibility.

\begin{exa}
Consider the following system where $\RR$ consists of
two rules and $\EE$ consists of one equation:
\begin{xalignat*}{3}
\m{f}(x,y) &\to \m{g}(x) &
\m{f}(x,y) &\to \m{g}(y) &
\m{g}(x) &\approx \m{g}(y)
\end{xalignat*}
Then $\EE^> \cup \RR$ is ground-complete if $>$ is the lexicographic
path order with $\m{f} > \m{g}$ as precedence.
We have $\RR' = \dot{\QQ} = \QQ = \RR$ and $\EE' = \EE$ but
$\ddot{\QQ} = \varnothing$.

Note that we obtain an equivalent ground-complete system if
we add, for instance, an equation $\m{g}(\m{g}(x)) \approx \m{g}(y)$.
This shows that even systems which are simplified with respect to the
procedure suggested by \thmref{reduced okb} are not unique.
\end{exa}

This section resumes our results on ordered completion~\cite{HMSW17}.
Like in Sections~\ref{sec:finite runs} and~\ref{sec:infinite runs}, our
proofs deviate from the standard approach~\cite{BDP89} in that we avoid
proof orders in favor of different, simpler orderings as required,
together with source decreasingness. Again, we also support prime
critical pairs.
For \thmref{complete presentation okb} and the interreduction result of
\thmref{reduced okb} we are not aware of earlier references in the
literature.

\section{Completeness Results for Ordered Completion}%
\label{sec:completeness}

Ordered completion never fails and its limit always constitutes a
ground-complete system.
On the other hand, if there is a \emph{complete} presentation that is
compatible with the employed reduction order, does ordered completion
also produce a complete presentation, ending with $\EEw = \varnothing$?
In this section we revisit two results from the literature which
provide sufficient conditions for ordered completion to
always derive a complete system, independent of the strategy
employed by a completion procedure.
In \ssecref{ground total} we reprove the result by Bachmair, Dershowitz,
and Plaisted for the case where the reduction order is ground
total~\cite{BDP89}.
The corresponding result by Devie~\cite{D91} for linear systems is
considered in~\ssecref{linear}.

\subsection{Ground-Total Orders}%
\label{ssec:ground total}

In this subsection we consider a fair run $\Gamma$ of ordered completion
\[
(\EE_0, \RR_0) ~\vdo~ (\EE_1, \RR_1) ~\vdo~ (\EE_2, \RR_2) ~\vdo~ \cdots
\]
with respect to a ground-total reduction order $>$. If
$\EEw = \varnothing$ then the TRS $\RRw$ is a
complete presentation of $\EE_0$ by
\thmref{complete presentation okb}.
According to Bachmair \etal~\cite[Theorem 2]{BDP89}, under
certain conditions fair runs always conclude with $\EEw = \varnothing$
whenever there exists a complete presentation of $\EE_0$ compatible with
$>$.
In the remainder of this subsection we give a formalized proof of
this result.
Like the original proof, it is based on the idea
that ground-completeness of $\RRw$ is preserved under signature
extension with constants.
Let $\KK$ be a set of different fresh constants $\hat{x}$ for every
variable $x \in \VV$.
We first show that the reduction order $>$ can be extended to
a ground-total order on the signature augmented by $\KK$ such that
minimum constants are preserved.

\begin{lem}%
\label{lem:gtK}
There exists a ground-total reduction order $>^\KK$ on
$\TT(\FF \cup \KK, \VV)$ such that ${>} \subseteq {>^\KK}$ and the
minimum constant with respect to $>$ is also minimum in $>^\KK$.
\hfill
\isaforlink{Ordered_Completion}{lem:less_sk}
\end{lem}
\begin{proof}
Let $\bot \in \FF$ be the minimum constant with respect to $>$.
We consider the KBO $\KBO[\sqsupset]$ with weights
$w_0 = 1$ and $w(f) = 1$ for all $f \in \FF \cup \KK$ together with
a precedence $\sqsupset$ which is total on $\FF \cup \KK$, has $\bot$ as
the minimum element, and satisfies $\hat{x} \sqsupset f$ for all
$f \in \FF$ and $\hat{x} \in \KK$.
Given a term $t \in \TT(\FF \cup \KK, \VV)$, we write
$t_\bot$ for the term
obtained from $t$ by replacing every constant in $\KK$
with $\bot$.
Furthermore, we define $s >^\KK t$ as $s_\bot > t_\bot$, or both
$s_\bot = t_\bot$ and $s \sqsupset_\kbo t$.
We show that $>^\KK$ is a ground-total reduction order with the
stated properties. Ground totality of $>^\KK$ follows from ground
totality of $\KBO[\sqsupset]$ given the total precedence.
Well-foundedness holds by construction as a lexicographic combination of
well-founded relations. Closure under substitutions is satisfied because
it holds for both $>$ and $\KBO[\sqsupset]$, and $s_\bot = t_\bot$ implies
$s\sigma_\bot = t\sigma_\bot$. Similar arguments apply to closure under
contexts and transitivity. By construction of $\sqsupset$ and the
definition of $>^\KK$, the constant $\bot$ is still minimal.
Moreover $>^\KK$ extends $>$ because $s > t$ implies
$s, t \in \TT(\FF,\VV)$, so
$s_\bot = s > t = t_\bot$ and hence $s >^\KK t$.
\end{proof}

We write $\hat{t}$ for the ground term that is obtained from $t$ by
replacing every variable $x$ by the constant $\hat{x}$.
In the next lemma we verify some basic properties related to this
\emph{grounding} operation.

\begin{lem}%
\label{lem:hat}
Let $\RR$ be a TRS over a signature $\FF$ and let $s, t \in \TT(\FF,\VV)$.
\begin{enumerate}
\item%
\label{lem:hat:1}
If $s > t$ then $\hat{s} >^\KK \hat{t}$.
\hfill
\isaforlink{Ordered_Completion}{lem:sk_less_compat}

\smallskip
\item%
\label{lem:hat:2}
Suppose $s \neq t$. Then $s \to_\RR t$ if and only if
$\hat{s} \to_\RR \hat{t}$.
\hfill
\isaforlink{Ordered_Completion}{lem:rstep_step_sk}
\end{enumerate}
\end{lem}
\begin{proof}
\hfill
\begin{enumerate}
\item
Suppose $s > t$. \lemref{gtK} yields $s >^\KK t$ and, because
$>^\KK$ is closed under substitutions, $\hat{s} >^\KK \hat{t}$.
\item We consider the two implications separately.
\begin{itemize}
\item
If $s \to_\RR t$ then $\Var(t) \subseteq \Var(s)$.
Let $\sigma$ be a substitution such that $\hat{s} = s\sigma$.
We have $\hat{t} = t\sigma$ and thus
$\hat{s} = s\sigma \to_\RR t\sigma = \hat{t}$.
\smallskip
\item
Conversely, if $\hat{s} \to_\RR \hat{t}$ then $\hat{s}|_p = \ell\sigma$
and $\hat{t} = \hat{s}{[r\sigma]}_p$ for some rule $\ell \to r \in \RR$,
position $p$, and substitution $\sigma$.
We denote the substitution
$\{ x \mapsto \phi(\sigma(x)) \mid x \in \VV \}$
by $\sigma_\phi$. Here $\phi(u)$ denotes the term obtained from
$u$ after replacing every constant $\hat{x}$ of $\KK$ by $x$.
Because $s|_p = \phi(\hat{s}|_p) = \phi(\ell\sigma) = \ell\sigma_\phi$
and $t = \phi(\hat{t}) = \phi(\hat{s}{[r\sigma]}_p) = s{[r\sigma_\phi]}_p$,
we obtain $s \to_\RR t$ as desired.
\qedhere
\end{itemize}
\end{enumerate}
\end{proof}

\noindent
It is not hard to see that the TRS $\SSw$ still constitutes a
ground-complete presentation of $\EE_0$ when considered over the
extended signature, as shown below.

\begin{lem}%
\label{lem:signature extension}
The TRS $\SSw$ is ground-complete over $\FF \cup \KK$
and ${\fromto_{\ERw}^*} = {\fromto_{\EE_0}^*}$.
\hfill
\isaforlink{Ordered_Completion}{lem:correctness_okb_sk}
\end{lem}
\begin{proof}
Since $>^\KK$ contains $>$ by \lemref{gtK}, the run $\Gamma$ is also a
valid run with respect to $>^\KK$. It is moreover fair since
${>} \subseteq {>^\KK}$ implies
$\PCP_{>^\KK}(\EE) \subseteq \PCP_{>}(\EE)$ for any set of equations
$\EE$, by \defref{extended overlap}.
Hence the result follows from \thmref{KBo correctness}.
\end{proof}

An important observation for the completeness proof below is that
normal forms with respect to the final system $\SSw$ and
with respect to the union $\SSi$ of intermediate systems
coincide, as shown below.

\begin{lem}%
\label{lem:reducibility}
The inclusion $\NF(\SSw) \subseteq \NF(\SSi)$ holds.
\hfill
\isaforlink{Ordered_Completion}{lem:NF_Sw_subset_NF_Sinf}
\end{lem}
\begin{proof}
The result is an immediate consequence of the following two claims:
\begin{enumerate}
\item[(a)]%
\label{lem:ERw_reducibility:E}
If $\ell \approx r \in \EEi^\pm$, $\ell \in \NF(\RRi)$, and
$\ell\sigma > r\sigma$ then $\ell\sigma \notin \NF(\EEw^>)$.
\smallskip
\item[(b)]%
\label{lem:ERw_reducibility:R}
If $\ell \to r \in \RRi$ then $\ell \notin \NF(\SSw)$.
\end{enumerate}
For claim (a) we use induction on $\{ \ell, r \}$ with respect to
$\succ_\mul$. If $\ell \approx r \in \EEw^\pm$ the result is immediate.
Otherwise, $\ell \approx r \in \EE_i \setminus \EE_{i+1}$ or
$r \approx \ell \in \EE_i\setminus \EE_{i+1}$
for some $i \geqslant 0$. Without loss of generality we assume the
former since the latter case is similar. From
\lemref{okb inclusions infinite}(\ref{lem:okb inclusions infinite:E}) 
we obtain $\ell
\mathrel{{(\to_{\SS_{i+1}}^{\pe_1} \cdot \mathrel{\EE_{i+1}}^\pm)}^\pm} r$,
$\ell \to r \in \RR_{i+1}$, $r \to \ell \in \RR_{i+1}$, or $\ell = r$.
The latter two cases are impossible because of the assumption
$\ell\sigma > r\sigma$ and the inclusion $\RR_{i+1} \subseteq
\RR_\infty \subseteq {>}$. Also $\ell \to r \in \RR_{i+1}$ is
impossible because of the assumption $\ell \in \NF(\RRi)$.
\begin{itemize}
\item
Suppose $\ell \to_{\SS_{i+1}}^{\pe_1} u$ and
$u \approx r \in \EE_{i+1}^\pm$ for some term $u$. The step
$\ell \to_{\SS_{i+1}} u$ cannot use a rule in $\RR_{i+1}$
because $\ell \in \NF(\RRi)$. So there must be an equation
$\ell' \approx r' \in \EE_{i+1}^\pm$, a substitution $\tau$, and a
position $p$ in $\ell$ such that $\ell|_p = \ell'\tau$, $u|_p = r'\tau$,
$\ell'\tau > r'\tau$, and $\ell \prencompasses \ell'$. Because of
$\ell \encompasses \ell'\tau > r'\tau \encompasses r'$ we have
$\ell \succ r'$, and therefore $\{ \ell, r \} \succ_\mul \{ \ell', r' \}$.
Moreover, $\ell' \in \NF(\RRi)$. The induction hypothesis yields
$\ell'\tau \notin \NF(\EEw^>)$. Since $\ell \encompasses \ell'\tau$, we
have $\ell \notin \NF(\EEw^>)$ and thus also
$\ell\sigma \notin \NF(\EEw^>)$.
\smallskip
\item
In the remaining case we have $r \to_{\SS_{i+1}}^{\pe_1} u$ and
$u \approx \ell \in \EE_{i+1}^\pm$ for some term $u$. We have $r > u$
and thus also
$r \succ u$ and $\{ \ell, r \} \succ_\mul \{ \ell, u \}$.
Because $\ell\sigma > r\sigma > u\sigma$, the result follows from the
induction hypothesis.
\end{itemize}
For claim (b) we use induction on $(\ell,r)$ with respect to
$\succ_\lex$. If $\ell \to r \in \RRw$ then $\ell \notin \NF(\SSw)$
trivially holds. Otherwise, $\ell \to r \in \RR_i \setminus \RR_{i+1}$
for some $i \geqslant 0$. From
\lemref{okb inclusions infinite}(\ref{lem:okb inclusions infinite:R}) 
we obtain $\ell \to_{\SS_{i+1}}^{\pe_2} \cdot \mathrel{\EE_{i+1}} r$
or $\ell \mathrel{\RR_{i+1}} \cdot \FromB{\SS_{i+1}} r$. In the latter
case there is a term $u$ such that $\ell \to u \in \RR_{i+1}$ and
$r \to_{\SS_{i+1}} u$. Since this implies $r > u$ and thus
$(\ell,r) \succ_\lex (\ell,u)$, we obtain $\ell \notin \NF(\SSw)$ from
the induction hypothesis. In the former case there is a term $u$ such that
$\ell \to_{\SS_{i+1}}^{\pe_2} u$ and $u \approx r \in \EE_{i+1}$.
If the step $\ell \to_{\SS_{i+1}} u$ uses a rule
$\ell' \to r' \in \RR_{i+1}$ then the result follows from the induction
hypothesis because $\ell \prencompasses \ell'$ implies
$(\ell,r) \succ_\lex (\ell',r')$, and $\ell' \notin \NF(\SSw)$ implies
$\ell \notin \NF(\SSw)$. Otherwise, there exist an equation
$\ell' \approx r' \in \EE_{i+1}$, a position $p$ in $\ell$, and a
substitution $\sigma$ such that $\ell|_p = \ell'\sigma$,
$r|_p = r'\sigma$, $\ell'\sigma > r'\sigma$, and
$\ell \prencompasses \ell'$. If $\ell' \in \NF(\RRi)$ then we obtain
$\ell'\sigma \notin \NF(\EEw^>)$ from claim (a) and thus
$\ell \notin \NF(\SSw)$ because $\ell \encompasses \ell'\sigma$. If
$\ell' \notin \NF(\RRi)$ then there exists some rule
$\ell'' \to r'' \in \RRi$ such that $\ell' \encompasses \ell''$. In this
case we have $\ell \succ \ell''$ and thus
$(\ell,r) \succ_\lex (\ell'',r'')$. We obtain
$\ell'' \notin \NF(\SSw)$ from the induction hypothesis. Hence also
$\ell \notin \NF(\SSw)$.
\end{proof}

\begin{cor}%
\label{cor:normalization equivalence of Si and Sw}
The identity $\NF(\SSw) = \NF(\SSi)$ holds.
\hfill
\isaforlink{Ordered_Completion}{lem:NF_Sw_NF_Sinf}
\end{cor}
\begin{proof}
We obtain $\NF(\SSi) \subseteq \NF(\SSw)$ from the
inclusion ${\to_{\SSw}} \subseteq {\to_{\SSi}}$ and hence the result
follows from \lemref{reducibility}.
\end{proof}

Hereafter we assume that there is a complete presentation $\RR$ of
$\EE_0$ with $\RR \subseteq {>}$.
We next show that grounded terms which are $\SSw$-normal forms
are also $\RR$-normal forms.

\begin{lem}%
\label{lem:R_NFs}
If $\hat{t} \in \NF(\SSw)$ then $\hat{t} \in \NF(\RR)$.
\hfill
\isaforlink{Ordered_Completion}{lem:NF_S_NF_R}
\end{lem}
\begin{proof}
Suppose $\hat{t} \in \NF(\SSw)$ but $\hat{t} \notin \NF(\RR)$, so
$\hat{t} \to_\RR u$ for some term $u$.
Since $\hat t$ is ground and $\RR$ is terminating, also $u$ is
ground. We obtain $\hat{t} \join_\SSw u$ from the ground-completeness of
$\SSw$ (\lemref{signature extension}).
Since $\hat{t} \mathrel{>^\KK} u$ by the global assumption
$\RR \subseteq {>}$ and \lemref[1]{hat},
the joining sequence cannot be of the form
$\hat{t} \From{\SSw}{*} u$ as this would imply $u \geqslant \hat{t}$
and thus $u \geqslant^\KK \hat{t}$,
contradicting the well-foundedness of $>^\KK$.
Therefore we must have
$\hat{t} \to^+_\SSw \cdot \From{\SSw}{*} u$ which means that
$\hat{t}$ is reducible in $\SSw$,
contradicting the assumption $\hat{t} \in \NF(\SSw)$.
\end{proof}

The preliminary results collected so far now lead to the following
key observation: If a grounded term $\hat{s}$ is reducible then so
is its (possibly non-ground) counterpart $s$.
In \lemref{oKBi_NF_Sw_subset_NF_R} below we can then connect
$\RR$-reducibility to $\SSw$-reducibility for terms over the original
signature.

\begin{lem}%
\label{lem:oKBi_hat_Si_step_Si_step}
If $\hat{s} \to_\SSi \hat{t}$ then $s \notin \NF(\SSi)$.
\hfill
\isaforlink{Ordered_Completion}{lem:sk_Swinf_step_imp_no_Sinf_NF}
\end{lem}
\begin{proof}
There exist an equation $\ell \approx r \in \EEi^\pm \cup \RRi$, a
position $p$, and a substitution $\sigma$ such that
$\hat{s}|_p = \ell\hat{\sigma}$,
$\hat{t} = \hat{s}{[r\hat{\sigma}]}_p$, and
$\ell\hat{\sigma} >^\KK r\hat{\sigma}$. We perform induction on $\hat{t}$
with respect to $>^\KK$. If $p \neq \epsilon$ then
$\hat{t} \rhd r\hat{\sigma}$ and thus
$\hat{t} >^\KK r\hat{\sigma}$ because $>^\KK$ is a ground-total
reduction order.
The induction hypothesis yields $\ell\sigma \notin \NF(\SSi)$, which
implies $s \notin \NF(\SSi)$. So in the following we assume that the
step $\hat{s} \to_\SSi \hat{t}$ takes place at the root position.
If $s > t$ then $s \to t \in \SSi$, from which the claim is immediate.
This covers the case $\ell \approx r \in \RRi$, so if
$s \not> t$ then $\ell \approx r \in \EEi^\pm$.
We distinguish two cases, $\hat{t} \in \NF(\SSi)$ and
$\hat{t} \notin \NF(\SSi)$.
\begin{itemize}
\item
If $\hat{t} \in \NF(\SSi)$ then $\hat{t} \in \NF(\SSw)$
by \corref{normalization equivalence of Si and Sw} and thus
$\hat{t} \in \NF(\RR)$ by \lemref{R_NFs}. From
\lemref{signature extension} and the fact that $\RR$ is a complete
presentation of $\EE_0$ we obtain
$\hat{s} \to_\RR^+ \hat{t}$.  The latter implies $s \to_\RR^+ t$
by \lemref[2]{hat} and thus $s > t$, which is a contradiction.
\smallskip
\item
Suppose $\hat{t} \notin \NF(\SSi)$.
We distinguish two further cases, depending on whether or not
$\ell \approx r$ belongs to $\EEw^\pm$.

\smallskip

If $\ell \approx r \notin \EEw^\pm$ then
$\ell \approx r \in {(\EE_i \setminus \EE_{i+1})}^\pm$
for some $i \geqslant 0$.
From \lemref{okb inclusions infinite}(\ref{lem:okb inclusions infinite:E}) 
we obtain
$\ell \mathrel{{(\to_{\SS_{i+1}} \cdot
\mathrel{\EE_{i+1}^\pm})}_{\phantom{i}}^\pm} r$,
$\ell \to r \in \RR_{i+1}$, $r \to \ell \in \RR_{i+1}$, or $\ell = r$.
The last two cases contradict $\hat{s} >^\KK \hat{t}$. If
$\ell \mathrel{\to_{\SS_{i+1}} \cdot \mathrel{\EE_{i+1}^\pm}} r$ or
$\ell \to r \in \RR_{i+1}$ then $\ell \notin \NF(\SSi)$ and thus
$s = \ell\sigma \notin \NF(\SSi)$. Otherwise, $r \to_{\SS_{i+1}} u$
for some term $u$ with $u \approx \ell \in \EE_{i+1}^\pm$.
We have $s = \ell\sigma \xlr{}_{\EEi} u\sigma
\FromB{\SSi} r\sigma = t$ and thus
$\hat{s} = \ell\hat{\sigma} \to_\SSi u\hat{\sigma} = \widehat{u\sigma}$
and $\hat{t} = r\hat{\sigma} >^\KK u\hat{\sigma}$.
The induction hypothesis yields $s \notin \NF(\SSi)$.

\smallskip

In the second case we assume $\ell \approx r \in \EEw^\pm$.
From the assumption
$\hat{t} \notin \NF(\SSi)$ we obtain a term $u$ such that
$\hat{t} \to_\SSi \hat{u}$. We have
$\hat{t} >^\KK \hat{u}$ and thus $t \notin \NF(\SSw)$ by the
induction hypothesis.
Consider an innermost $\SSw$-step starting from
$t$, say $t \ito_\SSw v$,
such that there exists a peak
\begin{align}
\label{eq:oKBi_hat_Si_step_Si_step_peak}
\tag{$\star$}
s = \ell\sigma \xlr{}_{\ell \approx r}
r\sigma = t \to_{\SSw}^{q} v
\end{align}
with $l\sigma = s \not> t = r\sigma$.
If the two steps form an overlap then we have
$s \fromto_{\PCP_>(\ERw)} v$ since $t \ito_\SSw v$
is innermost, and thus $s \fromto_{\EEi} v$ or $s \join_\SSw v$
is obtained from the fairness of the run.
In the former case, since
$\hat{s} >^\KK \hat{t} >^\KK \hat{v}$, we have
$\hat{s} \to_\SSi \hat{v}$ and thus the induction hypothesis applies.
If on the other hand $s \join_\SSw v$ then
we cannot have $v \to_\SSw^* s$ as this would imply $v \geqslant s$,
contradicting $\hat{s} >^\KK \hat{v}$ because $>$ and $>^\KK$
are compatible by \lemref[1]{hat}. So $s$ must be $\SSw$-reducible.

Otherwise, the peak~\eqref{eq:oKBi_hat_Si_step_Si_step_peak}
constitutes a variable overlap, so there is some variable
$x \in \Var(r)$ and positions $q_1$ and $q_2$ such that $r|_{q_1} = x$ and
$q = q_1q_2$. If $x \notin \Var(\ell)$ then
$s \xlr{}_{\EEi} v$ and the induction hypothesis applies
as before. Otherwise, $s = \ell\sigma$ is reducible in $\SSi$.
\qedhere
\end{itemize}
\end{proof}

\begin{lem}%
\label{lem:oKBi_NF_Sw_subset_NF_R}
The inclusion $\NF(\SSw) \subseteq \NF(\RR)$ holds.
\hfill
\isaforlink{Ordered_Completion}{lem:NF_Sw_subset_NF_R}
\end{lem}
\begin{proof}
Suppose $t \to_\RR u$, so $t > u$ and thus also
$\hat{t} >^\KK \hat{u}$. From \lemref{signature extension} we obtain
$\hat{t} \join_\SSw \hat{u}$.
Like in the proof of \lemref{R_NFs},
$\hat{u} \to_\SSw^* \hat{t}$ would imply
$\hat{u} \geqslant^\KK \hat{t}$, contradicting well-foundedness
of $>^\KK$. Therefore the joining sequence must be of the shape
$\hat{t} \to^{+}_\SSw \cdot \From{\SSw}{*} \hat{u}$ and thus
$\hat{t}$ is reducible in $\SSw$ and also in $\SSi$ because
$\SSw \subseteq \SSi$.
\lemref{oKBi_hat_Si_step_Si_step} yields $t \notin \NF(\SSi)$
and thus $t \notin \NF(\SSw)$
by \corref{normalization equivalence of Si and Sw}.
\end{proof}

We call $(\EEw,\RRw)$ \emph{simplified} if $\RRw$ is reduced and
all equations in $\EEw$ are irreducible with respect to $\RRw$,
unorientable with respect to $>$, and {non-trivial}.
We use the following auxiliary result before proving the main
completeness theorem.

\begin{lem}%
\label{lem:oKBi_NF_Rw_subset_NF_Sw}
If $(\EEw,\RRw)$ is simplified then $\NF(\RRw) \subseteq \NF(\SSw)$.
\hfill
\isaforlink{Ordered_Completion}{lem:NF_Rw_subset_NF_Sw}
\end{lem}
\begin{proof}
Suppose $(\EEw,\RRw)$ is simplified and let $s \in \NF(\RRw)$.
We prove $s \in \NF(\SSw)$ by induction on $s$ with respect to $\succ$.
If $s \notin \NF(\SSw)$ then there exist an equation
$\ell \approx r \in \EEw^\pm$, a context $C$, and a substitution $\sigma$
such that $s = C[\ell\sigma]$ and $\ell\sigma > r\sigma$.  Since
$s \encompasses \ell\sigma > r\sigma \encompasses r$ implies $s \succ r$
and $r \in \NF(\RRw)$ by the assumption that $(\EEw,\RRw)$ is simplified,
the induction hypothesis yields $r \in \NF(\SSw)$.
Hence $r \in \NF(\RR)$ by \lemref{oKBi_NF_Sw_subset_NF_R}. Since
$\RR$ is a complete presentation of $\EE_0$, we have $\ell \join_\RR r$
and thus $\ell \to_\RR^* r$. From $\RR \subseteq {>}$ we infer
$\ell > r$ or $\ell = r$, contradicting the
assumption that $(\EEw,\RRw)$ is simplified.
\end{proof}

\begin{thm}%
\label{thm:KBo completeness}
If $(\EEw,\RRw)$ is simplified then $\EEw = \varnothing$ and
$\RRw$ is literally similar to $\RR$.
\hfill
\isaforlink{Ordered_Completion}{lem:completeness}
\end{thm}
\begin{proof}
Suppose $\EEw \neq \varnothing$ and let $s \approx t$ be an equation in
$\EEw$. The terms $s$ and $t$ are $\RR$-normal forms by
the assumption that $(\EEw,\RRw)$ is simplified in combination with
Lemmata~\ref{lem:oKBi_NF_Sw_subset_NF_R} and~\ref{lem:oKBi_NF_Rw_subset_NF_Sw}.
Since $\RR$ is a complete presentation of $\EE_0$, we obtain $s = t$,
contradicting the assumption that $(\EEw,\RRw)$ is simplified.
Hence, $\EEw = \varnothing$ and therefore $\RRw$ is a complete
presentation of $\EE_0$ by \thmref{complete presentation okb}.
Since $(\EEw,\RRw)$ is simplified, $\RRw$ is even a canonical
presentation of $\EE_0$. As $\RRw \subseteq {>}$, literal similarity
of $\RRw$ and $\RR$ is concluded by \thmref{Metivier 2b}.
\end{proof}

A run of ordered completion is called \emph{simplifying} if its limit
$(\EEw,\RRw)$ is simplified.

\begin{exa}
Consider again the ES $\EE$ from \exaref{okb1} and its complete
presentation $\RR$, which cannot be derived using standard completion.
Termination of $\RR$ can be shown by a suitable KBO\@.
Thus, by \thmref{KBo completeness} any fair and simplifying run of
ordered completion on $\EE$ using the same order will succeed with a
variant of $\RR$, independent of the employed strategy.
\end{exa}

The results in this subsection are due to Bachmair, Dershowitz,
and Plaisted~\cite{BDP89}. However, our proof is structured into many
preliminary results, as opposed to the monolithic original version, and
we fill in numerous details omitted in the original version.

\subsection{Linear Systems}%
\label{ssec:linear}

The previously presented correctness and completeness results
(Theorems~\ref{thm:KBo correctness} and~\ref{thm:KBo completeness}) do not
state any properties of the system obtained when running \KBo\ with
a reduction order that is not ground-total. The following example
from Devie~\cite{D91} shows that the restriction to ground-total
orders can actually be severe.

\begin{exa}
Consider the ES $\EE$ consisting of the following equations:
\begin{xalignat*}{3}
\m{f}_1(\m{g}_1(\m{i}_1(x)))
 &\approx \m{g}_1(\m{i}_1(\m{f}_1(\m{g}_1(\m{i}_2(x))))) &
\m{h}_1(\m{g}_1(\m{i}_1(x)))
 &\approx \m{g}_1(\m{i}_1(x)) &
\m{f}_1(\m{a}) &\approx \m{a}\\
\m{f}_2(\m{g}_2(\m{i}_2(x)))
 &\approx \m{g}_2(\m{i}_2(\m{f}_2(\m{g}_2(\m{i}_1(x))))) &
\m{h}_2(\m{g}_2(\m{i}_2(x)))
 &\approx \m{g}_2(\m{i}_2(x)) &
\m{f}_2(\m{a}) &\approx \m{a}\\
\m{g}_1(\m{a}) &\approx \m{a} &
\m{h}_1(\m{a}) &\approx \m{a} &
\m{i}_1(\m{a}) &\approx \m{a} \\
\m{g}_2(\m{a}) &\approx \m{a} &
\m{h}_2(\m{a}) &\approx \m{a} &
\m{i}_2(\m{a}) &\approx \m{a}
\end{xalignat*}
When orienting all equations from left to right we obtain a TRS $\RR$
which is easily shown to be terminating by automatic tools. As all
critical pairs are joinable it is confluent, and thus
canonical since it is also reduced.
However, $\RR$ cannot be oriented by any ground-total reduction order $>$.
We have $\m{i}_1(\m{a}) \conv_{\EE} \m{i}_2(\m{a})$ but neither
$\m{i}_2(\m{a}) > \m{i}_1(\m{a})$ nor $\m{i}_1(\m{a}) > \m{i}_2(\m{a})$
can hold; using the rule $\m{f}_1(\m{g}_1(\m{i}_1(x))) \to
\m{g}_1(\m{i}_1(\m{f}_1(\m{g}_1(\m{i}_2(x)))))$,
the former would imply
\[
\m{f}_1(\m{g}_1(\m{i}_2(\m{a}))) > \m{f}_1(\m{g}_1(\m{i}_1(\m{a})))
 > \m{g}_1(\m{i}_1(\m{f}_1(\m{g}_1(\m{i}_2(\m{a})))))
\]
which contradicts well-foundedness, and for the latter a
similar argument applies. As a matter of fact, in~\cite{D91}
it is shown that any \KBo\ run starting from $\EE$ and using a
\emph{ground-total} reduction order will fail to generate a finite result.
\end{exa}

Devie~\cite{D91} gives a second sufficient condition for an ordered
completion procedure to compute a canonical result whenever such a
presentation exists, without imposing any restriction on the reduction
order. Instead, the set of input
equalities $\EE_0$ is required to be linear, and Devie considers an
ordered completion inference system with a modified
deduction rule to ensure that linearity is preserved.
He moreover shows that under these circumstances a relaxed
fairness condition is sufficient.
In this section we give a new proof of this result which has been
formalized.
First we recall Devie's inference system.

\begin{defi}[Linear Ordered Completion
\isaforlink{Ordered_Completion}{ind:oKBilin}]%
\label{def:KBl}
The inference system \KBl of linear ordered completion consists
of the rules \tsfs{orient}, \tsfs{delete}, \tsfs{compose},
\tsfs{simplify}, and \tsfs{collapse}$_{\smprencompasses}$ of \KBi
(\defref{KBi})
together with the following modified deduction rule:
\begin{center}
\bigskip
\begin{tabular}{@{}l@{\qquad}c@{\qquad}l@{}}
\tsfs{deduce}$_{\m{l}}$ &
$\displaystyle \frac
{\EE,\RR}
{\EE \cup \{ s \approx t \},\RR}$ &
\end{tabular}
if $s \xleftarrow[\EE^{\pm} \scup \RR]{} \cdot
\xrightarrow[\EE^{\pm} \scup \RR]{} t$ and $s \approx t$ is linear
\end{center}
\end{defi}

We write $(\EE,\RR) \vdl (\EE',\RR')$ if $(\EE',\RR')$ can be reached
from $(\EE,\RR)$ by employing one of the inference rules of \defref{KBl}.

\begin{lem}%
\label{lem:KBl_KBo}
The inclusion ${\textup{\KBl}} \subseteq {\textup{\KBo}}$ holds.
\hfill
\isaforlink{Ordered_Completion}{lem:oKBilin_step_imp_oKBi_step}
\end{lem}

Note that in contrast to the ordered completion system \KBo,
ordered rewriting using orientable instances of $\EE$ is not permitted
in \tsfs{compose}, \tsfs{simplify}, and
\tsfs{collapse}$_{\smprencompasses}$.
This is because ordered rewrite steps need not preserve linearity
as stated in \lemref{oKBilin_linearity} below.
For example, a \tsfs{compose} step in \KBo on the linear rule
$\m{g}(x) \to \m{f}(\m{f}(x))$ using the linear equation
$\m{f}(x) \approx \m{f}(y)$ may result in the nonlinear rule
$\m{g}(x) \to \m{f}(\m{h}(x,x))$ when
$\m{h}(x,x)$ is substituted for the variable $y$ and
a reduction order $>$
is used such that $\m{f}(\m{f}(x)) > \m{f}(\m{h}(x,x))$.

With these restrictions, it is not hard to prove that inference
steps preserve linearity.

\begin{lem}%
\label{lem:oKBilin_linearity}
If $\EE \cup \RR$ is linear and
$(\EE,\RR) \vdl (\EE',\RR')$ then $\EE' \cup \RR'$ is linear.
\hfill
\isaforlink{Ordered_Completion}{lem:oKBilin_step_linearity_preserving}
\end{lem}

From now on we consider $\EE_0 \cup \RR_0$ to be linear.

\begin{defi}%
\label{def:KBl fairness}
An extended overlap (\defref{extended overlap}) which satisfies
$\ell_1 > r_1$ and $r_2 \not> \ell_2$,
or $\ell_2 > r_2$ and $r_1 \not> \ell_1$ gives rise to a
\emph{linear critical pair}~\cite{D91}. The set of all linear critical
pairs originating from equations in $\EE$ is denoted $\LCP_>(\EE)$.
An infinite run
\[
(\EE_0,\RR_0) ~\vdl~ (\EE_1,\RR_1) ~\vdl~ (\EE_2,\RR_2) ~\vdl~ \cdots
\]
is \emph{fair} if the inclusion
$\LCP_>(\ERw) \subseteq {\join_\RRw} \cup {\fromto_{\EEi}}$ holds.
\end{defi}

Below, we consider an infinite fair run $\Gamma$. We next show that the
result of a fair run without persistent equations is indeed complete.

\begin{thm}%
\label{thm:KBl correctness}
If $\Gamma$ is fair and $\EEw = \varnothing$ then $\RRw$ is a complete
presentation of $\EEw$.
\hfill
\isaforlink{Ordered_Completion}{lem:Ew_empty_CR_Rw_linear}
\end{thm}
\begin{proof}
The run $\Gamma$ is also a valid \KBo run by \lemref{KBl_KBo}. We moreover
have $\PCP(\RRw) \subseteq \LCP(\ERw)$ since
$\RRw \subseteq {>}$, and hence
$\PCP(\RRw) \subseteq {\join_\RRw} \cup {\fromto_{\EEi}}$
by fairness. So the result follows from
\thmref{complete presentation okb1}.
\end{proof}

The following result relates equations in $\EEi$ and rules in $\RRi$ to
persistent equations and rules, respectively.

\begin{lem}\hfill%
\label{lem:oKBilin_ERinf_to_ERw}
\begin{enumerate}
\item%
\label{lem:oKBilin_E}
$\EEi \:\subseteq\:
{\xrightarrow[\RRi]{*} \cdot \xlr[\EEw]{=} \cdot \xleftarrow[\RRi]{*}}$
\hfill\isaforlink{Ordered_Completion}{lem:Einf_to_Ew}
\item%
\label{lem:oKBilin_R}
If $\ell \to r \in \RRi$ then
$\ell \xrightarrow[\RRw]{}\cdot\mathrel{{(\xlr[\ERw]{<\ell})}^*} r$.
\hfill\isaforlink{Ordered_Completion}{lem:Rinf_Rw_msteps}
\end{enumerate}
\end{lem}
\begin{proof}\hfill
\begin{enumerate}
\item
For an equation $s \approx t \in \EEi$ we prove the desired inclusion
by induction on $\{ s, t \}$ with respect to $>_\mul$.
\smallskip
\item
By induction on $(\ell,r)$ with respect to $>_\lex$.
\qedhere
\end{enumerate}
\end{proof}

\noindent
In order to show that $\RRw$ is Church-Rosser modulo $\EEw$, we
need a result about joinability of critical peaks
modulo persistent equations.

\begin{lem}%
\label{lem:oKBilin_peak}
If there is an equation $\ell \approx r \in  \EEw^\pm \cup \RRw$
with $r \not >\ell$ that is involved in a peak
$s \xleftarrow[r\,\approx\,\ell]{} \cdot \xrightarrow[\RRw]{} t$
then $s \xrightarrow[\RRi]{*} \cdot \xlr[\EEw]{=} \cdot
\xleftarrow[\RRi]{*} t$.
\hfill
\isaforlink{Ordered_Completion}{lem:linear_peak_cases}
\end{lem}
\begin{proof}
If the two steps occur at parallel positions then they commute
and thus $s \to_{\RRw} \cdot \FromB{r\,\approx\,\ell} t$.
If the peak constitutes an overlap then $s \xlr{}_{\LCP(\ERw)} t$
since $\RRw \subseteq {>}$ and $r \not> \ell$ by assumption. We thus have
$s \fromto_{\EEi} t$ or $s \join_{\RRw} t$
by fairness such that the claim follows from
\lemref{oKBilin_ERinf_to_ERw}(\ref{lem:oKBilin_E}) 
and $\RRw \subseteq \RRi$.
Otherwise, we have a variable overlap. By \lemref{oKBilin_linearity}
both $\EEw$ and $\RRw$ are linear. This implies
$s \to_{\RRw}^= \cdot \From{r\,\approx\,\ell}{=} t$, so the claim follows
from the inclusion $\RRw \subseteq \RRi$.
\end{proof}

\begin{lem}%
\label{lem:CRM}
The TRS $\RRw$ is Church-Rosser modulo $\EEw$.
\hfill\isaforlink{Ordered_Completion}{lem:CRm}
\end{lem}
\begin{proof}
Define the ARSs $\AA$ and $\BB$ with multiset labeling as follows:
\begin{itemize}
\item
$s \xrightarrow{M}_\AA t$ if $s \xrightarrow{\{ s' \}}_\RRw t$ and
$M = \{ s' \}$ for some term $s' \geqslant s$.
\item
$s \xrightarrow{M}_\BB t$ if $s \xlr{\{ s',\,t' \}}_{\EEw} t$ and
$M = \{ s', t' \}$ for some terms $s' \geqslant s$ and $t' \geqslant t$.
\end{itemize}
By equipping them with the well-founded order $>_\mul$
Lemmata~\ref{lem:oKBilin_peak} and~\ref{lem:oKBilin_ERinf_to_ERw}
imply the condition of peak decreasingness modulo. Hence,
\lemref{pdm => crm} applies.
\end{proof}

For a run of \KBl we call $(\EEw,\RRw)$ \emph{simplified} if $\RRw$
is reduced and $\EEw$ is irreducible with respect to $\RRw$ and does not
contain trivial equations. From now on we assume that $(\EEw,\RRw)$ is
simplified.
This allows us to establish relationships between $\RR$-normal forms
and normal forms with respect to the result of the linear completion run.

\begin{lem}%
\label{lem:Ew}
The inclusion $\NF(\RR) \subseteq \NF(\EEw^{\pm}) \cap \NF(\RRw)$ holds.
\hfill
\isaforlink{Ordered_Completion}{lem:NF_R_subset_NF_REw}
\end{lem}
\begin{proof}
Let $t \in \NF(\RR)$. Assume to the contrary that
$t \to u$ for some term $u$ by applying an equation
$\ell \approx r \in \EEw^\pm \cup \RRw$ from left to right.
Because $\RR$ is a complete presentation of $\ERw$, we have
$\ell \join_\RR r$. Since $t \in \NF(\RR)$ implies $\ell \in \NF(\RR)$,
we obtain $r \to_\RR^* \ell$.
If $\ell \approx r \in \RRw$ this contradicts $\RRw \subseteq {>}$,
otherwise $\ell \approx r \in \EEw^\pm$ and $r \to_\RR^* \ell$ contradict
unorientability and non-triviality of $\EEw$, which hold by the
assumption that $\EEw$ is simplified.
\end{proof}

\begin{lem}%
\label{lem:Rw and R}
The inclusion $\NF(\RRw) \subseteq \NF(\RR)$ holds.
\hfill
\isaforlink{Ordered_Completion}{lem:NF_Rw_subset_NF_R}
\end{lem}
\begin{proof}
We show that $\ell \to_{\RRw}^+ r$ for every $\ell \to r \in \RR$, which
is sufficient to prove the claim. Let $\ell \to r \in \RR$.
By \lemref{CRM} we have
\[
\ell \xrightarrow[\RRw]{*} u \xlr[\EEw]{}^* v \xleftarrow[\RRw]{*} r
\]
for some terms $u$ and $v$. Since $\RR$ is reduced, $r \in \NF(\RR)$.
According to \lemref{Ew},
both $r \in \NF(\EEw^\pm)$ and $r \in \NF(\RRw)$ hold.
Hence $r = v$ follows from $r \to_\RRw^* v$
and $u \fromto_{\EEw}^* v = r$ implies $u = v$.
Therefore $\ell \to_\RRw^* r$. Since
$\RR$ is terminating, $\ell = r$ is impossible and thus
$\ell \to_\RRw^+ r$ as desired.
\end{proof}

As in the previous section, the last result allows us to establish
the main completeness theorem.

\begin{thm}%
\label{thm:KBl completeness}
If $(\EEw,\RRw)$ is simplified then
$\EEw = \varnothing$ and $\RRw$ is literally similar to $\RR$.
\isaforlink{Ordered_Completion}{lem:linear_completeness}
\end{thm}
\begin{proof}
The TRS $\RR$ is complete, the TRS $\RRw$ is terminating, and the
inclusion ${\to_\RRw} \subseteq {\xlr{}_\RR^*}$ holds because $\RR$ is a
complete presentation. Moreover,
$\NF(\RRw) \subseteq \NF(\RR)$ by \lemref{Rw and R}. Hence,
\lemref[b]{reduced abstract} applies.
Since $\RRw$ is a complete presentation, $\EEw = \varnothing$
by the assumption of a simplified system.
\end{proof}

\begin{exa}
By \thmref{KBl completeness} any simplifying \KBl run on the
equational system $\EE$ and the reduction order $\LPO$ from
\exaref{strategy} will result in a canonical presentation, independent of
the order in which inference steps are applied.
Note that \thmref{KBo completeness} does not apply since the given order
$\LPO$ is not ground total.
\end{exa}

We conclude the subsection by showing the absence of a complete
presentation for the equational system mentioned in the first
paragraph of \secref{ordered completion}.

\begin{exa}
Let $\EE$ be the ES consisting of the two equations $\m{0} + x \approx x$
and $x + y \approx y + x$. We show that $\EE$ admits no complete
presentation. Assume to the contrary that $\RR$ is a complete
presentation of $\EE$. We use
$\to_\RR^+$ as the reduction order $>$ for \KBl. Because
$\m{0} + x \xlr{}_\EE^* x$ implies $\m{0} + x \join_\RR x$
and $x \in \NF(\RR)$, we have $\m{0} + x > x$. In the same way
$x + \m{0} > x$ is derived. Therefore, the following fair run of \KBl
is constructed:
\begin{alignat*}{2}
(\EE,\varnothing) ~
& \vdlr{orient} &~&
(\{ x + y \approx y + x \},
 \{ \m{0} + x \to x \}) \\
& \vdlr{deduce} &&
(\{ x + y \approx y + x, x + \m{0} \approx x \},
 \{ \m{0} + x \to x \}) \\
& \vdlr{orient} &&
(\{ x + y \approx y + x \},
 \{ \m{0} + x \to x, x + \m{0} \to x \}) \\
& \vdlr{deduce} &&
(\{ x + y \approx y + x, \m{0} \approx \m{0} \},
 \{ \m{0} + x \to x, x + \m{0} \to x \}) \\
& \vdlr{delete} &&
(\{ x + y \approx y + x \},
 \{ \m{0} + x \to x, x + \m{0} \to x \}) \\
& \vdlr{deduce} &&
(\{ x + y \approx y + x, \m{0} \approx \m{0} \},
 \{ \m{0} + x \to x, x + \m{0} \to x \}) \\
& \vdlr{delete} && \cdots
\end{alignat*}
It is easy to check that this run is fair; the
two non-trivial critical pairs $x + \m{0} \approx x$ and
$\m{0} + x \approx x$ belong to $\EEi^\pm$.
We have $\EEw = \{ x + y \approx y + x \}$ and
$\RRw = \{ \m{0} + x \to x, x + \m{0} \to x \}$. Note that
$(\EEw,\RRw)$ is simplified.
According to \thmref{KBl completeness},
the persistent set $\EEw$ must be empty. This is a contradiction and
thus $\RR$ does not exist.
\end{exa}

In summary, our proof of \thmref{complete presentation okb1} resembles
the approach by Devie~\cite{D91}, though our version
is structured along several preliminary results
that are of independent interest, such as Lemmata~\ref{lem:CRM},~\ref{lem:Ew}, and~\ref{lem:Rw and R}.

\section{Conclusion}%
\label{sec:conclusion}

In this paper we have presented new and formalized proofs for a number
of correctness and completeness results for abstract completion, ranging
from the decidable case of ground-completion to completeness results for
ordered completion. By using modern abstract confluence criteria,
we could avoid the use of proof orders, which had a positive effect
on the Isabelle/HOL formalization.

We mention some topics for future work.
Concerning completion of ground systems, the literature contains other
interesting results that we might consider as target for future
formalization efforts.
Gallier \etal~\cite{GNPRS93} showed that every ground ES $\EE$
can be transformed into an equivalent canonical TRS in $O(n^3)$ time,
where $n$ is the combined size of the terms appearing in $\EE$.
Snyder~\cite{S93} improved this result to an $O(n \log n)$ time
algorithm. Moreover, his algorithm can enumerate all canonical
presentations, of which there are at most $2^k$~\cite[Theorem~4.7]{S93},
where $k$ is the number of equations in $\EE$. Furthermore, all canonical
presentations have the same number of rules.

In the context of ordered completion, completeness remains an open
problem in the general case: It is unknown whether an ordered completion
run can find a complete system $\RR$ for a set of input equations $\EE$
if neither $\EE$ is linear (\thmref{KBl completeness})
nor $\RR$ is compatible with a ground-total
reduction order (\thmref{KBo completeness}).

There are several important extensions of completion that we did not
consider in this paper. We mention completion in the presence of
associative and commutative (AC) symbols~\cite{PS81}, normalized
completion~\cite{M96,WM13}, as well as maximal completion~\cite{KH11}.
They are natural candidates for future formalization efforts.

\section*{Acknowledgement}
We are grateful for the detailed comments by the anonymous reviewers,
which greatly helped us to improve the paper.

\bibliographystyle{alpha}
\bibliography{submission}

\newcommand{\etalchar}[1]{$^{#1}$}
\newcommand{\noop}[1]{}
\begin{thebibliography}{WSMK10}

\bibitem[Bac91]{B91}
Leo Bachmair.
\newblock {\em Canonical Equational Proofs}.
\newblock Birkh{\"a}user, 1991.
\newblock \doi{10.1007/978-1-4684-7118-2}.

\bibitem[BD86]{BD86}
Leo Bachmair and Nachum Dershowitz.
\newblock Commutation, transformation, and termination.
\newblock In {\em Proc.\ 8th International Conference on Automated Deduction},
  volume 230 of {\em Lecture Notes in Computer Science}, pages 5--20, 1986.
\newblock \doi{10.1007/3-540-16780-3_76}.

\bibitem[BD94]{BD94}
Leo Bachmair and Nachum Dershowitz.
\newblock Equational inference, canonical proofs, and proof orderings.
\newblock {\em Journal of the ACM}, 41(2):236--276, 1994.
\newblock \doi{10.1145/174652.174655}.

\bibitem[BDH86]{BDH86}
Leo Bachmair, Nachum Dershowitz, and Jieh Hsiang.
\newblock Orderings for equational proofs.
\newblock In {\em Proc.\ 1st IEEE Symposium on Logic in Computer Science},
  pages 346--357, 1986.

\bibitem[BDP89]{BDP89}
Leo Bachmair, Nachum Dershowitz, and David~A. Plaisted.
\newblock Completion without failure.
\newblock In H.~A{\"{\i}t} Kaci and M.~Nivat, editors, {\em Resolution of
  Equations in Algebraic Structures}, volume 2: Rewriting Techniques of {\em
  Progress in Theoretical Computer Science}, pages 1--30. Academic Press, 1989.

\bibitem[BN98]{BN98}
Franz Baader and Tobias Nipkow.
\newblock {\em Term Rewriting and All That}.
\newblock Cambridge University Press, 1998.
\newblock \doi{10.1017/CBO9781139172752}.

\bibitem[Bur01]{B01}
Serge Burckel.
\newblock Syntactical methods for braids of three strands.
\newblock {\em Journal of Symbolic Comp.}, 31:557--564, 2001.
\newblock \doi{10.1006/jsco.2000.0473}.

\bibitem[CCF{\etalchar{+}}11]{CCFPU11}
{\'E}velyn Contejean, Pierre Courtieu, Julien Forest, Olivier Pons, and Xavier
  Urbain.
\newblock Automated certified proofs with {CiME3}.
\newblock In {\em Proc.\ 22nd International Conference on Rewriting Techniques
  and Applications}, volume~10 of {\em Leibniz Intl. Proceedings in
  Informatics}, pages 21--30, 2011.
\newblock \doi{10.4230/LIPIcs.RTA.2011.21}.

\bibitem[Dev91]{D91}
Herv{\'e} Devie.
\newblock Linear completion.
\newblock In {\em Proc.\ 2nd International Workshop on Conditional and Typed
  Rewriting Systems}, pages 233--245. Springer, 1991.
\newblock \doi{10.1007/3-540-54317-1_94}.

\bibitem[GNP{\etalchar{+}}93]{GNPRS93}
Jean~H. Gallier, Paliath Narendran, David~A. Plaisted, Stan Raatz, and Wayne
  Snyder.
\newblock An algorithm for finding canonical sets of ground rewrite rules in
  polynomial time.
\newblock {\em Journal of the ACM}, 40(1):1--16, 1993.
\newblock \doi{10.1145/138027.138032}.

\bibitem[Gra01]{G01}
Bernhard Gramlich.
\newblock On interreduction of semi-complete term rewriting systems.
\newblock {\em Theoretical Computer Science}, 258(1-2):435--451, 2001.
\newblock \doi{10.1016/S0304-3975(00)00030-X}.

\bibitem[HMS14]{HMS14}
Nao Hirokawa, Aart Middeldorp, and Christian Sternagel.
\newblock A new and formalized proof of abstract completion.
\newblock In {\em Proc.\ 5th International Conference on Interactive Theorem
  Proving}, volume 8558 of {\em Lecture Notes in Computer Science}, pages
  292--307, 2014.
\newblock \doi{10.1007/978-3-319-08970-6_19}.

\bibitem[HMSW17]{HMSW17}
Nao Hirokawa, Aart Middeldorp, Christian Sternagel, and Sarah Winkler.
\newblock Infinite runs in abstract completion.
\newblock In {\em Proc.\ 2nd International Conference on Formal Structures for
  Computation and Deduction}, volume~84 of {\em Leibniz Intl. Proceedings in
  Informatics}, pages 19:1--19:16, 2017.
\newblock \doi{10.4230/LIPIcs.FSCD.2017.19}.

\bibitem[Hue80]{H80}
G{\'e}rard~P. Huet.
\newblock Confluent reductions: Abstract properties and applications to term
  rewriting systems.
\newblock {\em Journal of the ACM}, 27(4):797--821, 1980.
\newblock \doi{10.1145/322217.322230}.

\bibitem[Hue81]{H81}
G{\'e}rard~P. Huet.
\newblock A complete proof of correctness of the {K}nuth-{B}endix completion
  algorithm.
\newblock {\em Journal of Computer and System Sciences}, 23(1):11--21, 1981.
\newblock \doi{10.1016/0022-0000(81)90002-7}.

\bibitem[KB70]{KB70}
Donald~E. Knuth and Peter~B. Bendix.
\newblock Simple word problems in universal algebras.
\newblock In J.~Leech, editor, {\em Computational Problems in Abstract
  Algebra}, pages 263--297. Pergamon Press, 1970.

\bibitem[KH11]{KH11}
Dominik Klein and Nao Hirokawa.
\newblock Maximal completion.
\newblock In {\em Proc.\ 22nd International Conference on Rewriting Techniques
  and Applications}, volume~10 of {\em Leibniz Intl. Proceedings in
  Informatics}, pages 71--80. Schloss Dagstuhl -- Leibniz-Zentrum f{\"u}r
  Informatik, 2011.
\newblock \doi{10.4230/LIPIcs.RTA.2011.71}.

\bibitem[KL80]{KL80}
Sam Kamin and Jean-Jacques L{\'e}vy.
\newblock Two generalizations of the recursive path ordering.
\newblock Unpublished manuscript, University of Illinois, 1980.

\bibitem[Klo80]{K80}
Jan~Willem Klop.
\newblock {\em Combinatory Reduction Systems}.
\newblock PhD thesis, Utrecht University, 1980.

\bibitem[KMN88]{KMN88}
Deepak Kapur, David~R. Musser, and Paliath Narendran.
\newblock Only prime superpositions need be considered in the {K}nuth-{B}endix
  completion procedure.
\newblock {\em Journal of Symbolic Comp.}, 6(1):19--36, 1988.
\newblock \doi{10.1016/S0747-7171(88)80019-1}.

\bibitem[KN85]{KN85}
Deepak Kapur and Paliath Narendran.
\newblock A finite {T}hue system with decidable word problem and without
  equivalent finite canonical system.
\newblock {\em Theoretical Computer Science}, 35:337--344, 1985.
\newblock \doi{10.1016/0304-3975(85)90023-4}.

\bibitem[Mar96]{M96}
Claude March{\'e}.
\newblock Normalized rewriting: An alternative to rewriting modulo a set of
  equations.
\newblock {\em Journal of Symbolic Comp.}, 21(3):253--288, 1996.
\newblock \doi{10.1006/jsco.1996.0011}.

\bibitem[M{\'e}t83]{M83}
Yves M{\'e}tivier.
\newblock About the rewriting systems produced by the {K}nuth-{B}endix
  completion algorithm.
\newblock {\em Information Processing Letters}, 16(1):31--34, 1983.
\newblock \doi{10.1016/0020-0190(83)90009-1}.

\bibitem[vO94]{vO94}
Vincent van Oostrom.
\newblock Confluence by decreasing diagrams.
\newblock {\em Theoretical Computer Science}, 126(2):259--280, 1994.
\newblock \doi{10.1016/0304-3975(92)00023-K}.

\bibitem[vO07]{vO07}
Vincent van Oostrom.
\newblock Random descent.
\newblock In {\em Proc.\ 18th International Conference on Rewriting Techniques
  and Applications}, volume 4533 of {\em Lecture Notes in Computer Science},
  pages 314--328, 2007.
\newblock \doi{10.1007/978-3-540-73449-9_24}.

\bibitem[vO08]{vO08}
Vincent van Oostrom.
\newblock Confluence by decreasing diagrams -- converted.
\newblock In {\em Proc.\ 19th International Conference on Rewriting Techniques
  and Applications}, volume 5117 of {\em Lecture Notes in Computer Science},
  pages 306--320, 2008.
\newblock \doi{10.1007/978-3-540-70590-1_21}.

\bibitem[vOT16]{vOT16}
Vincent \noop{O}van Oostrom and Yoshihito Toyama.
\newblock Normalisation by random descent.
\newblock In {\em Proc.\ 1st International Conference on Formal Structures for
  Computation and Deduction}, volume~52 of {\em Leibniz Intl. Proceedings in
  Informatics}, pages 32:1--32:18, 2016.
\newblock \doi{10.4230/LIPIcs.FSCD.2016.32}.

\bibitem[PS81]{PS81}
Gerald~E. Peterson and Mark~E. Stickel.
\newblock Complete sets of reductions for some equational theories.
\newblock {\em Journal of the ACM}, 28(2):233--264, 1981.
\newblock \doi{10.1145/322248.322251}.

\bibitem[PSK96]{PSK96}
David Plaisted and Andrea Sattler-Klein.
\newblock Proof lengths for equational completion.
\newblock {\em Information and Computation}, 125(2):154--170, 1996.
\newblock \doi{10.1006/inco.1996.0028}.

\bibitem[Sny93]{S93}
Wayne Snyder.
\newblock A fast algorithm for generating reduced ground rewriting systems from
  a set of ground equations.
\newblock {\em Journal of Symbolic Comp.}, 15(4):415--450, 1993.
\newblock \doi{10.1006/jsco.1993.1029}.

\bibitem[ST13]{ST13}
Christian Sternagel and Ren{\'e} Thiemann.
\newblock Formalizing {K}nuth-{B}endix orders and {K}nuth-{B}endix completion.
\newblock In {\em Proc.\ 24th International Conference on Rewriting Techniques
  and Applications}, volume~21 of {\em Leibniz Intl. Proceedings in
  Informatics}, pages 287--302. Schloss Dagstuhl -- Leibniz-Zentrum f{\"u}r
  Informatik, 2013.
\newblock \doi{10.4230/LIPIcs.RTA.2013.287}.

\bibitem[SZ12]{kbcv}
Thomas Sternagel and Harald Zankl.
\newblock {\kbcv} -- {Knuth-Bendix} completion visualizer.
\newblock In {\em Proc.\ 6th International Joint Conference on Automated
  Reasoning}, volume 7364 of {\em Lecture Notes in Computer Science}, pages
  530--536. Springer, 2012.
\newblock \doi{10.1007/978-3-642-31365-3_41}.

\bibitem[Ter03]{Terese}
Terese.
\newblock {\em Term Rewriting Systems}, volume~55 of {\em Cambridge Tracts in
  Theoretical Computer Science}.
\newblock Cambridge University Press, 2003.

\bibitem[Toy91]{T91}
Yoshihito Toyama.
\newblock How to prove equivalence of term rewriting systems without induction.
\newblock {\em Theoretical Computer Science}, 90(2):369--390, 1991.

\bibitem[TS09]{TS09}
Ren{\'e} Thiemann and Christian Sternagel.
\newblock Certification of termination proofs using {CeTA}.
\newblock In {\em Proc.\ 22nd International Conference on Theorem Proving in
  Higher Order Logics}, volume 5674 of {\em Lecture Notes in Computer Science},
  pages 452--468, 2009.
\newblock \doi{10.1007/978-3-642-03359-9_31}.

\bibitem[WB86]{WB86}
Franz Winkler and Bruno Buchberger.
\newblock A criterion for eliminating unnecessary reductions in the
  {K}nuth-{B}endix algorithm.
\newblock In {\em Proc.\ Colloquium on Algebra, Combinatorics and Logic in
  Computer Science, Vol.\ II}, volume~42 of {\em Colloquia Mathematica
  Societatis J.~Bolyai}, pages 849--869, 1986.

\bibitem[WM13]{WM13}
Sarah Winkler and Aart Middeldorp.
\newblock Normalized completion revisited.
\newblock In {\em Proc.\ 24th International Conference on Rewriting Techniques
  and Applications}, volume~21 of {\em Leibniz Intl. Proceedings in
  Informatics}, pages 318--333. Schloss Dagstuhl -- Leibniz-Zentrum f{\"u}r
  Informatik, 2013.
\newblock \doi{10.4230/LIPIcs.RTA.2013.318}.

\bibitem[WSMK10]{mkbtt}
Sarah Winkler, Haruhiko Sato, Aart Middeldorp, and Masahito Kurihara.
\newblock Optimizing {\mkbtt}.
\newblock In {\em Proc.\ 21st International Conference on Rewriting Techniques
  and Applications}, volume~6 of {\em Leibniz Intl. Proceedings in
  Informatics}, pages 373--384. Schloss Dagstuhl -- Leibniz-Zentrum f{\"u}r
  Informatik, 2010.
\newblock \doi{10.4230/LIPIcs.RTA.2010.373}.

\bibitem[WSW06]{slothrop}
Ian Wehrman, Aaron Stump, and E.~Westbrook.
\newblock {\slothrop}: {Knuth-Bendix} completion with a modern termination
  checker.
\newblock In {\em Proc.\ 17th International Conference on Rewriting Techniques
  and Applications}, volume 4098 of {\em Lecture Notes in Computer Science},
  pages 287--296. Springer, 2006.
\newblock \doi{10.1007/11805618_22}.

\end{thebibliography}

\end{document}